\documentclass[acmsmall]{acmart}
\usepackage[utf8]{inputenc}
\usepackage[T1]{fontenc}

\newcommand{\packageGraphicx}{\usepackage{graphicx}}
\newcommand{\packageHyperref}{\usepackage{hyperref}}
\newcommand{\renewrmdefault}{\renewcommand{\rmdefault}{ptm}}
\newcommand{\packageRelsize}{\usepackage{relsize}}
\newcommand{\packageAmsmath}{\usepackage{amsmath}}
\newcommand{\packageMathabx}{\usepackage{mathabx}}
\newcommand{\packageWasysym}{
  \let\leftmoon\relax \let\rightmoon\relax \let\fullmoon\relax \let\newmoon\relax \let\diameter\relax
  \usepackage[nointegrals]{wasysym}}
\newcommand{\packageTxfonts}{
  \let\widering\relax
  \let\oldwidebar\widebar
  \let\widebar\relax
  \usepackage{newtxmath}
  \ifx\widebar\relax
    \let\widebar\oldwidebar
  \fi
}
\newcommand{\packageTextcomp}{\usepackage{textcomp}}
\newcommand{\packageFramed}{\usepackage{framed}}
\newcommand{\packageHyphenat}{\usepackage[htt]{hyphenat}}
\newcommand{\packageColor}{\usepackage[usenames,dvipsnames]{color}}
\newcommand{\doHypersetup}{\hypersetup{bookmarks=true,bookmarksopen=true,bookmarksnumbered=true}}
\newcommand{\packageTocstyle}{}
\newcommand{\packageCJK}{\IfFileExists{CJK.sty}{\usepackage{CJK}}{}}
\renewcommand\packageColor\relax
\renewcommand\packageTocstyle\relax
\renewcommand\packageMathabx{\ifx\bigtimes\undefined \usepackage{mathabx} \else \relax \fi}
\renewcommand\packageTxfonts\relax
\let\Footnote\undefined

\renewcommand{\renewrmdefault}{}
\packageGraphicx
\packageHyperref
\renewrmdefault
\packageRelsize
\packageAmsmath
\packageMathabx
\packageWasysym
\packageTxfonts
\packageTextcomp
\packageFramed
\packageHyphenat
\packageColor
\doHypersetup
\packageTocstyle
\packageCJK

\newcommand{\sectionNewpage}{}
\newcommand{\preDoc}{}
\newcommand{\postDoc}{}

\newcommand{\BookRef}[2]{\emph{#2}}
\newcommand{\ChapRef}[2]{\SecRef{#1}{#2}}
\newcommand{\SecRef}[2]{section~#1}
\newcommand{\PartRef}[2]{part~#1}
\newcommand{\BookRefUC}[2]{\BookRef{#1}{#2}}
\newcommand{\ChapRefUC}[2]{\SecRefUC{#1}{#2}}
\newcommand{\SecRefUC}[2]{Section~#1}
\newcommand{\PartRefUC}[2]{Part~#1}

\newcommand{\BookRefLocal}[3]{\hyperref[#1]{\BookRef{#2}{#3}}}
\newcommand{\ChapRefLocal}[3]{\hyperref[#1]{\ChapRef{#2}{#3}}}
\newcommand{\SecRefLocal}[3]{\hyperref[#1]{\SecRef{#2}{#3}}}
\newcommand{\PartRefLocal}[3]{\hyperref[#1]{\PartRef{#2}{#3}}}
\newcommand{\BookRefLocalUC}[3]{\hyperref[#1]{\BookRefUC{#2}{#3}}}
\newcommand{\ChapRefLocalUC}[3]{\hyperref[#1]{\ChapRefUC{#2}{#3}}}
\newcommand{\SecRefLocalUC}[3]{\hyperref[#1]{\SecRefUC{#2}{#3}}}
\newcommand{\PartRefLocalUC}[3]{\hyperref[#1]{\PartRefUC{#2}{#3}}}

\newcommand{\BookRefUN}[1]{\BookRef{}{#1}}

\newcommand{\SecRefUN}[1]{``#1''}

\newcommand{\BookRefLocalUN}[2]{\hyperref[#1]{\BookRefUN{#2}}}

\newcommand{\SecRefLocalUN}[2]{\hyperref[#1]{\SecRefUN{#2}}}

\newcommand{\SectionNumberLink}[2]{\hyperref[#1]{#2}}

\newcommand{\Scribtexttt}[1]{{\texttt{#1}}}

\newcommand{\textsuper}[1]{$^{\hbox{\textsmaller{#1}}}$}

\newcommand{\planetName}[1]{PLane\hspace{-0.1ex}T}

\newcommand{\Stttextless}{{\fontencoding{T1}\selectfont<}}

\makeatletter
\def\empty@finalstrut#1{%
  \unskip\ifhmode\nobreak\fi\vrule\@width\z@\@height\z@\@depth\z@}
\def\no@strut{\global\setbox\@arstrutbox\hbox{%
    \vrule \@height\z@
           \@depth\z@
           \@width\z@}%
    \gdef\@endpbox{\empty@finalstrut\@arstrutbox\par\egroup\hfil}%
}%
\def\yes@strut{\global\setbox\@arstrutbox\hbox{%
    \vrule \@height\arraystretch \ht\strutbox
           \@depth\arraystretch \dp\strutbox
           \@width\z@}%
    \gdef\@endpbox{\@finalstrut\@arstrutbox\par\egroup\hfil}%
}%
\def\@mkpream#1{\@firstamptrue\@lastchclass6
  \let\@preamble\@empty\def\empty@preamble{\add@ins}%
  \let\protect\@unexpandable@protect
  \let\@sharp\relax\let\add@ins\relax
  \let\@startpbox\relax\let\@endpbox\relax
  \@expast{#1}%
  \expandafter\@tfor \expandafter
    \@nextchar \expandafter:\expandafter=\reserved@a\do
       {\@testpach\@nextchar
    \ifcase \@chclass \@classz \or \@classi \or \@classii \or \@classiii
      \or \@classiv \or\@classv \fi\@lastchclass\@chclass}%
  \ifcase \@lastchclass \@acol
      \or \or \@preamerr \@ne\or \@preamerr \tw@\or \or \@acol \fi}
\def\@addamp{%
  \if@firstamp
    \@firstampfalse
    \edef\empty@preamble{\add@ins}%
  \else
    \edef\@preamble{\@preamble &}%
    \edef\empty@preamble{\expandafter\noexpand\empty@preamble &\add@ins}%
  \fi}
\newif\iftw@hlines \tw@hlinesfalse
\def\@xhline{\ifx\reserved@a\hline
               \tw@hlinestrue
             \else\ifx\reserved@a\Hline
               \tw@hlinestrue
             \else
               \tw@hlinesfalse
             \fi\fi
      \iftw@hlines
        \aftergroup\do@after
      \fi
      \ifnum0=`{\fi}%
}
\def\do@after{\emptyrow[\the\doublerulesep]}
\def\emptyrow{\noalign\bgroup\@ifnextchar[\@emptyrow{\@emptyrow[\z@]}}
\def\@emptyrow[#1]{\no@strut\gdef\add@ins{\vrule \@height\z@ \@depth#1 \@width\z@}\egroup%
\empty@preamble\\
\noalign{\yes@strut\gdef\add@ins{\vrule \@height\z@ \@depth\z@ \@width\z@}}%
}
\def\tabrow#1{\noalign\bgroup\@ifnextchar[{\@tabrow{#1}}{\@tabrow{#1}[]}}
\def\@tabrow#1[#2]{\no@strut\egroup#1\ifx.#2.\\\else\\[#2]\fi\noalign{\yes@strut}}
\def\endpltstabular{\crcr\egroup\egroup \egroup}
\expandafter \let \csname endpltstabular*\endcsname = \endpltstabular
\def\pltstabular{\let\@halignto\@empty\@pltstabular}
\@namedef{pltstabular*}#1{\def\@halignto{to#1}\@pltstabular}
\def\@pltstabular{\leavevmode \bgroup \let\@acol\@tabacol
   \let\@classz\@tabclassz
   \let\@classiv\@tabclassiv \let\\\@tabularcr\@stabarray}
\def\@stabarray{\m@th\@ifnextchar[\@sarray{\@sarray[c]}}
\def\@sarray[#1]#2{%
  \bgroup
  \setbox\@arstrutbox\hbox{%
    \vrule \@height\arraystretch\ht\strutbox
           \@depth\arraystretch \dp\strutbox
           \@width\z@}%
  \@mkpream{#2}%
  \edef\@preamble{%
    \ialign \noexpand\@halignto
      \bgroup \@arstrut \@preamble \tabskip\z@skip \cr}%
  \let\@startpbox\@@startpbox \let\@endpbox\@@endpbox
  \let\tabularnewline\\%
    \let\@sharp##%
    \set@typeset@protect
    \lineskip\z@skip\baselineskip\z@skip
    \@preamble}

\makeatother

\newenvironment{bigtabular}{\begin{pltstabular}}{\end{pltstabular}}
\newlength{\stabLeft}
\newcommand{\bigtableleftpad}{\hspace{\stabLeft}}
\newcommand{\atItemizeStart}[0]{\addtolength{\stabLeft}{\labelsep}
                                \addtolength{\stabLeft}{\labelwidth}}

\newenvironment{SingleColumn}{\begin{list}{}{\topsep=0pt\partopsep=0pt%
\listparindent=0pt\itemindent=0pt\labelwidth=0pt\leftmargin=0pt\rightmargin=0pt%
\itemsep=0pt\parsep=0pt}\item}{\end{list}}

\newcommand{\Sendabbrev}[1]{#1\@}

\newcommand{\SCodePreSkip}{\vskip\abovedisplayskip}
\newcommand{\SCodePostSkip}{\vskip\belowdisplayskip}
\newenvironment{SCodeFlow}{\SCodePreSkip\begin{list}{}{\topsep=0pt\partopsep=0pt%
\listparindent=0pt\itemindent=0pt\labelwidth=0pt\leftmargin=2ex\rightmargin=2ex%
\itemsep=0pt\parsep=0pt}\item}{\end{list}\SCodePostSkip}
\newcommand{\SCodeInsetBox}[1]{\setbox1=\hbox{\hbox{\hspace{2ex}#1\hspace{2ex}}}\vbox{\SCodePreSkip\vtop{\box1\SCodePostSkip}}}

\newcommand{\SVInsetPreSkip}{\vskip\abovedisplayskip}
\newcommand{\SVInsetPostSkip}{\vskip\belowdisplayskip}

\newcommand{\titleAndVersionAndAuthors}[3]{\title{#1\\{\normalsize \SVersionBefore{}#2}}\author{#3}\maketitle}

\newcommand{\titleAndEmptyVersionAndAuthors}[3]{\title{#1}\author{#3}\maketitle}

\newcommand{\SAuthor}[1]{#1}
\newcommand{\SAuthorSep}[1]{\qquad}
\newcommand{\SVersionBefore}[1]{Version }

\newcommand{\SNumberOfAuthors}[1]{}

\let\SOriginalthesubsection\thesubsection
\let\SOriginalthesubsubsection\thesubsubsection

\newcommand{\Ssection}[2]{\section[#1]{#2}\let\thesubsection\SOriginalthesubsection}
\newcommand{\Ssubsection}[2]{\subsection[#1]{#2}\let\thesubsubsection\SOriginalthesubsubsection}

\newcommand{\Ssectionstar}[1]{\section*{#1}\renewcommand*\thesubsection{\arabic{subsection}}\setcounter{subsection}{0}}

\newcommand{\Ssectionstarx}[2]{\Ssectionstar{#2}\phantomsection\addcontentsline{toc}{section}{#1}}

\newcounter{GrouperTemp}

\newenvironment{SVerbatim}{}{}

\newcommand{\Snolinkurl}[1]{\nolinkurl{#1}}

\newcommand{\SAuthorinfo}[4]{#1}
\newcommand{\SAuthorPlace}[1]{#1}
\newcommand{\SAuthorEmail}[1]{#1}

\newcommand{\SConferenceInfo}[2]{}
\newcommand{\SCopyrightYear}[1]{}
\newcommand{\SCopyrightData}[1]{}
\newcommand{\Sdoi}[1]{}

\newcommand{\SCategory}[3]{}
\newcommand{\SCategoryPlus}[4]{}
\newcommand{\STerms}[1]{}
\newcommand{\SKeywords}[1]{}

\newcommand{\AutobibLink}[1]{\color{ACMPurple}{#1}}
\newenvironment{AutoBibliography}{\begin{small}}{\end{small}}
\newcommand{\Autobibentry}[1]{\hspace{0.05\linewidth}\parbox[t]{0.95\linewidth}{\parindent=-0.05\linewidth#1\vspace{1.0ex}}}
\newcommand{\Autobibtarget}[1]{\phantomsection#1}

\usepackage{calc}
\newlength{\ABcollength}

\newcommand{\Autobibref}[1]{#1}

\providecommand{\AutobibLink}[1]{#1}

\usepackage{ccaption}

\makeatletter
\@ifundefined{belowcaptionskip}{}{}
\makeatother

\newcommand{\Legend}[1]{~

                        \hrule width \hsize height .33pt
                        \vspace{4pt}
                        \legend{#1}}

\newcommand{\FigureTarget}[2]{#1}
\newcommand{\FigureRef}[2]{#1}

\newlength{\FigOrigskip}
\FigOrigskip=\parskip

\newcommand{\FigureSetRef}{\refstepcounter{figure}}

\newenvironment{Figure}{\begin{figure}\FigureSetRef}{\end{figure}}
\newenvironment{FigureMulti}{\begin{figure*}[t!p]\FigureSetRef}{\end{figure*}}

\newenvironment{Centerfigure}{\begin{Xfigure}\centering\item}{\end{Xfigure}}

\newenvironment{Xfigure}{\begin{list}{}{\leftmargin=0pt\topsep=0pt\parsep=\FigOrigskip\partopsep=0pt}}{\end{list}}

\newenvironment{FigureInside}{}{}

\newcommand{\Centertext}[1]{\begin{center}#1\end{center}}

\newcommand{\SColorize}[2]{\color{#1}{#2}}

\newcommand{\SHyphen}[1]{#1}

\newcommand{\inColor}[2]{{\SHyphen{\Scribtexttt{\SColorize{#1}{#2}}}}}
\definecolor{PaleBlue}{rgb}{0.90,0.90,1.0}
\definecolor{LightGray}{rgb}{0.90,0.90,0.90}
\definecolor{CommentColor}{rgb}{0.76,0.45,0.12}
\definecolor{ParenColor}{rgb}{0.52,0.24,0.14}
\definecolor{IdentifierColor}{rgb}{0.15,0.15,0.50}
\definecolor{ResultColor}{rgb}{0.0,0.0,0.69}
\definecolor{ValueColor}{rgb}{0.13,0.55,0.13}
\definecolor{OutputColor}{rgb}{0.59,0.00,0.59}

\newcommand{\RktPn}[1]{\inColor{ParenColor}{#1}}

\newcommand{\RktSym}[1]{\inColor{IdentifierColor}{#1}}

\newcommand{\RktVal}[1]{\inColor{ValueColor}{#1}}

\newcommand{\RktMeta}[1]{\inColor{IdentifierColor}{#1}}

\newcommand{\RktRdr}[1]{\inColor{black}{#1}}

\newenvironment{RktBlk}{}{}

\newcommand{\RBackgroundLabel}[1]{}

\newcommand{\NoteBox}[1]{\footnote{#1}}
\newcommand{\NoteContent}[1]{#1}

\newcommand{\Footnote}[1]{\footnote{#1}}
\newcommand{\FootnoteRef}[1]{}
\newcommand{\FootnoteTarget}[1]{}
\newcommand{\FootnoteContent}[1]{#1}

\newenvironment{FootnoteBlock}{\renewcommand{\noindent}{}}{}
\newcommand{\FootnoteBlockContent}[1]{}
\renewcommand{\titleAndVersionAndAuthors}[3]{\title{#1}#3\maketitle}
\renewcommand{\titleAndEmptyVersionAndAuthors}[3]{\titleAndVersionAndAuthors{#1}{#2}{#3}}

\def\SAuthor#1{\SAutoAuthor#1\SAutoAuthorDone{#1}}
\def\SAutoAuthorDone#1{}
\def\SAutoAuthor{\futurelet\next\SAutoAuthorX}
\def\SAutoAuthorX{\ifx\next\SAuthorinfo \let\Snext\relax \else \let\Snext\SToAuthorDone \fi \Snext}
\def\SToAuthorDone{\futurelet\next\SToAuthorDoneX}
\def\SToAuthorDoneX#1{\ifx\next\SAutoAuthorDone \let\Snext\SAddAuthorInfo \else \let\Snext\SToAuthorDone \fi \Snext}
\newcommand{\SAddAuthorInfo}[1]{\SAuthorinfo{#1}{}{}}

\renewcommand{\SAuthorinfo}[4]{\author{#1}{#2}{#3}{#4}}
\renewcommand{\SAuthorSep}[1]{}

\renewcommand{\SAuthorPlace}[1]{\affiliation{#1}}
\renewcommand{\SAuthorEmail}[1]{\email{#1}}

\renewcommand{\SConferenceInfo}[2]{\conferenceinfo{#1}{#2}}
\renewcommand{\SCopyrightYear}[1]{\copyrightyear{#1}}
\renewcommand{\SCopyrightData}[1]{\copyrightdata{#1}}

\renewcommand{\SCategory}[3]{\category{#1}{#2}{#3}}
\renewcommand{\SCategoryPlus}[4]{\category{#1}{#2}{#3}[#4]}
\renewcommand{\STerms}[1]{\terms{#1}}
\renewcommand{\SKeywords}[1]{\keywords{#1}}

\DeclareUnicodeCharacter{2032}{'}
\DeclareUnicodeCharacter{0394}{\ensuremath{\Delta}}
\DeclareUnicodeCharacter{03D1}{\ensuremath{\vartheta}}

\AtEndPreamble{
  \theoremstyle{remark}
  
  \newtheorem*{remark*}{Remark}
  \newtheorem*{note*}{Note}
}
\usepackage{boxedminipage}
\usepackage{enumerate}
\usepackage{fancyvrb}
\usepackage{tikz}
\usetikzlibrary{shapes.geometric}
\usetikzlibrary{tikzmark}

\usepackage{letltxmacro}
\usepackage[figure,vlined]{algorithm2e}
\usepackage{wrapfig}

\LetLtxMacro{\AIf}{\If}
\SetKwProg{Proc}{Procedure}{}{end}
\renewcommand{\FuncSty}[1]{\textnormal{\textsf{#1}}\unskip} % was: \textnormal \texttt
\DeclareMathOperator{\OrAssign}{\texttt{|}\!\!=}

\usepackage{textgreek} % for \textlambda

\newcommand{\Gl}{\lambda}
\newcommand{\tGl}{\textlambda}

\newcommand{\Gg}{\gamma}

\newcommand{\GD}{\Delta}

\newcommand{\lset}[1]{\left\{#1\right\}}

\newcommand{\subst}[3]{#1\!\left.\left[#3 \right/ #2\right]}
\DeclareMathOperator{\lappend}{+\!\!+}

\newcommand{\dom}[1]{\mathsf{dom}\!\left(#1\right)}
\newcommand{\To}{\rightarrow}
\newcommand{\Implies}{\Rightarrow}
\newcommand{\From}{\leftarrow}
\newcommand{\ouro}{\mathsf{unreachable}}
\newcommand{\relE}{\mathcal{E}}
\newcommand{\relV}{\mathcal{V}}
\newcommand{\relG}{\mathcal{G}}
\newcommand{\Lsafep}{\textsc{Safe}}
\newcommand{\safep}[1]{\Lsafep\!\left(#1\right)}

\newcommand{\Lundefp}{\textsc{Undef}}
\newcommand{\undefp}[1]{\Lundefp\!\left(#1\right)}

\newcommand{\Ifk}[0]{\textsf{if}}
\renewcommand{\If}[3]{\textsf{(}\Ifk\,#1\,#2\,#3\textsf{)}}
\newcommand{\Op}[3]{\mathit{#1}\,#2\,#3}
\newcommand{\Func}[2]{\text{\tGl}#1.#2}
\newcommand{\App}[2]{#1\,#2}
\newcommand{\Seqk}[0]{\textsf{begin}}
\newcommand{\Seq}[2]{\textsf{(}\Seqk\,#1\,#2\textsf{)}}
\newcommand{\True}{\mathsf{true}}
\newcommand{\False}{\mathsf{false}}
\newcommand{\error}[1]{\mathsf{error}_{#1}}
\DeclareMathOperator{\arrs}{\mathit{R}_s}
\newcommand{\arrS}{\To_{s}}

\DeclareMathOperator{\arrp}{\mathit{R}_p}
\newcommand{\arrP}{\To_{p}}
\DeclareMathOperator{\arru}{\mathit{R}_{u}}
\newcommand{\arrU}{\To_{u}}
\newcommand{\uarr}{ _{u}\From}
\DeclareMathOperator{\arrm}{\mathit{R}_{m}}
\newcommand{\arrM}{\To_{m}}
\newcommand{\arrC}{\To_{c}}
\newcommand{\wfe}[2]{#1 \mmodels #2}
\newcommand{\wfctx}[3]{#1 \mmodels #3 : #2}
\newcommand{\wfarrm}[3]{#1 \vdash #2 \arrm #3}
\newcommand{\mmodels}{\Vdash}

\usepackage[inference]{semantic}
\usepackage{cleveref}
\usepackage{mathpartir}
\usepackage{stmaryrd}
\usepackage{bussproofs}
\usepackage{latexsym}
\newcommand{\mi}[1]{\ensuremath{\mathit{#1}}}

\newcommand{\bnfdef}[0]{\ensuremath{\mathrel{::=}}}

\newcommand{\ctx}[0]{\ensuremath{\mi{C}}}

\newcommand{\hole}[1]{\ensuremath{\left[#1\right]}}
\newcommand{\evalctx}[0]{\ensuremath{\mi{E}}}
\newcommand{\commoncol}[0]{black}    % CarnationPink

\newcommand{\col}[2]{\ensuremath{{\color{#1}{#2}}}}

\newcommand{\com}[1]{\mi{\col{\commoncol }{#1}}}

\newcounter{typerule}
\crefname{typerule}{rule}{rules}

\newcommand{\typeruleInt}[5]{%									    % #1 is the title, #2 is the hypotheses. #3 is the thesis, #4 is the label for referencing
	\def\thetyperule{#1}%
	\refstepcounter{typerule}%
  \ensuremath{\begin{array}{c}\inference{#2}{#3}[#5]\end{array}} 
}
\newcommand{\typerule}[4]{%									        % #1 is the title, #2 is the hypotheses. #3 is the thesis, #4 is the label for referencing
  \typeruleInt{#1}{#2}{#3}{#4}{\rul{#1}}
}

\Crefname{corollary}{Corollary}{Corollaries}
\Crefname{informal}{Definition}{Definition}
\Crefname{assumption}{Assumption}{Assumptions}
\crefname{assumption}{Assumption}{Assumptions}
\Crefname{property}{Property}{Properties}
\crefname{property}{Property}{Properties}
\Crefname{lstlisting}{Listing}{Listings}
\Crefname{problem}{Problem}{Problems}
\Crefname{equation}{Rule}{Rules}

\newcommand{\rul}[1]{\textsc{#1}}
\DeclareMathOperator{\tteq}{\texttt{=}}

\newcommand{\Syntax}{
  \begin{align*} 
    \com{e} \bnfdef&\ \com{x} \mid \com{c} \mid \Func{x}{e} \mid \Op{op}{e}{e} \mid \App{e}{e} \mid \If{e}{e}{e} \mid \Seq{e}{e}  \mid \ouro \mid \error{k} \\ 
    \com{c} \bnfdef&\ \com{n} \mid \False \mid \True \hbox to .3in{} \com{n}\ \in\ \mbox{numbers} \hbox to .3in{}     \com{op}\ \in\ \mbox{primitive operations}\\
  \end{align*}
}

\newcommand{\StandardReductions}{
      \hfill\fbox{$e\ \arrS\ e\vphantom{\arrs}$}~~\fbox{$e \ \arrs \ e$}

  \begin{align*} 
    \com{\evalctx} \bnfdef&\ \hole{} \mid \Op{op}{\evalctx}{e} \mid \Op{op}{v}{\evalctx} \mid \If{\evalctx}{e}{e}\mid \App{\evalctx}{e} \mid \App{v}{\evalctx} \mid \Seq{\evalctx}{e} \\
  \end{align*}
  \vspace*{-.4in}
  \begin{align*} 
    \com{v}\ \bnfdef\ \com{c} \mid \Func{x}{e}
    \hbox to .4in{}
    \com{a}\ \bnfdef\ \com{c} \mid \Func{x}{e} \mid \ouro \mid \error{k} \\
  \end{align*}
  \vspace*{-.25in}
    \begin{center}
      \typerule{S.1}{}{
        \If{\False}{e_1}{e_2} \ \arrs \ e_2
      }{S.1}
      \typerule{S.2}{v\not\equiv\False}{
        \If{v}{e_1}{e_2}\ \arrs \ e_1
      }{S.2}
      \typerule{S.3}{}{
        (\App{(\Func{x}{e})}{v}) \ \arrs \ \subst{e}{x}{v}
      }{S.3}

      \vspace{0.57em}
      \typerule{S.4}{
      }{
        \Seq{v}{e} \ \arrs \ e
      }{S.4}
      \typerule{S.5}{
        \delta(op, v_1, v_2) \equiv c
      }{
        (\Op{op}{v_1}{v_2}) \ \arrs \ c
      }{S.5}
      \typerule{S.6}{
        e \ \arrs e'
      }{
        \evalctx\hole{e} \ \arrS \ \evalctx\hole{e'}
      }{S.6}

      \vspace{0.75em}
      \typerule{S.7}{
          \evalctx \not \equiv \hole{}
      }{
        \evalctx\hole{\ouro} \ \arrS \ \ouro
      }{S.7}
      \typerule{S.8}{
          \evalctx \not \equiv \hole{}
      }{
        \evalctx\hole{\error{k}} \ \arrS \ \error{k}
      }{S.8}

      \vspace{0.75em}

      \typerule{S.9}{
        v_1 \not \equiv \Func{x}{e}
      }{
        (\App{v_1}{v_2}) \ \arrs \ \error{\beta}
      }{S.9}
      \typerule{S.10}{
        (op, v_1, v_2) \not\in \dom{\delta}
      }{
        (\Op{op}{v_1}{v_2}) \ \arrs \ \error{\delta}
      }{S.10}
    \end{center}
    \vspace*{.1in}
}

\newcommand{\CompilerRulesContext}{
\hfill\fbox{$e \ \arrC \ e\vphantom{\arrp}$}

\vspace*{-.2in}
  \begin{center}
    \typerule{}{e \arrP e'}{
      e \arrC e'
    }{}\typerule{}{e \arrU e'}{
      e \arrC e'
    }{}
  \end{center}

  \begin{align*} 
    \com{\ctx} \bnfdef&\ \hole{} \mid \Op{op}{\ctx}{e} \mid \Op{op}{e}{\ctx} \mid \If{\ctx}{e}{e} \mid \If{e}{\ctx}{e} \mid \If{e}{e}{\ctx} \\
    \mid &\ \App{\ctx}{e} \mid \App{e}{\ctx}  
    \mid \Func{x}{\ctx} \mid \Seq{\ctx}{e} \mid \Seq{e}{\ctx}
  \end{align*}

}

\newcommand{\CompilerRulesP}{
\hfill\fbox{$e\ \arrp\ e\vphantom{\arrp}$}~~\fbox{$e \ \arrP \ e\vphantom{\arrp}$}

  \begin{center}
    \typerule{P.1}{\safep{e}}{
      \Seq{e}{\ouro} \ \arrp \ouro
    }{P.1}\typerule{P.2}{}{
      \Seq{\ouro}{e} \ \arrp \ouro
    }{P.2}
    
    \vspace{0.75em}
    \typerule{P.3}{}{
      (\App{(\Func{x}{\ouro})}{e}) \ \arrp \Seq{e}{\ouro}
    }{P.3}

    \vspace{0.75em}
    \typerule{P.4}{}{
      (\App{\ouro}{e}) \ \arrp \ \ouro
    }{P.4}
    \typerule{P.5}{}{
      (\App{e}{\ouro}) \ \arrp \ \Seq{e}{\ouro}
    }{P.5}

    \typerule{Ctx.P}{e \ \arrp \ e'}{
      \ctx\hole{e} \arrP \ctx\hole{e'}
    }{Ctx.P}
    \typerule{CtxSym.P}{e' \ \arrp \ e}{
      \ctx\hole{e} \arrP \ctx\hole{e'}
    }{CtxSym.P}

  \end{center}

}

\newcommand{\CompilerRulesU}{
\hfill\fbox{$e\ R_u\ e\vphantom{\arrp}$}~~\fbox{$e \ \arrU \ e\vphantom{\arrp}$}

~

    \begin{center}
      \typerule{U.1}{}{
        \If{e_c}{\ouro}{e_f} \ R_u \ \Seq{e_c}{e_f}
      }{U.1}
      \typerule{U.2}{}{
        \If{e_c}{e_t}{\ouro} \ R_u \ \Seq{e_c}{e_t}
      }{U.2}
    \end{center}

  \begin{center}
      \typerule{Ctx.U}{e \ \arru \ e'}{
      \ctx\hole{e} \arrU \ctx\hole{e'}
    }{Ctx.U}
  \end{center}
}

\newcommand{\WellformedTermsSelectedRules}
{
    \AxiomC{$x \in \GD$}
    \RightLabel{\rul{W.Var}}
    \UnaryInfC{$\wfe{\GD}{x}$ }
  \DisplayProof
  \hfill
  \fbox{$\wfe{\GD}{e}$}

  \vspace{0.75em}
    \AxiomC{$\wfe{\GD, x}{e}  $}
    \RightLabel{\rul{W.Lambda}}
    \UnaryInfC{$\wfe{\GD}{\Func{x}{e}}$ }
  \DisplayProof
}

\newcommand{\WellFormedContextsSelectedRules}{
    \AxiomC{$\GD'\supseteq\GD$  }
    \RightLabel{\rul{C.Id}}
    \UnaryInfC{$\wfctx{\GD'}{\GD}{\hole{}}$ }
  \DisplayProof
  \hspace{0.1cm}
    \AxiomC{$\wfctx{\GD',x}{\GD}{C}$  }
    \RightLabel{\rul{C.Lambda}}
    \UnaryInfC{$\wfctx{\GD'}{\GD}{\Func{x}{C}}$ }
  \DisplayProof
  \hfill
  \fbox{$\wfctx{\GD'}{\GD}{C}$}

  \vspace{0.75em}
    \AxiomC{$\wfctx{\GD'}{\GD}{C}$}
    \AxiomC{$\wfe{\GD'}{e_t} $}
    \AxiomC{$\wfe{\GD'}{e_f} $}
    \RightLabel{\rul{C.IfC}}
    \TrinaryInfC{$\wfctx{\GD'}{\GD}{\If{C}{e_t}{e_f}}$ }
  \DisplayProof
}

\newcommand{\ContextPreservingRules}{
  \hfill\fbox{$e \ \arrC \ e\vphantom{\arrs}$}~~\fbox{$\wfarrm{\GD}{e}{e}\vphantom{\arrs}$}~~\fbox{$e \ \arrM \ e\vphantom{\arrs}$} % {$\wfe{\GD}{e}\vphantom{\arrs}$}
  \begin{center}
    \vspace{0.75em}
    \typerule{}{e \arrM e'}{
      e \arrC e'
    }{}
    \vspace{0.75em}
      \typerule{M.1}{\wfe{\GD}{e_t} & \wfe{\GD}{e_f} & v \not\equiv \False}{
      \wfarrm{\GD}{e_t}{\If{v}{e_t}{e_f}}
    }{M.1}
    \typerule{M.2}{\wfe{\GD}{e_f} & \wfe{\GD}{e_t} }{
      \wfarrm{\GD}{e_f}{\If{\False}{e_t}{e_f}}
    }{M.2}

    \vspace{0.75em}
    \typerule{M.3}{\wfe{\GD ,x}{e} & \wfe{\GD}{e'} & \safep{e'}}{
      \wfarrm{\GD}{\subst{e}{x}{e'}}{\App{(\Func{x}{e})}{e'}}
    }{M.3}

    \vspace{0.75em}
    \typerule{M.4}{\wfe{\GD}{e_1} & \wfe{\GD}{e_2} & \safep{e_1}}{
      \wfarrm{\GD}{\Seq{e_1}{e_2}}{e_2}
    }{M.4}
    \typerule{M.5}{\delta(op,v_1,v_2) = c}{
      \wfarrm{\GD}{c}{(\Op{op}{v_1}{v_2})}
    }{M.5}

    \vspace{1em}
    \typerule{Ctx.M}{\wfctx{\cdot}{\GD}{C} & \wfarrm{\GD}{e}{e'}}{
      \ctx\hole{e} \arrM \ctx\hole{e'}
    }{Ctx.M}
    \typerule{CtxSym.M}{\wfctx{\cdot}{\GD}{C} & \wfarrm{\GD}{e'}{e}}{
      \ctx\hole{e} \arrM \ctx\hole{e'}
    }{CtxSym.M}
    \vspace{0.75em}
  \end{center}
}

\newcommand{\AdditionalRMRules}{
  \hfill\fbox{$\wfctx{\GD}{\GD}{\ctx}\vphantom{\arrs}$}~~\fbox{$\wfarrm{\GD}{e}{e}\vphantom{\arrs}$}~~\fbox{$\wfe{\GD}{e}\vphantom{\arrs}$}
\vspace*{-.1in}
  \begin{align*}
    \com{v^+}  \bnfdef&\ \com{x} \mid \com{v} \\
    \com{\evalctx^+} \bnfdef&\ \hole{} \mid \Op{op}{\com{\evalctx^+}}{e} \mid \Op{op}{v^+}{\com{\evalctx^+}} \mid \If{\com{\evalctx^+}}{e}{e}\mid \App{\com{\evalctx^+}}{e} \mid \App{v^+}{\com{\evalctx^+}} \mid \Seq{\com{\evalctx^+}}{e}\\
  \end{align*}
    \begin{center}
    \typerule{M.6}{\wfctx{\GD}{\GD}{\evalctx^+}  & \wfe{\GD}{e_f} & \wfe{\GD}{e_s}}{
    \wfarrm{\GD}{\evalctx^+\!\hole{\Seq{e_f}{e_s}}}{\Seq{e_f}{(\evalctx^+[e_s])}}
    }{M.6}

     \vspace{0.5em}
    \typerule{M.7}{\wfctx{\GD}{\GD}{\evalctx^+}  &  \wfe{\GD}{e_c} & \wfe{\GD}{e_f} & \wfe{\GD}{e_t} }{
      \wfarrm{\GD}{\evalctx^+\!\hole{\If{e_c}{e_f}{e_s}}}{\If{e_c}{(\evalctx^+[e_f])}{(\evalctx^+[e_s])}}
    }{M.7}

    \vspace{0.5em}
    \typerule{M.8}{\wfe{\GD ,x}{e_t} &\wfe{\GD ,x}{e_f}}{
      \wfarrm{\GD}{\If{x}{e_t}{e_f}}{\If{x}{e_t}{(\subst{e_f}{x}{\False})}}
    }{M.8}

    \typerule{M.9}{\wfe{\GD}{e_c} & \wfe{\GD}{e_t} & \wfe{\GD}{e_f}}{
      \wfarrm{\GD}{\If{\If{e_c}{\True}{\False}}{e_t}{e_f}}{\If{e_c}{e_t}{e_f}}
    }{M.9}

    \vspace{0.5em}
    \typerule{M.10}{\wfe{\GD}{e_1} & \wfe{\GD}{e_2} & \wfe{\GD}{e_3}}{
      \wfarrm{\GD} {\If{(x\tteq n_1)}{e_1}{\If{(x\tteq n_2)}{e_1}{e_3}}} {\If{(x\tteq n_2)}{e_1}{\If{(x\tteq n_1)}{e_1}{e_3}}}
    }{M.10}
  \end{center}
}

\newcommand{\Domk}{\succcurlyeq}
\newcommand{\Dom}[2]{#1\Domk{}#2}
\newcommand{\CommentColor}{teal}

\SetKwFunction{SimplifyUnreachable}{simplifyUnreachable}
\SetKwFunction{SimplifyFunctionCFG}{SimplifyFunctionCFG}
\SetKwFunction{ChangeToUnreachable}{changeToUnreachable}
\SetKwFunction{MarkAliveBlocks}{markAliveBlocks}
\SetKwFunction{DeleteDeadBlock}{DeleteDeadBlock}
\SetKwFunction{Length}{length}

\newcommand{\Tmn}{\mathit{tmn}}
\newcommand{\Val}{\mathit{val}}
\newcommand{\GetIncomingValueForBlock}{\mathit{getIncomingValueForBlock}}
\newcommand{\PhiNode}[2]{#1=\mathbf{phi}\,#2}
\newcommand{\Ret}[1]{\mathbf{ret}\,#1}
\newcommand{\CallError}{\mathbf{call}\,\textnormal{\textsf{error}()}}
\newcommand{\BrUncond}[1]{\mathbf{br}\,#1}
\newcommand{\BrCond}[3]{\mathbf{br}\,#1\,#2\,#3}

\definecolor{somewhatgolden}{RGB}{255,194,10}
\newcommand{\hlboxm}[1]{{\setlength{\fboxsep}{2pt}\colorbox{somewhatgolden}{$#1$}}}

\newcommand{\NTVminus}{
  \[\begin{array}{@{}c@{}}
    \begin{array}{@{}rcl@{}}
      \com{f} &\bnfdef & \mathbf{fun}\,\{ \overline{b} \}
        \\
      \com{b} &\bnfdef & (l\; \overline{\phi}\; \overline{c}\; \Tmn)
        \\
      \com{\phi} &\bnfdef& \PhiNode{r}{\overline{[\Val_j,l_j]}^j}
        \\
      \com{c} &\bnfdef & r \coloneq \Val_1 \,op\, \Val_2 \;|\; \hlboxm{\CallError}
        \\
      \Val &\bnfdef & r \;|\; n
        \\
      \Tmn &\bnfdef & \BrCond{\Val}{l_1}{l_2} \;|\; \Ret{\Val}
        \\
        & & |\; \hlboxm{\BrUncond{l}} \;|\; \hlboxm{\ouro}
    \end{array} \vspace{0.5em}\\
    l \in \textnormal{\text{Label}} \hbox to .5in{} r \in \textnormal{\text{Register}}
  \end{array}\]
}

\newcommand{\KH}{\mathcal{H}}
\newcommand{\KHcmd}{\KH_{\mathit{cs}}}
\newcommand{\KHproc}{\KH_{\mathit{proc}}}
\newcommand{\KHjump}{\KH_{\mathit{jump}}}
\newcommand{\KHterm}{\KH_{\mathit{term}}}

\newcommand{\mwhere}{\textnormal{where}}

\newcommand{\KHjumpfn}{
  \KHproc(f) = \Func{x}{\left(\App{\KHjump(f, l_0)\,}{0}\right)}
    \hfill
    \;
    \mwhere\, x\, \text{is fresh and $l_0$ is the label of the entry block}
    \vspace{0.5em}
    \\
  \KHjump(f,l) = \Func{r_1\dots r_n}{\;\KHcmd(f,l,\overline{c})}
    \\
  \quad \quad \mwhere \;
      f[l]=\lfloor l\;\overline{\phi}\;\overline{c}\;\text{$\Tmn$}\rfloor
      \;\text{and}\;
      \overline{\phi}=
      (\PhiNode{r_1}{\overline{[\Val_{j_1},l_{j_1}]}^{j_1}})\ldots
      (\PhiNode{r_n}{\overline{[\Val_{j_n},l_{j_n}]}^{j_n}})
      \\
  \quad \quad
      \text{if $\overline{\phi}=[]$, we introduce a dummy fresh variable $r_0$}
}

\newcommand{\KHcmdfn}{
  \KHcmd : f \; l\; \overline{c} \longrightarrow e
    \\
  \begin{array}{@{}l@{\;}c@{\;}l@{}}
    \KHcmd(f,l,((\CallError),\, \overline{c'})) &=&
      \Seq{\,\error{}}{\,\KHcmd(f,l,\overline{c'})}
      \\
    \KHcmd(f,l,((r \coloneq \Val_1\,\mathit{op}\,\Val_2),\, \overline{c'})) &=&
      (\textsf{let}\;\;([r\;(\Op{\mathit{op}}{\Val_1}{\Val_2})]) \;\; \KHcmd(f,l,\overline{c'}))
      \\
    \KHcmd(f,l,[]) &=&
      (\textsf{letrec}\;([l_1\;\KHjump(f,l_1)]
      \dots
      [l_m\;\KHjump(f,l_m)])
      \\
    &&
      \quad \KHterm(f,l))
      \\
      \multicolumn{3}{@{}l@{}}{
        \quad \quad
        \mwhere\,l_1\dots l_m\,
        \text{are the children of node $l$
        in the dominator tree}
      }
    \end{array}
}

\newcommand{\KHtermfn}{
  \KHterm : f \; l \longrightarrow e
    \\
  \begin{array}{@{}l@{\,}c@{\,}l@{\quad \quad \quad \quad}l@{}}
    \KHterm(f,l) &=&
      \Val
      & \text{if}\;
      f[l]=\lfloor l\; \overline{\phi}\; \overline{c}\; (\Ret{\Val}) \rfloor
      \\
    \KHterm(f,l) &=&
      \ouro
      & \text{if}\;
      f[l]=\lfloor l\; \overline{\phi}\; \overline{c}\; \ouro \rfloor
      \\
    \KHterm(f,l) &=&
      (\App{l'}{\,\overline{\Val'}})
      & \text{if}\; f[l]=\lfloor l\; \overline{\phi}\; \overline{c}\; (\BrUncond{l'}) \rfloor
      \\
      \multicolumn{4}{@{}l@{}}{\begin{array}{@{}l@{}}
      \quad \quad
        \mwhere \;
        f[l']=\lfloor l'\;\overline{\phi'}\;\overline{c'}\;\Tmn'\rfloor, \;
        \overline{\GetIncomingValueForBlock(\phi',l)}=\overline{\Val'}
        \\
      \quad \quad
        \text{when $\overline{\phi'}=[]$, we supply a dummy argument $0$}
      \end{array}}
      \\
    \KHterm(f,l) &=&
      \textnormal{$\If{\Val}{
        (\App{l'}{\,\overline{\Val'}})
      }{
        (\App{l''}{\,\overline{\Val''}})
      }$}
      & \text{if}\;
        f[l]=\lfloor l\; \overline{\phi}\; \overline{c}\; (\BrCond{\Val}{l'}{l''}) \rfloor
      \\
      \multicolumn{4}{@{}l@{}}{\begin{array}{@{}l@{}}
      \quad \quad \mwhere\; 
        f[l']=\lfloor l'\;\overline{\phi'}\;\overline{c'}\;\Tmn'\rfloor, \;
        \overline{\GetIncomingValueForBlock(\phi',l)}=\overline{\Val'}
        \\
      \hphantom{\quad \quad \mwhere\; }
        f[l'']=\lfloor l''\;\overline{\phi''}\;\overline{c''}\;\Tmn''\rfloor, \;
        \overline{\GetIncomingValueForBlock(\phi'',l)}=\overline{\Val''}
        \\
      \quad \quad
        \text{we supply $0$ as a dummy argument if there are no $\phi$ nodes as in the case for $(\BrUncond{l'})$}
      \end{array}}
  \end{array} \vspace{0.5em} \\
  \GetIncomingValueForBlock : \phi \; l \longrightarrow \Val
    \\
  \textnormal{$\GetIncomingValueForBlock(\phi,l) = \Val_i$}
  \hfill \mwhere \;
  \phi = (\PhiNode{r}{[\Val_1,l_1]\dots{}[\Val_n,l_n]})
  \,\text{and}\,
  l=l_i
}

\renewcommand{\SColorize}[2]{#2}

\newcommand{\NamedTheorem}[2]{\begin{theorem}[#1]#2\end{theorem}}
\newcommand{\NamedLemma}[2]{\begin{lemma}[#1]#2\end{lemma}}

\begin{document}
\preDoc

\begin{abstract}In Racket, the LLVM IR, Rust, and other modern languages, programmers and
static analyses can hint, with special annotations, that certain parts of
a program are unreachable. Same as other assumptions about undefined
behavior; the compiler assumes these hints are correct and transforms the
program aggressively.

While compile{-}time transformations due to undefined behavior often perplex
compiler writers and developers, we show that the essence of
transformations due to unreachable code can be distilled in a surprisingly
small set of simple formal rules. Specifically, following the
well{-}established tradition of understanding linguistic phenomena through
calculi, we introduce the first calculus for unreachable. Its
term{-}rewriting rules that take advantage of unreachable  fall into two
groups.  The first group allows the compiler to delete any code downstream
of unreachable, and any effect{-}free code upstream of unreachable.  The
second group consists of rules that eliminate conditional expressions when
one of their branches is unreachable. We show the correctness of the rules
with a novel logical relation, and we examine how
they correspond to transformations due to unreachable in Racket and LLVM.\end{abstract}\titleAndEmptyVersionAndAuthors{A Calculus for Unreachable Code}{}{\SNumberOfAuthors{6}\SAuthor{\SAuthorinfo{Peter Zhong}{}{\SAuthorPlace{\institution{PLT, Northwestern University}\country{U.S.A.}}}{\SAuthorEmail{peterzhong2023@u.northwestern.edu}}}\SAuthorSep{}\SAuthor{\SAuthorinfo{Shu{-}Hung You}{}{\SAuthorPlace{\institution{PLT, Northwestern University}\country{U.S.A.}}}{\SAuthorEmail{shu-hung.you@eecs.northwestern.edu}}}\SAuthorSep{}\SAuthor{\SAuthorinfo{Simone Campanoni}{}{\SAuthorPlace{\institution{Northwestern University}\country{U.S.A.}}}{\SAuthorEmail{simonec@eecs.northwestern.edu}}}\SAuthorSep{}\SAuthor{\SAuthorinfo{Robert Bruce Findler}{}{\SAuthorPlace{\institution{PLT, Northwestern University}\country{U.S.A.}}}{\SAuthorEmail{robby@cs.northwestern.edu}}}\SAuthorSep{}\SAuthor{\SAuthorinfo{Matthew Flatt}{}{\SAuthorPlace{\institution{University of Utah}\country{U.S.A.}}}{\SAuthorEmail{mflatt@cs.utah.edu}}}\SAuthorSep{}\SAuthor{\SAuthorinfo{Christos Dimoulas}{}{\SAuthorPlace{\institution{PLT, Northwestern University}\country{U.S.A.}}}{\SAuthorEmail{chrdimo@northwestern.edu}}}}
\label{t:x28part_x22Ax5fCalculusx5fforx5fUnreachablex5fCodex22x29}

\noindent 

\noindent 

\noindent 

\noindent 

\noindent

\sectionNewpage

\Ssection{The Fleeting Essence of Unreachable}{The Fleeting Essence of Unreachable}\label{t:x28part_x22Thex5fFleetingx5fEssencex5fofx5fUnreachablex22x29}

Modern compilers routinely take advantage of undefined behavior.
Specifically, they assume that a program never performs operations that
exhibit
undefined behavior, and then they use this assumption to optimize
the program. After all, a program that lives outside the confines of defined behavior
is meaningless. Hence, the effect of optimizations on its behavior
is a moot point. While the use of undefined behavior in semantics is controversial,
and with good reason\Autobibref{~(\hyperref[t:x28autobib_x22CVEx2d2014x2d01602014httpsx3ax2fx2fcvex2emitrex2eorgx2fcgix2dbinx2fcvenamex2ecgix3fnamex3dcvex2d2014x2d0160The_Heartbleed_Bugx2e_Retrievedx3a_Julyx2c_2022x2e__Discovered_by_Neel_Mehta_from_Googlex2ex22x29]{\AutobibLink{CVE{-}2014{-}0160}} \hyperref[t:x28autobib_x22CVEx2d2014x2d01602014httpsx3ax2fx2fcvex2emitrex2eorgx2fcgix2dbinx2fcvenamex2ecgix3fnamex3dcvex2d2014x2d0160The_Heartbleed_Bugx2e_Retrievedx3a_Julyx2c_2022x2e__Discovered_by_Neel_Mehta_from_Googlex2ex22x29]{\AutobibLink{2014}})}, undefined behavior is widely used and, as \Autobibref{\hyperref[t:x28autobib_x22Ralf_JungUndefined_Behavior_deserves_a_better_reputation2021httpsx3ax2fx2fblogx2esigplanx2eorgx2f2021x2f11x2f18x2fundefinedx2dbehaviorx2ddeservesx2dax2dbetterx2dreputationx2fRetrievedx3a_Julyx2c_2022x2ex22x29]{\AutobibLink{Jung}}~(\hyperref[t:x28autobib_x22Ralf_JungUndefined_Behavior_deserves_a_better_reputation2021httpsx3ax2fx2fblogx2esigplanx2eorgx2f2021x2f11x2f18x2fundefinedx2dbehaviorx2ddeservesx2dax2dbetterx2dreputationx2fRetrievedx3a_Julyx2c_2022x2ex22x29]{\AutobibLink{2021}})}
eloquently argued, not all forms of undefined behavior are the same.

In this paper, we attempt to bring a sound semantics footing of a particular form of
undefined behavior: \texorpdfstring{\ensuremath{\ouro}}{unreachable}. Programmers and tools that
transform code, such as static analyses, introduce the construct \texorpdfstring{\ensuremath{\ouro}}{unreachable} in
programs to claim that some part of the program can never be evaluated.
Concretely, figure~\hyperref[t:x28counter_x28x22figurex22_x22figx3atwox2dflavorsx22x29x29]{\FigureRef{1}{t:x28counter_x28x22figurex22_x22figx3atwox2dflavorsx22x29x29}} depicts two code snippets that
demonstrate \texorpdfstring{\ensuremath{\ouro}}{unreachable} in two different linguistic contexts. In the
Racket snippet, \RktSym{unsafe{-}assert{-}unreachable} signals to the Racket
compiler that the else{-}branch of the \RktSym{if} expression can never be
reached. Hence,
\RktSym{unsafe{-}assert{-}unreachable} gives the
Racket compiler the freedom to optimize away the entire \RktSym{if} expression,
effectively asserting that \RktSym{x} is always a pair.
In the
same spirit, LLVM{'}s \RktSym{unreachable}\RktMeta{} instruction informs the compiler that the
control{-}flow of a program never reaches that instruction.
Therefore, in the LLVM IR snippet in
figure~\hyperref[t:x28counter_x28x22figurex22_x22figx3atwox2dflavorsx22x29x29]{\FigureRef{1}{t:x28counter_x28x22figurex22_x22figx3atwox2dflavorsx22x29x29}}, all three basic blocks can be
collapsed into one. In other words, same as for other forms of undefined
behavior, \texorpdfstring{\ensuremath{\ouro}}{unreachable} describes an assumption about the evaluation of a
program that a compiler can take advantage of while optimizing the
program. If the assumption is true, then the program is transformed to an
equivalent, possibly more performant, new version. If the assumption is
false, then the program is meaningless to begin with and the semantics
of the language gives the compiler a free hand to produce any code whatsoever.

\begin{Figure}\begin{Centerfigure}\begin{FigureInside}\begin{bigtabular}{@{\bigtableleftpad}c@{}c@{}c@{}}
\SCodeInsetBox{\begin{RktBlk}\begin{tabular}[c]{@{}l@{}}
\hbox{\RktPn{(}\RktSym{define}\mbox{\hphantom{\Scribtexttt{x}}}\RktPn{(}\RktSym{func}\mbox{\hphantom{\Scribtexttt{x}}}\RktSym{x}\RktPn{)}} \\
\hbox{\mbox{\hphantom{\Scribtexttt{xx}}}\RktPn{(}\RktSym{if}\mbox{\hphantom{\Scribtexttt{x}}}\RktPn{(}\RktSym{pair{\hbox{\texttt{?}}}}\mbox{\hphantom{\Scribtexttt{x}}}\RktSym{x}\RktPn{)}} \\
\hbox{\mbox{\hphantom{\Scribtexttt{xxxxxx}}}\RktPn{(}\RktSym{car}\mbox{\hphantom{\Scribtexttt{x}}}\RktSym{x}\RktPn{)}} \\
\hbox{\mbox{\hphantom{\Scribtexttt{xxxxxx}}}\RktPn{(}\RktSym{unsafe{-}assert{-}unreachable}\RktPn{)}\RktPn{)}\RktPn{)}}\end{tabular}\end{RktBlk}} &
\hbox{\mbox{\hphantom{\Scribtexttt{xxxxxx}}}} &
\begin{SVerbatim}\begin{tabular}[c]{@{}l@{}}
\begin{minipage}[c]{1.0\linewidth}
\Scribtexttt{define i32 @func(i32 \%0) {\char`\{}} \end{minipage}
 \\
\begin{minipage}[c]{1.0\linewidth}
\Scribtexttt{1{\hbox{\texttt{:}}} \%2 = icmp sgt i32 \%0, 0} \end{minipage}
 \\
\begin{minipage}[c]{1.0\linewidth}
\Scribtexttt{}\mbox{\hphantom{\Scribtexttt{xxx}}}\Scribtexttt{br i1 \%2, label \%3, label \%4} \end{minipage}
 \\
\begin{minipage}[c]{1.0\linewidth}
\Scribtexttt{3{\hbox{\texttt{:}}} ret i32 0}\mbox{\hphantom{\Scribtexttt{xxxxxxxxx}}}\Scribtexttt{; preds = \%1} \end{minipage}
 \\
\begin{minipage}[c]{1.0\linewidth}
\Scribtexttt{4{\hbox{\texttt{:}}} unreachable}\mbox{\hphantom{\Scribtexttt{xxxxxxx}}}\Scribtexttt{; preds = \%1} \end{minipage}
 \\
\begin{minipage}[c]{1.0\linewidth}
\Scribtexttt{{\char`\}}} \end{minipage}
\end{tabular}\end{SVerbatim}\end{bigtabular}\end{FigureInside}\end{Centerfigure}

\Centertext{\Legend{\FigureTarget{\label{t:x28counter_x28x22figurex22_x22figx3atwox2dflavorsx22x29x29}\textsf{Fig.}~\textsf{1}. }{t:x28counter_x28x22figurex22_x22figx3atwox2dflavorsx22x29x29}\textsf{Two Flavors of \texorpdfstring{\ensuremath{\ouro}}{unreachable}: Racket (Left) and
         LLVM IR (Right)}}}\end{Figure}

The discussion so far suggests that the phrase {``}the meaning of \texorpdfstring{\ensuremath{\ouro}}{unreachable}{''} is an
oxymoron. However, we show that staple tools of PL
semanticists, such as calculi and equational theories, can faithfully describe
the interplay between \texorpdfstring{\ensuremath{\ouro}}{unreachable}, compile{-}time transformations and
program behavior, and hence, reveal the essence of \texorpdfstring{\ensuremath{\ouro}}{unreachable}
through the way compilers use it.

Standing on the shoulders of fifty years of PL
tradition\Autobibref{~(\hyperref[t:x28autobib_x22Godon_Dx2e_Plotkinx7bCallx2dbyx2dnamex2c_callx2dbyx2dvalue_and_the_x3bbx2dcalculusTheoretical_Computer_Sciencex281x282x29x29x2c_ppx2e_125x2dx2d1591975x22x29]{\AutobibLink{Plotkin}} \hyperref[t:x28autobib_x22Godon_Dx2e_Plotkinx7bCallx2dbyx2dnamex2c_callx2dbyx2dvalue_and_the_x3bbx2dcalculusTheoretical_Computer_Sciencex281x282x29x29x2c_ppx2e_125x2dx2d1591975x22x29]{\AutobibLink{1975}})}, we turn to the $\lambda$ calculus
and we construct the first calculus that captures the essence of
\texorpdfstring{\ensuremath{\ouro}}{unreachable}. In our \texorpdfstring{\ensuremath{\ouro}}{unreachable} calculus, we describe the legal basic steps that
a compiler performs when transforming a program due to
\texorpdfstring{\ensuremath{\ouro}}{unreachable} as local term{-}rewriting rules. At a high level,
there are two groups of rules: those that eliminate \texorpdfstring{\ensuremath{\ouro}}{unreachable}
and those that propagate it through a program. The first group of rules dictates how a
compiler can discard computations around \texorpdfstring{\ensuremath{\ouro}}{unreachable} and these are the
the same as rules that
discard computations around expressions that are guaranteed to signal errors.
The second group describes how a
compiler can leverage \texorpdfstring{\ensuremath{\ouro}}{unreachable} to eliminate control{-}flow paths that
cannot
be reached. This group consists of just a single rule and its symmetric
counterpart, showing how \Scribtexttt{if} expressions and \Scribtexttt{unreachable} interact.

There are two questions about the \texorpdfstring{\ensuremath{\ouro}}{unreachable}
calculus. The first question is correctness. Specifically, do the
rules of the calculus preserve program semantics {---} as induced by the
standard reduction of the calculus {---} under the assumption
that evaluation of the program never reaches \texorpdfstring{\ensuremath{\ouro}}{unreachable}? The answer is affirmative but
the proof is subtle because the assumption about
\texorpdfstring{\ensuremath{\ouro}}{unreachable} is not compositional. Even when it is true that a
particular program does not reach \texorpdfstring{\ensuremath{\ouro}}{unreachable} no
matter its input, there may be subexpressions that, if they
were to be evaluated directly, would reach \texorpdfstring{\ensuremath{\ouro}}{unreachable}.
Worse, the guarantee that ensures the entire program cannot
reach \texorpdfstring{\ensuremath{\ouro}}{unreachable} may be arbitrarily complex. We show how
to handle this complexity via the construction of a novel
logical relation, while keeping the rules of the calculus simple and
intuitive.

The second question is about generality. Does the simplicity of the
calculus keep it from capturing the optimizations in production compilers?
To answer this question, we focus on the two compilers: Racket{'}s and
LLVM. For Racket, we performed an audit of the Racket codebase, which
revealed that all of the rewrites that Racket performs because of
\texorpdfstring{\ensuremath{\ouro}}{unreachable} map to the rules of our calculus. For LLVM, we cannot draw a
similarly direct conclusion because LLVM is too different from our $\lambda${-}based
calculus. Instead, inspired by an audit of the LLVM
codebase and additions to LLVM codebase at the same time as the addition
of \texorpdfstring{\ensuremath{\ouro}}{unreachable} we design a transformation for a subset of LLVM IR (technically
Vminus\Autobibref{~(\hyperref[t:x28autobib_x22_Jianzhou_Zhaox2c__Santosh_Nagarakattex2c__Milo_Mx2eKx2e_Martinx2c_and__Steve_ZdancewicFormal_Verification_of_SSAx2dBased_Optimizations_for_LLVMIn_Procx2e_ACM_Conference_on_Programming_Language_Design_and_Implementationx2c_PLDI_x2713x2c_ppx2e_175x2dx2d1862013httpsx3ax2fx2fdoix2eorgx2f10x2e1145x2f2491956x2e2462164x22x29]{\AutobibLink{Zhao et al\Sendabbrev{.}}} \hyperref[t:x28autobib_x22_Jianzhou_Zhaox2c__Santosh_Nagarakattex2c__Milo_Mx2eKx2e_Martinx2c_and__Steve_ZdancewicFormal_Verification_of_SSAx2dBased_Optimizations_for_LLVMIn_Procx2e_ACM_Conference_on_Programming_Language_Design_and_Implementationx2c_PLDI_x2713x2c_ppx2e_175x2dx2d1862013httpsx3ax2fx2fdoix2eorgx2f10x2e1145x2f2491956x2e2462164x22x29]{\AutobibLink{2013}})}) that captures some of the ways
that LLVM exploits \texorpdfstring{\ensuremath{\ouro}}{unreachable}.
To do so, we compose our transformation with a version of
\Autobibref{\hyperref[t:x28autobib_x22_Richard_Ax2e_KelseyA_Correspondence_between_Continuation_Passing_Style_and_Static_Single_Assignment_FormIn_Procx2e_Papers_from_the_1995_ACM_SIGPLAN_Workshop_on_Intermediate_Representationsx2c_IR_x2795x2c_ppx2e_13x2dx2d221995httpsx3ax2fx2fdoix2eorgx2f10x2e1145x2f202530x2e202532x22x29]{\AutobibLink{Kelsey}}~(\hyperref[t:x28autobib_x22_Richard_Ax2e_KelseyA_Correspondence_between_Continuation_Passing_Style_and_Static_Single_Assignment_FormIn_Procx2e_Papers_from_the_1995_ACM_SIGPLAN_Workshop_on_Intermediate_Representationsx2c_IR_x2795x2c_ppx2e_13x2dx2d221995httpsx3ax2fx2fdoix2eorgx2f10x2e1145x2f202530x2e202532x22x29]{\AutobibLink{1995}})}{'}s functions that translate SSA to and from the $\lambda$ calculus.

Overall, this paper identifies \texorpdfstring{\ensuremath{\ouro}}{unreachable} as an undefined
behavior whose semantics can, surprisingly, be captured via a simple set of rules.
Even better, these rules match how
compiler writers conceptualize optimization as local, context{-}insensitive rewrites.
In short, we see this paper as a first step towards a systematic demystification of
undefined behavior.

The remainder of the paper is organized as follows.
\ChapRefUC{\SectionNumberLink{t:x28part_x22secx3aexamplesx22x29}{2}}{The Essence of Unreachable, by Example} explains the essence of \texorpdfstring{\ensuremath{\ouro}}{unreachable}, and of our
calculus, through examples. \ChapRefUC{\SectionNumberLink{t:x28part_x22secx3acalculusx22x29}{3}}{The Essence of Unreachable, Formally} presents the
\texorpdfstring{\ensuremath{\ouro}}{unreachable} calculus while \ChapRef{\SectionNumberLink{t:x28part_x22secx3acorrectnessx2dcompilerx22x29}{4}}{The Correctness of the Compile{-}Time Reductions} proves its
correctness. \ChapRefUC{\SectionNumberLink{t:x28part_x22secx3aconventionalx2dtransformationsx22x29}{5}}{Extending the Calculus
with Additional Rules} discusses the
extension of the calculus with arbitrary semantics{-}preserving rules that
are useful when proving equivalences with the calculus.
\ChapRefUC{\SectionNumberLink{t:x28part_x22secx3aotherx2dubx22x29}{6}}{Beyond Unreachable Code} shows how our calculus can shed light to other
forms of undefined behavior beyond \texorpdfstring{\ensuremath{\ouro}}{unreachable} as evidence for our
approach{'}s value in understanding undefined behavior. Sections 7 and 8
connect our calculus to Racket and LLVM, respectively. The last two
sections discuss related work and conclude.

\sectionNewpage

\Ssection{The Essence of Unreachable, by Example}{The Essence of Unreachable, by Example}\label{t:x28part_x22secx3aexamplesx22x29}

From a compiler engineering perspective, the value of \texorpdfstring{\ensuremath{\ouro}}{unreachable} is that
it contributes to the clean separation between the different passes of the
compilation process. For example, a static analysis pass can deduce that
some part of a program is unreachable, and mark it as such with
\texorpdfstring{\ensuremath{\ouro}}{unreachable} to facilitate a subsequent pass that aims to optimize the
program. Given that the ability of compilers to reliably detect
unreachable code is limited, programmers can also take advantage of the
same mechanism to share with the compiler facts about whether some
parts of a program are unreachable. Once the programmer (or the analysis)
has communicated that some code is unreachable, the compiler can take
advantage of this information in various different ways.

The simplest such way is by treating \texorpdfstring{\ensuremath{\ouro}}{unreachable} as erroneous code, or more
generally, as code that never returns. Specifically, the compiler can erase
terminating, effect{-}free computations that happens before \texorpdfstring{\ensuremath{\ouro}}{unreachable}, and
all computation that is guaranteed to happen after it. For example, the
Racket compiler can simplify the expression

\begin{SCodeFlow}\begin{RktBlk}\begin{SingleColumn}\RktPn{(}\RktSym{begin}\mbox{\hphantom{\Scribtexttt{x}}}\RktPn{(}\RktSym{+}\mbox{\hphantom{\Scribtexttt{x}}}\RktSym{x}\mbox{\hphantom{\Scribtexttt{x}}}\RktVal{1}\RktPn{)}

\mbox{\hphantom{\Scribtexttt{xxxxxxx}}}\RktPn{(}\RktSym{unsafe{-}assert{-}unreachable}\RktPn{)}

\mbox{\hphantom{\Scribtexttt{xxxxxxx}}}\RktPn{(}\RktSym{+}\mbox{\hphantom{\Scribtexttt{x}}}\RktSym{y}\mbox{\hphantom{\Scribtexttt{x}}}\RktVal{2}\RktPn{)}\RktPn{)}\end{SingleColumn}\end{RktBlk}\end{SCodeFlow}

\noindent to just \RktPn{(}\RktSym{unsafe{-}assert{-}unreachable}\RktPn{)}, where
\RktPn{(}\RktSym{unsafe{-}assert{-}unreachable}\RktPn{)} is a syntactic form that programmers
use to declare that some part of an expression is
unreachable.\NoteBox{\NoteContent{Racket also offers \RktSym{assert{-}unreachable}, which
is guaranteed to throw an error when evaluated. Similar to
Racket, Rust also has two variants, one that allows
aggressive optimization and one that is safe, for testing
and debugging.}}

Treating \texorpdfstring{\ensuremath{\ouro}}{unreachable} as erroneous code, however, misses optimization
opportunities because it ignores the semantics of unreachable code. To
see how, consider the following Racket code:

\begin{SCodeFlow}\begin{RktBlk}\begin{SingleColumn}\RktPn{(}\RktSym{define}\mbox{\hphantom{\Scribtexttt{x}}}\RktPn{(}\RktSym{m{-}dist}\mbox{\hphantom{\Scribtexttt{x}}}\RktSym{p}\RktPn{)}

\mbox{\hphantom{\Scribtexttt{xx}}}\RktPn{(}\RktSym{match}\mbox{\hphantom{\Scribtexttt{x}}}\RktSym{p}

\mbox{\hphantom{\Scribtexttt{xxxx}}}\RktPn{[}\RktPn{(}\RktSym{cons}\mbox{\hphantom{\Scribtexttt{x}}}\RktSym{x}\mbox{\hphantom{\Scribtexttt{x}}}\RktSym{y}\RktPn{)}\mbox{\hphantom{\Scribtexttt{x}}}\RktPn{(}\RktSym{+}\mbox{\hphantom{\Scribtexttt{x}}}\RktPn{(}\RktSym{abs}\mbox{\hphantom{\Scribtexttt{x}}}\RktSym{x}\RktPn{)}\mbox{\hphantom{\Scribtexttt{x}}}\RktPn{(}\RktSym{abs}\mbox{\hphantom{\Scribtexttt{x}}}\RktSym{y}\RktPn{)}\RktPn{)}\RktPn{]}\RktPn{)}\RktPn{)}\end{SingleColumn}\end{RktBlk}\end{SCodeFlow}

\noindent The \RktSym{m{-}distance} function computes the Manhattan distance for a point represented
as a pair of (real) numbers. Its body consists of use of the
\RktSym{match} which roughly expands to a conditional and appropriate uses
of accessor functions:

\begin{SCodeFlow}\begin{RktBlk}\begin{SingleColumn}\RktPn{(}\RktSym{define}\mbox{\hphantom{\Scribtexttt{x}}}\RktPn{(}\RktSym{m{-}dist}\mbox{\hphantom{\Scribtexttt{x}}}\RktSym{p}\RktPn{)}

\mbox{\hphantom{\Scribtexttt{xx}}}\RktPn{(}\RktSym{if}\mbox{\hphantom{\Scribtexttt{x}}}\RktPn{(}\RktSym{pair{\hbox{\texttt{?}}}}\mbox{\hphantom{\Scribtexttt{x}}}\RktSym{p}\RktPn{)}

\mbox{\hphantom{\Scribtexttt{xxxxxx}}}\RktPn{(}\RktSym{+}\mbox{\hphantom{\Scribtexttt{x}}}\RktPn{(}\RktSym{abs}\mbox{\hphantom{\Scribtexttt{x}}}\RktPn{(}\RktSym{car}\mbox{\hphantom{\Scribtexttt{x}}}\RktSym{x}\RktPn{)}\RktPn{)}\mbox{\hphantom{\Scribtexttt{x}}}\RktPn{(}\RktSym{abs}\mbox{\hphantom{\Scribtexttt{x}}}\RktPn{(}\RktSym{cdr}\mbox{\hphantom{\Scribtexttt{x}}}\RktSym{x}\RktPn{)}\RktPn{)}\RktPn{)}

\mbox{\hphantom{\Scribtexttt{xxxxxx}}}\RktPn{(}\RktSym{error}\mbox{\hphantom{\Scribtexttt{x}}}\RktVal{{\textquotesingle}}\RktVal{match{-}failure}\RktPn{)}\RktPn{)}\RktPn{)}\end{SingleColumn}\end{RktBlk}\end{SCodeFlow}

In general the two accessor functions in the above snippet, \RktSym{car} and
\RktSym{cdr}, come with checks that defensively inspect their argument to
make sure it is a pair. However, the Racket
compiler can deduce that the uses of the
accessor functions in the snippet are guarded by the \RktSym{pair{\hbox{\texttt{?}}}} (on
a variable that isn{'}t modified). Hence, the compiler replaces them with
their unsafe variants:

\begin{SCodeFlow}\begin{RktBlk}\begin{SingleColumn}\RktPn{(}\RktSym{define}\mbox{\hphantom{\Scribtexttt{x}}}\RktPn{(}\RktSym{m{-}dist}\mbox{\hphantom{\Scribtexttt{x}}}\RktSym{p}\RktPn{)}

\mbox{\hphantom{\Scribtexttt{xx}}}\RktPn{(}\RktSym{if}\mbox{\hphantom{\Scribtexttt{x}}}\RktPn{(}\RktSym{pair{\hbox{\texttt{?}}}}\mbox{\hphantom{\Scribtexttt{x}}}\RktSym{p}\RktPn{)}

\mbox{\hphantom{\Scribtexttt{xxxxxx}}}\RktPn{(}\RktSym{+}\mbox{\hphantom{\Scribtexttt{x}}}\RktPn{(}\RktSym{abs}\mbox{\hphantom{\Scribtexttt{x}}}\RktPn{(}\RktSym{unsafe{-}car}\mbox{\hphantom{\Scribtexttt{x}}}\RktSym{p}\RktPn{)}\RktPn{)}\mbox{\hphantom{\Scribtexttt{x}}}\RktPn{(}\RktSym{abs}\mbox{\hphantom{\Scribtexttt{x}}}\RktPn{(}\RktSym{unsafe{-}cdr}\mbox{\hphantom{\Scribtexttt{x}}}\RktSym{p}\RktPn{)}\RktPn{)}\RktPn{)}

\mbox{\hphantom{\Scribtexttt{xxxxxx}}}\RktPn{(}\RktSym{error}\mbox{\hphantom{\Scribtexttt{x}}}\RktVal{{\textquotesingle}}\RktVal{match{-}failure}\RktPn{)}\RktPn{)}\RktPn{)}\end{SingleColumn}\end{RktBlk}\end{SCodeFlow}

And that{'}s how far the Racket compiler can go on its own.
However, if the author of the code knows that \RktSym{m{-}dist}
is only applied to points, i.e., pairs, then,
they can convey this information to the Racket compiler
by adding to the \RktSym{match} expression in the body of \RktSym{m{-}dist}
a catch{-}all case that uses
\RktSym{unsafe{-}assert{-}unreachable}, the Racket variant of \texorpdfstring{\ensuremath{\ouro}}{unreachable}
that has undefined behavior, giving the compiler the license to \emph{assume}
that it is never reached. Rust has a similar construct, \Scribtexttt{unreachable{\char`\_}unchecked},
which also has undefined behavior if it is ever reached.

\begin{SCodeFlow}\begin{RktBlk}\begin{SingleColumn}\RktPn{(}\RktSym{define}\mbox{\hphantom{\Scribtexttt{x}}}\RktPn{(}\RktSym{m{-}dist}\mbox{\hphantom{\Scribtexttt{x}}}\RktSym{p}\RktPn{)}

\mbox{\hphantom{\Scribtexttt{xx}}}\RktPn{(}\RktSym{match}\mbox{\hphantom{\Scribtexttt{x}}}\RktSym{p}

\mbox{\hphantom{\Scribtexttt{xxxx}}}\RktPn{[}\RktPn{(}\RktSym{cons}\mbox{\hphantom{\Scribtexttt{x}}}\RktSym{x}\mbox{\hphantom{\Scribtexttt{x}}}\RktSym{y}\RktPn{)}\mbox{\hphantom{\Scribtexttt{x}}}\RktPn{(}\RktSym{+}\mbox{\hphantom{\Scribtexttt{x}}}\RktPn{(}\RktSym{abs}\mbox{\hphantom{\Scribtexttt{x}}}\RktSym{x}\RktPn{)}\mbox{\hphantom{\Scribtexttt{x}}}\RktPn{(}\RktSym{abs}\mbox{\hphantom{\Scribtexttt{x}}}\RktSym{y}\RktPn{)}\RktPn{)}\RktPn{]}

\mbox{\hphantom{\Scribtexttt{xxxx}}}\RktPn{[}\RktSym{{\char`\_}}\mbox{\hphantom{\Scribtexttt{x}}}\RktPn{(}\RktSym{unsafe{-}assert{-}unreachable}\RktPn{)}\RktPn{]}\RktPn{)}\RktPn{)}\end{SingleColumn}\end{RktBlk}\end{SCodeFlow}

As in the steps as above, the Racket compiler
simplifies the body of \RktSym{m{-}dist} to:

\begin{SCodeFlow}\begin{RktBlk}\begin{SingleColumn}\RktPn{(}\RktSym{define}\mbox{\hphantom{\Scribtexttt{x}}}\RktPn{(}\RktSym{m{-}dist}\mbox{\hphantom{\Scribtexttt{x}}}\RktSym{p}\RktPn{)}

\mbox{\hphantom{\Scribtexttt{xx}}}\RktPn{(}\RktSym{if}\mbox{\hphantom{\Scribtexttt{x}}}\RktPn{(}\RktSym{pair{\hbox{\texttt{?}}}}\mbox{\hphantom{\Scribtexttt{x}}}\RktSym{p}\RktPn{)}

\mbox{\hphantom{\Scribtexttt{xxxxxx}}}\RktPn{(}\RktSym{+}\mbox{\hphantom{\Scribtexttt{x}}}\RktPn{(}\RktSym{abs}\mbox{\hphantom{\Scribtexttt{x}}}\RktPn{(}\RktSym{unsafe{-}car}\mbox{\hphantom{\Scribtexttt{x}}}\RktSym{p}\RktPn{)}\RktPn{)}\mbox{\hphantom{\Scribtexttt{x}}}\RktPn{(}\RktSym{abs}\mbox{\hphantom{\Scribtexttt{x}}}\RktPn{(}\RktSym{unsafe{-}cdr}\mbox{\hphantom{\Scribtexttt{x}}}\RktSym{p}\RktPn{)}\RktPn{)}\RktPn{)}

\mbox{\hphantom{\Scribtexttt{xxxxxx}}}\RktPn{(}\RktSym{unsafe{-}assert{-}unreachable}\RktPn{)}\RktPn{)}\RktPn{)}\end{SingleColumn}\end{RktBlk}\end{SCodeFlow}

Since, the else{-}branch of the conditional is marked as unreachable,
the Racket compiler can assume that in a program that uses
\RktSym{m{-}dist}, \RktSym{m{-}dist} consumes only pairs. Hence,
the compiler can transform the \RktSym{m{-}dist} to eliminate the
unreachable branch and replace it with a \RktSym{begin} expression:

\begin{SCodeFlow}\begin{RktBlk}\begin{SingleColumn}\RktPn{(}\RktSym{define}\mbox{\hphantom{\Scribtexttt{x}}}\RktPn{(}\RktSym{m{-}dist}\mbox{\hphantom{\Scribtexttt{x}}}\RktSym{p}\RktPn{)}

\mbox{\hphantom{\Scribtexttt{xx}}}\RktPn{(}\RktSym{begin}\mbox{\hphantom{\Scribtexttt{x}}}\RktPn{(}\RktSym{pair{\hbox{\texttt{?}}}}\mbox{\hphantom{\Scribtexttt{x}}}\RktSym{p}\RktPn{)}

\mbox{\hphantom{\Scribtexttt{xxxxxxxxx}}}\RktPn{(}\RktSym{+}\mbox{\hphantom{\Scribtexttt{x}}}\RktPn{(}\RktSym{abs}\mbox{\hphantom{\Scribtexttt{x}}}\RktPn{(}\RktSym{unsafe{-}car}\mbox{\hphantom{\Scribtexttt{x}}}\RktSym{p}\RktPn{)}\RktPn{)}\mbox{\hphantom{\Scribtexttt{x}}}\RktPn{(}\RktSym{abs}\mbox{\hphantom{\Scribtexttt{x}}}\RktPn{(}\RktSym{unsafe{-}cdr}\mbox{\hphantom{\Scribtexttt{x}}}\RktSym{p}\RktPn{)}\RktPn{)}\RktPn{)}\RktPn{)}\RktPn{)}\end{SingleColumn}\end{RktBlk}\end{SCodeFlow}

In short, the Racket compiler applies one of the fundamental rules of \texorpdfstring{\ensuremath{\ouro}}{unreachable}:

\begin{SCodeFlow}\begin{RktBlk}\begin{SingleColumn}\RktPn{(}\RktSym{if}\mbox{\hphantom{\Scribtexttt{x}}}\RktSym{e{\char`\_}1}\mbox{\hphantom{\Scribtexttt{x}}}\RktSym{e{\char`\_}2}\mbox{\hphantom{\Scribtexttt{x}}}\RktPn{(}\RktSym{unsafe{-}assert{-}unreachable}\RktPn{)}\RktPn{)}

\RktSym{=}

\RktPn{(}\RktSym{begin}\mbox{\hphantom{\Scribtexttt{x}}}\RktSym{e{\char`\_}1}\mbox{\hphantom{\Scribtexttt{x}}}\RktSym{e{\char`\_}2}\RktPn{)}\end{SingleColumn}\end{RktBlk}\end{SCodeFlow}

The use of the rule unlocks one more opportunity for simplifying
\RktPn{(}\RktSym{m{-}dist}\RktPn{)}. Since \RktSym{pair{\hbox{\texttt{?}}}} has no side effects,
the Racket compiler drops it to obtain the final version of \RktSym{m{-}dist}:

\begin{SCodeFlow}\begin{RktBlk}\begin{SingleColumn}\RktPn{(}\RktSym{define}\mbox{\hphantom{\Scribtexttt{x}}}\RktPn{(}\RktSym{m{-}dist}\mbox{\hphantom{\Scribtexttt{x}}}\RktSym{p}\RktPn{)}

\mbox{\hphantom{\Scribtexttt{xx}}}\RktPn{(}\RktSym{+}\mbox{\hphantom{\Scribtexttt{x}}}\RktPn{(}\RktSym{abs}\mbox{\hphantom{\Scribtexttt{x}}}\RktPn{(}\RktSym{unsafe{-}car}\mbox{\hphantom{\Scribtexttt{x}}}\RktSym{p}\RktPn{)}\RktPn{)}\mbox{\hphantom{\Scribtexttt{x}}}\RktPn{(}\RktSym{abs}\mbox{\hphantom{\Scribtexttt{x}}}\RktPn{(}\RktSym{unsafe{-}cdr}\mbox{\hphantom{\Scribtexttt{x}}}\RktSym{p}\RktPn{)}\RktPn{)}\RktPn{)}\RktPn{)}\end{SingleColumn}\end{RktBlk}\end{SCodeFlow}

Surprisingly, the two examples in this section are sufficient to describe the full
 essence of the compile{-}time treatment of \texorpdfstring{\ensuremath{\ouro}}{unreachable}: either compilers
 treat \texorpdfstring{\ensuremath{\ouro}}{unreachable} same as an error (first example), or as the
 justification for simplifying the branches of conditionals (second
 example). The following section gives a formal account of this insight
 with the unreachable calculus and its compile{-}time reductions.

\sectionNewpage

\Ssection{The Essence of Unreachable, Formally}{The Essence of Unreachable, Formally}\label{t:x28part_x22secx3acalculusx22x29}

Our calculus extends the call{-}by{-}value $\lambda$ calculus with a
construct \texorpdfstring{\ensuremath{\ouro}}{unreachable} and \texorpdfstring{\ensuremath{\ouro}}{unreachable}{-}specific term rewriting rules.
Specifically, there are two groups of such rules: the first, which extends the standard
reduction of the call{-}by{-}value $\lambda$ calculus, corresponds
to the run time behavior of  \texorpdfstring{\ensuremath{\ouro}}{unreachable}; the second corresponds
to legal compile{-}time transformations of \texorpdfstring{\ensuremath{\ouro}}{unreachable}.

\Ssubsection{The Unreachable Calculus and its Standard Reduction}{The Unreachable Calculus and its Standard Reduction}\label{t:x28part_x22Thex5fUnreachablex5fCalculusx5fandx5fitsx5fStandardx5fReductionx22x29}

\begin{Figure}\begin{Centerfigure}\begin{FigureInside}\Syntax\end{FigureInside}\end{Centerfigure}

\Centertext{\Legend{\FigureTarget{\label{t:x28counter_x28x22figurex22_x22figx3asyntaxx22x29x29}\textsf{Fig.}~\textsf{2}. }{t:x28counter_x28x22figurex22_x22figx3asyntaxx22x29x29}\textsf{The Syntax of the \texorpdfstring{\ensuremath{\ouro}}{unreachable} Calculus}}}\end{Figure}

\begin{Figure}\begin{Centerfigure}\begin{FigureInside}\StandardReductions\end{FigureInside}\end{Centerfigure}

\Centertext{\Legend{\FigureTarget{\label{t:x28counter_x28x22figurex22_x22figx3astandardx2dreductionx22x29x29}\textsf{Fig.}~\textsf{3}. }{t:x28counter_x28x22figurex22_x22figx3astandardx2dreductionx22x29x29}\textsf{The Standard Reduction of the \texorpdfstring{\ensuremath{\ouro}}{unreachable} Calculus}}}\end{Figure}

Figure~\hyperref[t:x28counter_x28x22figurex22_x22figx3asyntaxx22x29x29]{\FigureRef{2}{t:x28counter_x28x22figurex22_x22figx3asyntaxx22x29x29}} presents the syntax of the unreachable calculus.
Expressions consist of variables, constants, $\lambda$
expressions, binary operators, conditionals, sequencing, and a
family of different errors (indexed by $k$). The only unique syntactic
element is \texorpdfstring{\ensuremath{\ouro}}{unreachable}, which represents annotations that indicate
unreachable code like Racket{'}s \RktSym{unsafe{-}assert{-}unreachable}.

Figure~\hyperref[t:x28counter_x28x22figurex22_x22figx3astandardx2dreductionx22x29x29]{\FigureRef{3}{t:x28counter_x28x22figurex22_x22figx3astandardx2dreductionx22x29x29}} describes the standard
reduction of the calculus. It starts with a standard
definition of evaluation contexts that ensure a
left{-}to{-}right order of evaluation. The $v$ non{-}terminal
describes values and $a$ final answers. The remainder of
the figure defines the notions of reduction $\arrs$, whose
compatible closure over evaluation contexts is the core of the
standard reduction $\arrS$ ({``}s{''} for standard).

The notions of reduction of the calculus are mostly as expected.  Similar
to Racket, the \texorpdfstring{\ensuremath{\ouro}}{unreachable} calculus employs {``}truthy{''} values in its
conditionals. Hence, rule S.1 reduces an $\Ifk$ expression to
its else{-}branch if the test position reduces to the value
$\False$. Any other value causes the conditional to reduce to
its then{-}branch (rule S.2); we use $\equiv$ for syntactic equality, here
and throughout the paper. The $\arrs$ relation also includes
$\beta${-}value (rule S.3) and the usual sequencing rule that discards the first
sub{-}expression of a $\Seqk$ expression when it is a value  (rule S.4).
Rule S.5 defers to the $\delta$ function to handle calls to primitive
operators.

Rules S.9 and S.10 introduce errors. When a function application involves
a value other than a $\lambda$ expression in the function position, it reduces to a
specific error with the label $\beta$. Similarly, when a primitive operator
encounters arguments that do not make sense, it reduces to an error with
the label $\delta$. We equip the calculus with a family of errors in order to
account for the common linguistic setting where there are multiple,
semantically distinct types of side{-}effects besides non{-}termination.

Rules S.6 and S.8 form a typical definition of the standard reduction $\arrS$
in a language with errors. Rule S.6 lifts the
notions of reduction  to evaluation contexts. Rule S.8 discards the
evaluation context around an error to terminate the reduction.

Rule S.7 is the only rule that is specific to \texorpdfstring{\ensuremath{\ouro}}{unreachable}. It  treats
\texorpdfstring{\ensuremath{\ouro}}{unreachable} the same as an error. The rule is based on Rust{'}s and Racket{'}s
safe variants of \texorpdfstring{\ensuremath{\ouro}}{unreachable} that are supposed to be used to check
assumptions about code unreachability during debugging. After all, from
the perspective of a compiler, which is the perspective we examine in this
paper, programs that evaluate \texorpdfstring{\ensuremath{\ouro}}{unreachable} are plain wrong, and hence
evaluating \texorpdfstring{\ensuremath{\ouro}}{unreachable} should result in some form of error.  Furthermore,
as we discuss further in \ChapRef{\SectionNumberLink{t:x28part_x22secx3acorrectnessx2dcompilerx22x29}{4}}{The Correctness of the Compile{-}Time Reductions}, this run
time behavior of \texorpdfstring{\ensuremath{\ouro}}{unreachable} plays an important role in  the compositional
proof of the correctness of the compile{-}time tranformations that the
unreachable calculus captures.

\Ssubsection{Compile{-}Time Transformations as Reduction Relations}{Compile{-}Time Transformations as Reduction Relations}\label{t:x28part_x22compilerx2dtranformationsx2dreductionx22x29}

\begin{Figure}\begin{Centerfigure}\begin{FigureInside}\CompilerRulesContext\end{FigureInside}\end{Centerfigure}

\Centertext{\Legend{\FigureTarget{\label{t:x28counter_x28x22figurex22_x22figx3acontextx22x29x29}\textsf{Fig.}~\textsf{4}. }{t:x28counter_x28x22figurex22_x22figx3acontextx22x29x29}\textsf{The Compile{-}time Relation (and the definition of contexts)}}}\end{Figure}

\begin{Figure}\begin{Centerfigure}\begin{FigureInside}\CompilerRulesP \vspace{0.75em}\end{FigureInside}\end{Centerfigure}

\Centertext{\Legend{\FigureTarget{\label{t:x28counter_x28x22figurex22_x22figx3acompilerrulePx22x29x29}\textsf{Fig.}~\textsf{5}. }{t:x28counter_x28x22figurex22_x22figx3acompilerrulePx22x29x29}\textsf{The Compile{-}time Unreachable Propagation Relations}}}\end{Figure}

\begin{Figure}\begin{Centerfigure}\begin{FigureInside}\CompilerRulesU \vspace{0.75em}\end{FigureInside}\end{Centerfigure}

\Centertext{\Legend{\FigureTarget{\label{t:x28counter_x28x22figurex22_x22figx3acompilerruleUx22x29x29}\textsf{Fig.}~\textsf{6}. }{t:x28counter_x28x22figurex22_x22figx3acompilerruleUx22x29x29}\textsf{The Compile{-}time Unreachable Undefined Behavior Relations}}}\end{Figure}

Figure~\hyperref[t:x28counter_x28x22figurex22_x22figx3acontextx22x29x29]{\FigureRef{4}{t:x28counter_x28x22figurex22_x22figx3acontextx22x29x29}} shows the relation $\arrC$
({``}c{''} for compile), which captures compile{-}time transformations
due to \texorpdfstring{\ensuremath{\ouro}}{unreachable}. This relation does not define a program
transformation algorithm that a compiler might use; instead
it describes the set of valid basic transformations steps
that a compiler is allowed to perform. In other words, the
relation is a specification of how a correct compiler may
take advantage of \texorpdfstring{\ensuremath{\ouro}}{unreachable} in order to simplify a program.

The $\arrC$ relation is broken down into two pieces:
$\arrP$ ({``}p{''} for propagate), which captures
transformations of \texorpdfstring{\ensuremath{\ouro}}{unreachable} that are legal for a safe
variant of \texorpdfstring{\ensuremath{\ouro}}{unreachable}, namely when it behaves
like an error, and $\arrU$ ({``}u{''} for undefined), which
captures the undefined behavior of \texorpdfstring{\ensuremath{\ouro}}{unreachable}, namely how
it interacts with conditional expressions.
Figure~\hyperref[t:x28counter_x28x22figurex22_x22figx3acontextx22x29x29]{\FigureRef{4}{t:x28counter_x28x22figurex22_x22figx3acontextx22x29x29}} also gives the definition of
contexts, which are used in the definitions of both
$\arrP$ and $\arrU$.

Figure~\hyperref[t:x28counter_x28x22figurex22_x22figx3acompilerrulePx22x29x29]{\FigureRef{5}{t:x28counter_x28x22figurex22_x22figx3acompilerrulePx22x29x29}} presents $\arrP$. Rules
P.1 to P.5 in figure~\hyperref[t:x28counter_x28x22figurex22_x22figx3acompilerrulePx22x29x29]{\FigureRef{5}{t:x28counter_x28x22figurex22_x22figx3acompilerrulePx22x29x29}} allow the
compiler to propagate \texorpdfstring{\ensuremath{\ouro}}{unreachable} downstream and upstream
in a compound expression. Specifically, rules P.2 and P.4
enable the compiler to eliminate all expressions that follow
\texorpdfstring{\ensuremath{\ouro}}{unreachable}. In essence, they capture the idea that all
computation downstream of unreachable code is also
unreachable since computation never progresses past
\texorpdfstring{\ensuremath{\ouro}}{unreachable}. Similarly, the compiler can eliminate all
expressions that precede \texorpdfstring{\ensuremath{\ouro}}{unreachable}, as long as they are
safe. Roughly, safe expressions are those that evaluate to a
value, i.e., they have no side{-}effects:
\begin{definition}[safety]\label{def:safety}An expression $e$ is \emph{safe}, written as $\Lsafep(e)$,
if for all closing substitutions $\vartheta$,
we have $\vartheta(e) \arrS^* v$.\end{definition}
\noindent Given this definition of safe expressions, rule
P.1 embodies how a compiler benefits
from the combination of safe expressions and \texorpdfstring{\ensuremath{\ouro}}{unreachable}.
It allows \texorpdfstring{\ensuremath{\ouro}}{unreachable} to {``}eat{''} safe expressions upstream.
After all, if the value{-}result of a safe expression is not
used, then it can be discarded without affecting the rest of
the evaluation. However, rule P.1 is restricted and operates
only on $\Seqk$ expressions. To mitigate this restriction,
rules P.3 and P.5 reshape expressions that contain
\texorpdfstring{\ensuremath{\ouro}}{unreachable} into $\Seqk$ expressions.
Since a compiler should be able to perform transformations
in any context, the Ctx.P rule lifts all of the $\arrp$
rules to arbitrary contexts. Also, all of the transformations
in $\arrp$ are legal in either direction, so CtxSym.P adds
in the symmetric variants.

Figure~\hyperref[t:x28counter_x28x22figurex22_x22figx3acompilerruleUx22x29x29]{\FigureRef{6}{t:x28counter_x28x22figurex22_x22figx3acompilerruleUx22x29x29}} presents $\arrU$. Rules
U.1 and U.2 in figure~\hyperref[t:x28counter_x28x22figurex22_x22figx3acompilerruleUx22x29x29]{\FigureRef{6}{t:x28counter_x28x22figurex22_x22figx3acompilerruleUx22x29x29}} define the
essential notions of reduction and they
capture the undefined behavior of \texorpdfstring{\ensuremath{\ouro}}{unreachable},
matching the behavior of Racket{'}s
\RktSym{unsafe{-}assert{-}unreachable} and Rust{'}s
\Scribtexttt{unreachable{\char`\_}unchecked}. They specify how a compiler can
eliminate unreachable branches of conditionals. In other
words, they are the formal counterparts of the essential
transformation steps that the Racket compiler performs in
the examples from \ChapRef{\SectionNumberLink{t:x28part_x22secx3aexamplesx22x29}{2}}{The Essence of Unreachable, by Example}.

The $\arru$ axioms are sound in any context but, unlike
$\arrP$, they are not sound in reverse. The issue is that
the $R_p$ notions of reductions eliminate immaterial
pieces of an expression. However, a symmetric $\arrU$ reduction would
inject arbitrary expressions back, possibly affecting the
meaning of the expression. For instance, the compiler is
only justified to transform $\If{e}{3}{\ouro}$ to
$\Seq{e}{3}$ because it assumes that this
\texorpdfstring{\ensuremath{\ouro}}{unreachable} is indeed unreachable. Reversing the reduction,
though, introduces an \texorpdfstring{\ensuremath{\ouro}}{unreachable} for which the compiler
knows nothing about. This asymmetry is the source of the
main challenge for proving the correctness of the
compile{-}time reductions, and we revisit the issue in
\ChapRef{\SectionNumberLink{t:x28part_x22secx3acorrectnessx2dcompilerx22x29}{4}}{The Correctness of the Compile{-}Time Reductions}.

\sectionNewpage

\Ssection{The Correctness of the Compile{-}Time Reductions}{The Correctness of the Compile{-}Time Reductions}\label{t:x28part_x22secx3acorrectnessx2dcompilerx22x29}

The simplicity of the compile{-}time reductions of the unreachable calculus
comes with a price: establishing their correctness is challenging.  The root
of the challenge is that unreachable is a kind of undefined behavior, and
this affects radically what the correctness of the reductions means.  In
general, the ultimate correctness criterion of any program transformation
is that it preserves the meaning of programs. However, transformations
that take advantage of undefined behavior, such as those captured by
$\arrC$, are supposed to preserve the meaning only of programs
that do not exhibit undefined behavior; the rest are outside the universe
of programs a compiler should handle.  In other words, in a program that
exhibits undefined behavior, $\arrC$ has the liberty to alter the
program{'}s behavior. It is this liberty that impedes a compositional proof of correctness
for the reductions. Specifically, a compositional proof requires reasoning
about whether a compiler preserves the meaning of a piece of a
program in isolation from the rest of the program, but undefined
behavior can only be determined for a whole program.

Fortunately, \texorpdfstring{\ensuremath{\ouro}}{unreachable} is a rather {``}well{-}behaved{''} undefined behavior.
In particular, the small set of intuitive compile{-}time reductions from
\ChapRef{\SectionNumberLink{t:x28part_x22secx3acalculusx22x29}{3}}{The Essence of Unreachable, Formally} have the property that they preserve the meaning
of any expression in any context under certain conditions.  In turn, these
conditions can be described with a nonstandard, novel logical relation
that hides the complexity of the proof of correctness, keeping the syntax
and the reductions of the calculus simple.

To get a sense of the correctness theorem we are aiming for,
consider two complete programs $e$ and $e'$ such that
$e \arrC e'$, i.e., where $e'$ is a program that the compiler
is allowed to transform $e$ into. Accordingly, if the evaluation of $e$ does not exhibit
undefined behavior, we wish to ensure that $e'$ does not
exhibit undefined behavior either and that $e'$ and $e$
either both diverge or both terminate with the same result.

Translating this informal description into a formal
statement requires making the notion of undefined behavior
precise, using the standard reduction relation:

\begin{definition}[Undefined Behavior]\label{def:undef}The behavior of a closed expression $e$ is \emph{undefined}, written as $\Lundefp(e)$,
if $e \arrS^* \ouro$.\end{definition}
\noindent
which leads to the formal definition of compiler correctness:

\begin{proposition}\label{prop:compiler-correctness}  For all closed expressions $e, e'$ such that $e\arrC^* e'$,
if $\neg\undefp{e}$ then

\noindent \begin{itemize}\atItemizeStart

\item $\neg\undefp{e'}$,

\item $\forall c. \, (e\arrS^* c \iff e'\arrS^* c$),

\item $(\exists e_1.\ e\arrS^* \Func{x}{e_1}) \iff (\exists e_1'.\ e'\arrS^* \Func{x}{e_1'})$, and

\item $\forall k. \, (e\arrS^* \error{k} \iff e'\arrS^* \error{k}$).\end{itemize}

\noindent \begin{remark*}Because both $e$ and $e'$ do not terminate with $\ouro$,
the conclusion of \cref{prop:compiler-correctness} implies
that non{-}termination is preserved by $\arrC$;
$e'$ diverges if and only if $e$ diverges.\end{remark*}\end{proposition}

The standard, logical{-}relations based approach to proving
\cref{prop:compiler-correctness} is to define a step{-}indexed
syntax{-}based binary relation that uses the standard reduction to reduce
expressions to values. This logical relation captures a notion of logical
approximation such that two expressions that approximate each other are
contextually equivalent. Soundness of the logical relation with respect to
contextual equivalence is established by showing its Fundamental Property,
i.e., any expression $e$, if $\neg\Lundefp(e)$,
then $e$ is related to itself.

Once that is done, one would attempt to prove the  compile{-}time reductions
correct by showing that the left{-}hand side and the right{-}hand side of each
reduction rule logically approximate each other or, more precisely,
if $\neg\Lundefp(e)$ and $e \arrC e'$, then $e$ is related to  $e'$
and $e'$ is related to $e$.

Unfortunately, this approach does not work because of the
way \Lundefp$~$ is defined. Consider the situation where we
know that some application expression
$\App{e_s}{e_s^\prime}$ that reduces via $\arrC$ to
$\App{e_c}{e_c^\prime}$ and we wish to show that
$\App{e_c}{e_c^\prime}$ is related to $\App{e_s}{e_s^\prime}$.
In this case, we do not benefit from the
$\neg\Lundefp(\App{e_s}{e_s^\prime})$ assumption.
In essence, $\neg\undefp{\App{e_s}{e_s^\prime}}$ does not translate to
facts about the pieces of $\App{e_s}{e_s^\prime}$ that we can map to the
the pieces of $\App{e_c}{e_c^\prime}$ to complete
the proof inductively.

\begin{Figure}\begin{Centerfigure}\begin{FigureInside}\[
\begin{array}{@{}l@{}}
\begin{array}{ccl}\ensuremath{f_\mathrm{src}} &:=& \Func{p x}{(994 + \If{(\App{p}{x})}{\ouro}{x})}
\\
\ensuremath{f_\mathrm{opt}} &:=& \Func{p x}{(994 + \Seq{(\App{p}{x})}{x})}\end{array}
\vspace{0.3em}
\\
\begin{array}{lll}
\App{\App{\ensuremath{f_\mathrm{src}}}{(\Func{y}{\False})}}{n} \arrS^* 994+n & &
\App{\App{\ensuremath{f_\mathrm{src}}}{(\Func{y}{\True})}}{n} \arrS^* \ouro
\\
\App{\App{\ensuremath{f_\mathrm{opt}}}{(\Func{y}{\False})}}{n} \arrS^* 994+n & \quad &
\App{\App{\ensuremath{f_\mathrm{opt}}}{(\Func{y}{\True})}}{n} \arrS^* 994+n
 \end{array}
\end{array}
\]

$\vspace*{-0.1in}$\end{FigureInside}\end{Centerfigure}

\Centertext{\Legend{\FigureTarget{\label{t:x28counter_x28x22figurex22_x22figx3apxx2dsrcoptx22x29x29}\textsf{Fig.}~\textsf{7}. }{t:x28counter_x28x22figurex22_x22figx3apxx2dsrcoptx22x29x29}\textsf{Examples of Semantics{-}Altering Transformations Due to $\arrU$ Reductions.}}}\end{Figure}

Functions  \ensuremath{f_\mathrm{src}}, \ensuremath{f_\mathrm{opt}}
in Figure~\hyperref[t:x28counter_x28x22figurex22_x22figx3apxx2dsrcoptx22x29x29]{\FigureRef{7}{t:x28counter_x28x22figurex22_x22figx3apxx2dsrcoptx22x29x29}} demonstrate the issue.
Consider a pair of programs
$(e,e') :\equiv (\App{\App{\ensuremath{f_\mathrm{src}}}{(\Func{y}{b})}}{n},
\App{\App{\ensuremath{f_\mathrm{opt}}}{(\Func{y}{b})}}{n})$.
Given that $e'\arrS^* 994+n$ for all $n$ and $b$, we would like to
be able to use the logical relation to deduce that
$e\arrS^* 994+n$
based entirely on the fact that $(\ensuremath{f_\mathrm{src}},\ensuremath{f_\mathrm{opt}})$ are related.
But as figure~\hyperref[t:x28counter_x28x22figurex22_x22figx3apxx2dsrcoptx22x29x29]{\FigureRef{7}{t:x28counter_x28x22figurex22_x22figx3apxx2dsrcoptx22x29x29}} demonstrates,
$e$ exhibits undefined behavior  when $b\equiv\False$.
In fact, the conditions under which $e$ exhibits undefined behavior
become involved if we consider an arbitrary argument $p$
that can exhibit undefined behavior:
\[
\begin{array}{r@{}c@{}lll}
\App{\App{\ensuremath{f_\mathrm{src}}}{\left(
  \Func{y}{\vphantom{\int}
   \If{(y\tteq 0)}{\ouro}{b}}
  \right)}}{n}
&\;\arrS^*\;& 994+n
& \iff &
b\equiv\False\land n\not\equiv 0 \\
\App{\App{\ensuremath{f_\mathrm{opt}}}{(\Func{y}{\Seq{(y\tteq 0)}{b}})}}{n} &\arrS^*& 994+n
\end{array}
\]
Put differently,
the assumption $\neg\undefp{e}$ has turned into
an arbitrary constraint on the arguments of the two
functions, which is unclear how to incorporate in a logical relation.

The literature on logical relations suggests a way around the problem. One
can approximate a global property, such as \Lundefp, with an
inductively{-}defined approximate predicate. For example,
RustBelt\Autobibref{~(\hyperref[t:x28autobib_x22Ralf_Jungx2c_Jacquesx2dHenri_Jourdanx2c_Robbert_Krebbersx2c_and_Derek_DreyerRustBeltx3a_securing_the_foundations_of_the_Rust_programming_languageProceedings_of_the_ACM_on_Programming_Languages_x28POPLx29_2x2c_ppx2e_66x3a1x2dx2d66x3a342018httpsx3ax2fx2fdoix2eorgx2f10x2e1145x2f3158154x22x29]{\AutobibLink{Jung et al\Sendabbrev{.}}} \hyperref[t:x28autobib_x22Ralf_Jungx2c_Jacquesx2dHenri_Jourdanx2c_Robbert_Krebbersx2c_and_Derek_DreyerRustBeltx3a_securing_the_foundations_of_the_Rust_programming_languageProceedings_of_the_ACM_on_Programming_Languages_x28POPLx29_2x2c_ppx2e_66x3a1x2dx2d66x3a342018httpsx3ax2fx2fdoix2eorgx2f10x2e1145x2f3158154x22x29]{\AutobibLink{2018}})} takes advantage of Rust{'}s type system and
ensures the absence of undefined behavior based
on a semantic type judgment.  In other words, such predicates aim to break the
whole{-}program property into compositional facts about the structural
pieces of an expression. Unfortunately, such an approach does not work here in a
straightforward manner because \Lundefp$\relax$ simply is not compositional.

Instead of attempting to come up with some conservative inductive
substitute for \Lundefp, we follow a different path that aims to
keep the calculus as close to the way compiler writers use the
full{-}blown, non{-}compositional definition of undefined behavior when
reasoning about compiler transformations. We define
forward{-} and backward{-}aproximation logical relations that, when an expression does not
exhibit undefined behavior, relate it with its transformed version after a
sequence of $\arrC$ reductions, and vice versa. In fact, we define four
such relations: two for $\arrU$ and two for $\arrP$ reductions.
Uncharacteristically, the forward and backward approximations for
$\arrU$ reductions do not mirror each other. This complicates
showing their soundness but achieves the goal of
hiding all complexity away from the calculus.

Before delving into the details of the logical relations, we first introduce
well{-}formedness judgments for managing free variables
since the fundamental properties for the logical relations assume open expressions.
Figure~\hyperref[t:x28counter_x28x22figurex22_x22figx3awellformednessx22x29x29]{\FigureRef{8}{t:x28counter_x28x22figurex22_x22figx3awellformednessx22x29x29}} presents a few selected inference rules. The
complete list of wellformedness rules can be found on Appendix B.
In the figure, $\GD$ and $\GD'$ denote sets of variables.
The judgment $\wfe{\GD}{e}$ holds if
the set of free variables in the expression $e$ is a subset of $\GD$.
The judgment $\wfctx{\GD'}{\GD}{C}$ asserts that the context $C$
maps an expression that refers to variables in $\GD$
into an expression that refers to variables in $\GD'$.
Put differently, if $\wfe{\GD}{e}$ and $\wfctx{\GD'}{\GD}{C}$
then $\wfe{\GD'}{C[e]}$.

\begin{Figure}\begin{Centerfigure}\begin{FigureInside}\begin{minipage}{0.95\textwidth} \centering
\begin{minipage}{0.27\textwidth} \centering
\WellformedTermsSelectedRules
\end{minipage}
\begin{minipage}{0.01\textwidth}
\hspace{0.01\textwidth}
\end{minipage}
\begin{minipage}{0.67\textwidth} \centering
\WellFormedContextsSelectedRules
\end{minipage}
\end{minipage}\end{FigureInside}\end{Centerfigure}

\Centertext{\Legend{\FigureTarget{\label{t:x28counter_x28x22figurex22_x22figx3awellformednessx22x29x29}\textsf{Fig.}~\textsf{8}. }{t:x28counter_x28x22figurex22_x22figx3awellformednessx22x29x29}\textsf{Well{-}formed Expressions and Contexts (Selected Rules)}}}\end{Figure}

The remainder of this section establishes the correctness of
the $\arrC$ reductions. The first two subsections describe
the forward and backward approximation logical relations for
$\arrU$ reductions and their soundness. Then, the final
subsection discusses the correctness of the $\arrP$
reductions to complete the proof of correctness of our
compile{-}time reductions.

\Ssubsection{Forward Approximation of  $\arrU$ Reductions}{Forward Approximation of  $\arrU$ Reductions}\label{t:x28part_x22secx3aforwardx2drelx22x29}

As a first step to prove the correctness of $\arrU$ reductions, we
design the binary \nameref{def:forward-rel}. In detail, for
all $i \ge 0$, we define the step{-}indexed logical relation for
values and  \emph{closed} expressions as $\relV_i^{\arrU}$ and
$\relE_i^{\arrU}$, respectively. In the definition, the notation
$e\arrS^j e'$ means that $e$ reduces to $e'$ using exactly $j$
steps under the standard reduction of the calculus.

\begin{definition}[Forward-Approximation Logical Relation]\label{def:forward-rel}
\[
\begin{array}{r@{}c@{}l}
 \relV_i^{\arrU} & = &
\{   (\Func{x}{e_1},\Func{x}{e_2}) \;|\;
 \forall j<i. \;
 \forall v_1\, v_2.\; (v_1,v_2)\in \relV_j^{\arrU},
 (\subst{e_1}{x}{v_1},\subst{e_2}{x}{v_2})\in \relE_j^{\arrU}
 \}
\\
& & \; \cup \; \lset{ (c,c) }\\
\relE_i^{\arrU} & \;=\; &
\begin{array}[t]{@{}c@{}c@{}l@{}}
\{ (e_1,e_2) & \;|\; &
 \forall j < i. \;
 \forall a_1. \; (\neg\undefp{e_1}) \land e_1 \arrS^j a_1 \Implies \\
 & & \quad \quad % \quad \quad \quad
 \forall e_2'.\; e_2 \arrU^* e_2' \Implies \\
 & & \quad \quad \quad \quad % \quad \quad
 \exists a_2.\; e_2' \arrS^* a_2 \land
 % \\ & & \quad \quad \quad \quad \quad \quad \quad
 ((a_1,a_2)\in \relV_{i-j}^{\arrU}\;
 \lor \exists k.\, a_1\equiv a_2\equiv \error{k}) \}
\end{array}\\
  \end{array}
\]
 \begin{remark*}The definition of $\relV^{\arrU}_i$ uses $\relV^{\arrU}_j$ and
$\relE^{\arrU}_j$ for $j$ that is strictly less than $i$, and that
the definition of $\relE^{\arrU}_i$ uses $\relV^{\arrU}_{i-j}$
for $0\le j< i$. Thus the mutually{-}referential relations are well{-}founded.\end{remark*}\end{definition}

The value relation $\relV_i^{\arrU}$ is standard.  Two values
$(v_1,v_2)$ are related by $\relV_i^{\arrU}$ at step $i\ge 0$ if
$v_1$ and $v_2$ are the same constant $c$, or if they are
both lambdas and, for all $j < i$,
applying arguments that are related at $j$ steps produces expressions related at
$j$ steps.

The expression relation $\relE_i^{\arrU}$, in contrast, is not
standard. It is crafted to resemble one of the directions of
\cref{prop:compiler-correctness}, but specialized to $\arrU$
reductions. Specifically, under the assumption that $\neg\undefp{e_1}$
and $e_1$ evaluates to an answer after $j<i$ standard
reduction steps, $e_1$ is related to $e_2$ at step index $i$ if all
expressions $e_2^\prime$ that are results of simplifying $e_2$ with
$\arrU$ reductions evaluate to answers that are related to $e_1${'}s answer.
Two answers are related if they are the same error, or if they are related
values at step index $i-j$.

In other words, the expression approximation
aims to establish directly that after some $\arrU$ reductions the resulting
expression approximates the meaning of the initial one. For that reason,
the expression approximation has two built{-}in assumptions that are not
found in standard logical relations: $\neg\undefp{e_1}$ and $e_2
\arrU^* e_2^\prime$. As discussed above, in a standard logical relation the
first would be an assumption for its Fundamental Property, while the
second would be an assumption of a separate theorem that uses the
soundness of the logical relation to prove that $\arrU$ reductions are
correct.

Equipped with these relations, we are ready to prove the
\nameref{lem:forward-fundamental}. As usual, it states that any open
expression $e$ is related to itself:

\begin{lemma}[Fundamental Property of $\relE^{\arrU}$]\label{lem:forward-fundamental}For all $e,\, i\ge 0,\, \GD$ and $\Gg$,
if $\wfe{\GD}{e}$ and $\Gg\in\relG_i^{\arrU}[\GD]$ then
$(\Gg_1(e), \Gg_2(e))\in \relE_i^{\arrU}$.

\begin{remark*}$\relG_i^{\arrU}$ is the standard relation on pairs of substitutions that
map the same identifier to values related by $\relV_i^{\arrU}$.
The complete definition can be found on page 4 of Appendix E.\end{remark*}\end{lemma}

\begin{proof}[Proof Sketch]We prove the Fundamental Property by induction on $e$.
In most cases, the transformation $\Gg_2(e)\arrU^* e_2'$ does not
change the outermost shape of $\Gg_2(e)$ and hence the proofs
are straightforward. When $e:\equiv\If{e_c}{\ouro}{e_f}$,
the branch{-}elimination transformation may reduce $\Gg_2(\If{e_c}{\ouro}{e_f})$
to $\Seq{\Gg_2(e_c)}{\Gg_2(e_f)}$.\Footnote{\FootnoteRef{\textsuper{\hyperref[t:x28x7ccounterx2dx28x29x7c_x23x28structx3ageneratedx2dtagx29x29]{1}}}\FootnoteContent{The actual proof needs to generalize
$\Gg_2(e_c)$ and $\Gg_2(e_f)$ further. See \cref{prop:u-if}.}}
In this case, we need to prove that the latter
reduces to a related answer based on the assumption that
$\Gg_1(\If{e_c}{\ouro}{e_f})\arrS^j a_1$ and $a_1\not\equiv\ouro$.
As it turns out, this variation poses no issue to the proof
because the sub{-}expressions reduce in a related manner by induction.
The complete proof can be found on page 6 of Appendix E.\end{proof}

A consequence of the Fundamental Property is the forward direction of
\cref{prop:compiler-correctness} specialized to  $\arrU$:
when $e$ does not
exhibit undefined behavior, for any transformation $e\arrU^* e'$, if
$e$ terminates then $e'$ terminates with a related answer.

\begin{corollary}\label{coro:forward-correct}Assume $\wfe{\GD}{e}$ and $e \arrU^* e'$.
For all $C$, $a$, if $~\wfctx{}{\GD}{C}$,
$\neg\undefp{C[e]}$ and $C[e]\arrS^* a$ then
there exists $a'$ and $j\ge 0$ such that
 $C[e']\arrS^* a'$ and either
 $(a,a')\in\relV_j^{\arrU} \text{ or } a\equiv a'\equiv \error{k}.$\end{corollary}
\begin{proof}[Proof Sketch]Assume $C[e]\arrS^i a$. Because $\arrU$ allows us to apply
the transformation in any context, composing $C$
with each of the expressions in $e\arrU^* e'$
yields $C[e]\arrU^* C[e']$.
Now, the \nameref{lem:forward-fundamental} gives
$(C[e],C[e])\in\relE_{i+1}^{\arrU}$.
By $\neg\undefp{C[e]}$, $C[e]\arrS^i a$ and $C[e]\arrU^* C[e']$,
we conclude that there exists $a_2$ such that
$C[e']\arrS^* a_2$ and either $(a,a_2)\in\relV_1^{\arrU}$
or $a\equiv a_2\equiv\error{k}$.
The complete proof can be found on page 4 of Appendix E.\end{proof}\label{t:x28part_x28gentag_0x29x29}

\begin{FootnoteBlock}\FootnoteBlockContent{\FootnoteTarget{\textsuper{\label{t:x28x7ccounterx2dx28x29x7c_x28gentag_1x29x29}\textsf{1}}}The actual proof needs to generalize
$\Gg_2(e_c)$ and $\Gg_2(e_f)$ further. See \cref{prop:u-if}.}\end{FootnoteBlock}

\Ssubsection{Backward Approximation of $\arrU$ Reductions}{Backward Approximation of $\arrU$ Reductions}\label{t:x28part_x22secx3abackwardx2drelx22x29}

Having established the forward direction of \cref{prop:compiler-correctness}
for $\arrU$,
we turn to the backward one.
That is, we would like to show that the behavior of
the original expression approximates the behavior
of the transformed expression, assuming that the original expression does
not exhibit undefined behavior: if $e\arrU^* e'$ and $\neg\undefp{e}$ then

\noindent \begin{itemize}\atItemizeStart

\item $\neg\undefp{e'}$,

\item $\forall c.\,(e\arrS^* c \Longleftarrow e'\arrS^* c$),

\item $(\exists e_1.\ e\arrS^* \Func{x}{e_1}) \Longleftarrow (\exists e_1'.\  e'\arrS^* \Func{x}{e_1'})$, and

\item $\forall k.\,(e\arrS^* \error{k} \Longleftarrow e'\arrS^* \error{k}$).\end{itemize}

\noindent However, as discussed at the beginning of this section, the
$\neg\undefp{e}$ assumption complicates designing a
backward{-}approximation logical relation. In an ideal world, we
would like to proceed by induction on $e'$, as our assumption tells us a lot about
how it evaluates. Unfortunately, our assumption also includes $\neg\undefp{e}$, which
is not helpful when we are working by induction on $e'$.

As it turns out, the treatment of \texorpdfstring{\ensuremath{\ouro}}{unreachable} by the standard reduction
of our calculus as an error offers a way forward. While
$\neg\undefp{e}$ is not compositional, $\undefp{e}$ is. In detail, if
we know whether the sub{-}expressions of $\undefp{e}$ evaluate to
$\ouro$, we can decide whether $e\arrS^* \ouro$ based on the structure
of $e$. Accordingly, while working by induction on $e'$ in the proof, it is easier to establish an
additional proof goal $\undefp{e}$ than it is to discharge the premise
that the sub{-}expressions of $e$ do not evaluate to $\ouro$ (when applying the inductive hypothesis).
Therefore, we reshape the statement of
backward approximation: we omit conjunct $\neg\undefp{e}$ from the
premise and include $\undefp{e}$ as another disjunct in
the conclusion of backward approximation. This gives us an equivalent
proposition, but that is conducive to proof by induction.

To prove the reshaped property, we construct a
logical relation that guarantees the original expression produces a
related answer to that of the transformed expression, but also permits the
original expression to exhibit undefined behavior, just as we did for the forward
logical relation.
\begin{definition}[Backward-Approximation Logical Relation]\label{def:backward-rel}\[
\begin{array}{r@{}c@{}l}
\relV_i^{\uarr} & = &
\{ (\Gl x.e_1,\Gl x.e_2) \;|\;
 \forall j<i. \;
 \forall v_1\, v_2.\; (v_1,v_2)\in \relV_j^{\uarr} \Implies
 (\subst{e_1}{x}{v_1},\subst{e_2}{x}{v_2})\in \relE_j^{\uarr}
 \}
\\
& & \; \cup \; \lset{ (c,c) }
\\
\relE_i^{\uarr} & \;=\; &
\begin{array}[t]{@{}c@{}c@{}l@{}}
\{ (e_1,e_2) & \;|\; &
 \forall j < i. \;
 \forall e_2',a_2. \; \; e_2 \arrU^* e_2' \land e_2' \arrS^j a_2 \Implies \\
 & & \quad
 \undefp{e_1} \; \lor \\
 & & \quad
 (\exists a_1.\; e_1 \arrS^* a_1 \land
 ((a_1,a_2)\in \relV_{i-j}^{\uarr} \lor \exists k.\, a_1\equiv a_2\equiv \error{k})
 )\}
\end{array} \\
     \end{array}
\]\end{definition}

Same as for the forward{-}approximation logical relation, the value relation
$\relV_i^{\uarr}$ is standard, but the expression relation
$\relE_i^{\uarr}$ is not.  For any pair of related expressions
$(e_1,e_2)$, the antecedent of $\relE_i^{\uarr}$ is that $e_2
\arrU^* e_2'$ and $e_2'$ terminates with the answer $a_2$ after $j$
standard reduction steps.  The consequent then asserts that either $e_1$
exhibits undefined behavior, or $e_1$ also terminates with a related
answer $a_1$.  In particular, $\relE_i^{\uarr}$ relates $a_1$ and
$a_2$ if they are related \emph{values} or the same error.
Consequently, if $a_2\equiv\ouro$, i.e. $e_2'$ triggers undefined
behavior, $e_1$ must end up with undefined behavior as $\ouro$ is not
related to any other answers.

As with $\relE^{\arrU}$, we prove the
\nameref{lem:backward-fundamental}.

\begin{lemma}[Fundamental Property for $\relE^{\uarr}$]\label{lem:backward-fundamental}For all $e,\, i\ge 0,\, \GD$ and $\Gg$,
if $\wfe{\GD}{e}$ and $\Gg\in\relG_i^{\uarr}[\GD]$ then
$(\Gg_1(e), \Gg_2(e))\in \relE_i^{\uarr}$.

 \begin{remark*} Similar to the Fundamental Property for the forward logical approximation,
 $\relG_i^{\uarr}[\GD]$ is
the set of closing substitutions mapping $x\in\GD$ to a pair of values
related by $\relV_i^{\uarr}$.\end{remark*}\end{lemma}

\begin{proof}[Proof Sketch]By induction on $e$. Similar to the \nameref{lem:forward-fundamental},
we discuss the case of conditional expressions. The proofs for the other
constructs
follow from routine case analysis and the inductive hypothesis.
Let $e :\equiv \If{e_c}{e_t}{e_f}$, $i\ge 0$
and $\Gg\in\relG_i^{\uarr}[\GD]$ be given.
We need to prove that
\[
\begin{array}{@{}l@{}}\forall j < i. \;
\forall e_2' \, a_2. \; \; \If{\Gg_2(e_c)}{\Gg_2(e_t)}{\Gg_2(e_f)}
\arrU^* e_2' \land e_2' \arrS^j a_2 \Implies \\
\quad
\Lundefp\If{\Gg_1(e_c)}{\Gg_1(e_t)}{\Gg_1(e_f)}
\;\lor \\
\quad
(\exists a_1.\; \If{\Gg_1(e_c)}{\Gg_1(e_t)}{\Gg_1(e_f)} \arrS^* a_1 \land
((a_1,a_2)\in \relV_{i-j}^{\uarr} \lor \exists k.\, a_1\equiv a_2\equiv \error{k})
)\end{array}
\]
By \cref{prop:u-if}, $\If{\Gg_2(e_c)}{\Gg_2(e_t)}{\Gg_2(e_f)} \arrU^* e_2'$
can be decomposed into three reduction sequences
$\Gg_2(e_c)\arrU^* e_c'$, $\Gg_2(e_t)\arrU^* e_t'$ and $\Gg_2(e_f)\arrU^* e_f'$
such that either (i) $e_2'\equiv \If{e_c'}{e_t'}{e_f'}$,
(ii) $e_2'\equiv \Seq{e_c'}{e_t'}$ and $e_f'\equiv\ouro$, or
(iii) $e_2'\equiv \Seq{e_c'}{e_f'}$ and $e_t'\equiv\ouro$.
Among all situations, we focus on case (iii) since it contains an
application of $\arrU$ to the whole expression.
Now, if $a_2$ equals some value $v_f'$,
the evaluation $\Seq{e_c'}{e_f'}\arrS^j a_2$
must have the pattern
\[
\Seq{e_c'}{e_f'} \arrS^{j_c} \Seq{v_c'}{e_f'} \arrS e_f' \arrS^{j_f} v_f'.
\]
Therefore $e_c'\arrS^{j_c} v_c'$ and $e_f'\arrS^{j_f'} v_f'$ for some
steps $j_c'$ and $j_f'$.
Together with the transformations $\Gg_2(e_c)\arrU^* e_c'$ and
$\Gg_2(e_f)\arrU^* e_f'$, the induction hypothesis yields
\begin{equation}
\begin{array}{@{}l@{}}\undefp{\Gg_1(e_c)} \lor
(\exists a_c.\; \Gg_1(e_c) \arrS^* a_c \land
((a_c,v_c')\in \relV_{i-j_c}^{\uarr} \lor
\exists k.\, a_c\equiv v_c'\equiv \error{k})
)
\\
\undefp{\Gg_1(e_f)} \lor
(\exists a_f.\; \Gg_1(e_f) \arrS^* a_f \land
((a_f,v_f')\in \relV_{i-j_c-1-j_f}^{\uarr} \lor
\exists k.\, a_f\equiv v_f'\equiv \error{k})
)\end{array}
\tag{$\star$}
\label{eq:if-IH-res}
\end{equation}
Because $v_c'$ ($v_f'$) is a value, the $\error{k}$ case
in \labelcref{eq:if-IH-res} cannot happen.
The answer $a_c$ ($a_f$) accordingly is related to $v_c'$ ($v_f'$).
As a result, if $\undefp{\Gg_1(e_c)}$ or $\undefp{\Gg_1(e_f)}$,
i.e. a subexpression triggers undefined behavior,
the treatment of \texorpdfstring{\ensuremath{\ouro}}{unreachable} by the standard reduction guarantees that
$\Lundefp\If{\Gg_1(e_c)}{\Gg_1(e_t)}{\Gg_1(e_f)}$.
Otherwise, we have $\Gg_1(e_c) \arrS^* a_c \land (a_c,v_c')\in \relV_{i-j_c}^{\uarr}$
and $\Gg_1(e_f) \arrS^* a_f \land (a_f,v_f')\in \relV_{i-j_c-1-j_f}^{\uarr}$.
The evaluation of $\If{\Gg_1(e_c)}{\Gg_1(e_t)}{\Gg_1(e_f)}$
thus depends on whether $a_c$ is $\False$ or not.

When $a_c$ is $\False$, $\If{\Gg_1(e_c)}{\Gg_1(e_t)}{\Gg_1(e_f)}$
reduces $\Gg_1(e_f)$. Hence the fact
$(a_f,v_f')\in \relV_{i-j_c-1-j_f}^{\uarr}$ entails our desired conclusion.
When $a_c$ is truthy, $\If{\Gg_1(e_c)}{\Gg_1(e_t)}{\Gg_1(e_f)}$
reduces to $\Gg_1(e_t)$. Since $\Gg_2(e_t)\arrU^* e_t'$ and
$e_t'\equiv\ouro$ in case (iii), $e_t$ itself must be \texorpdfstring{\ensuremath{\ouro}}{unreachable}.
Hence, the full expression $\If{\Gg_1(e_c)}{\Gg_1(e_t)}{\Gg_1(e_f)}$
terminates with \texorpdfstring{\ensuremath{\ouro}}{unreachable}. Nevertheless,
exhibiting undefined behavior is precisely one
of the expected behaviors.

At this point, we have covered the sub{-}case $a_2\equiv v_f'$ of case (iii).
The sub{-}case where $a_2$ is \texorpdfstring{\ensuremath{\ouro}}{unreachable} is proved with the same strategy.
The complete proof can be found on page 5 of Appendix F.\end{proof}

The proof of the \nameref{lem:backward-fundamental}
relies on the following property of  sequences of $\arrU$ reductions:

\noindent \begin{proposition}\label{prop:u-if}
If $\If{e_c}{e_t}{e_f}\arrU^* e'$, there exists three transformation
sequences $e_c\arrU^* e_c'$, $e_t\arrU^* e_t'$ and $e_f\arrU^* e_f'$
such that one of the following holds:

\noindent \begin{itemize}\atItemizeStart

\item $e' \equiv \If{e_c'}{e_t'}{e_f'}$

\item $e' \equiv \Seq{e_c'}{e_t'}$ and $e_f'\equiv\ouro$, or

\item $e' \equiv \Seq{e_c'}{e_f'}$ and $e_t'\equiv\ouro$\end{itemize}\end{proposition}

\noindent \begin{proof}[Proof Sketch]By induction on $\If{e_c}{e_t}{e_f}\arrU^* e'$.
The complete proof can be found on page 12 of Appendix C.\end{proof}

Analogously to the Fundamental Property for the forward{-}approximation relation, a
corollary of the \nameref{lem:backward-fundamental} is the backward direction of
\cref{prop:compiler-correctness} specialized to  $\arrU$: either the
transformed program and the original program evaluate to
related answers, or the original program exhibits undefined
behavior:

\begin{corollary}\label{coro:backward-correct}Assume that $\wfe{\GD}{e}$ and $e \arrU^* e'$.
For all $C$ and $a'$, if $\wfctx{}{\GD}{C}$ and
$C[e']\arrS^* a'$ then either $\undefp{C[e]}$
or there exists $a$ and $j\ge 0$ such that
$C[e]\arrS^* a$ and either
$(a,a')\in\relV_j^{\arrU} \text{ or } a\equiv a'\equiv \error{k}.$
\begin{remark*}In \cref{coro:backward-correct}, $a'$ can be
\texorpdfstring{\ensuremath{\ouro}}{unreachable} in which case the conclusion is $\undefp{C[e]}$.\end{remark*}\end{corollary}
\begin{proof}[Proof Sketch]Assume $C[e']\arrS^i a'$. Similar to \cref{coro:forward-correct},
we compose $C$ with each expression in $e\arrU^* e'$
to obtain $C[e]\arrU^* C[e']$.
By the \nameref{lem:backward-fundamental},
$(C[e],C[e])\in\relE_{i+1}^{\uarr}$.
The fact that $C[e]\arrU^* C[e']$ and $C[e']\arrS^i a'$
yields $\undefp{C[e]}$
or that there exists $a_1$ such that
$C[e]\arrS^* a_1$ and either $(a_1,a')\in\relV_1^{\uarr}$
or $a_1\equiv a'\equiv\error{k}$.
The complete proof can be found on page 2 of Appendix F.\end{proof}

\Ssubsection{Completing the Proof of Correctness of the Compile{-}Time Reductions}{Completing the Proof of Correctness of the Compile{-}Time Reductions}\label{t:x28part_x22secx3acorrectx22x29}

The fundamental properties of the two above logical relations entail
only one piece of the correctness of our compile{-}time reduction relation.
After all, $\arrC$ includes $\arrP$ as well as $\arrU$.

Fortunately, and unlike the challenging $\arrU$ reductions, $\arrP$
reductions preserve the meaning of expressions in all contexts without
any conditions. Hence, the proof of \cref{prop:compiler-correctness}
specialized to  $\arrP$, i.e. the correctness of the  $\arrP$
reductions, can be established the standard way we discuss at the
beginning of this section {---} with a standard logical relation. Indeed,
the correctness of the $\arrP$ reductions  is a simple collorary of
\cref{lem:std-sound}, which we discuss in
\ChapRefUC{\SectionNumberLink{t:x28part_x22secx3aconventionalx2dtransformationsx22x29}{5}}{Extending the Calculus
with Additional Rules} that introduces extensions of
the compile{-}time reduction with useful, semantics{-}preserving but
unrelated{-}to{-}$\ouro$ transformations. To avoid repetition, we omit
further discussion herein. The interested reader can also find the
 the full formal details in page 28 of Appendix G and Appendix H.

With the correctness of $\arrU$ and $\arrP$ in hand, we proceed to
prove the correctness of $\arrC$. In fact, we establish a  a generalized
compiler correctness theorem that adapts
\cref{prop:compiler-correctness} to open expressions, formalizing the
intuitive idea that compiler transformations preserve the semantics of
program pieces. In keep with the nature of undefined behavior, the
theorem only holds for contexts $C$ that close an expression $e$ so that
$\neg\undefp{C[e]}$:

\NamedTheorem{Generalized Compiler Correctness}{\label{thm:compiler-correctness}}%
{Assume that $\wfe{\GD}{e}$ and $e \arrC^* e'$.
For all $C$ such that $\wfctx{}{\GD}{C}$,
if $\neg\undefp{C[e]}$ we have

\noindent \begin{itemize}\atItemizeStart

\item $\neg\undefp{C[e']}$,

\item $\forall c.\ C[e]\arrS^* c \iff C[e'] \arrS^* c$,

\item $(\exists e_1. C[e]\arrS^* \Func{x}{e_1}) \iff (\exists e_1'. C[e'] \arrS^* \Func{x}{e_1'})$, and

\item $\forall k.\ C[e]\arrS^* \error{k} \iff C[e'] \arrS^* \error{k}$\end{itemize}}

\noindent \begin{proof}[Proof Sketch]By induction on $e\arrC^* e'$
and application of \cref{coro:forward-correct} and
\cref{coro:backward-correct} for each $\arrU$ reduction.
For $\arrP$ reductions, we apply \cref{lem:std-sound}. The complete proof
can be found on Appendix J.\end{proof}

\sectionNewpage

\Ssection{Extending the Calculus
with Additional Rules}{Extending the Calculus
with Additional Rules}\label{t:x28part_x22secx3aconventionalx2dtransformationsx22x29}

The $\arrU$ and $\arrP$ reductions capture the essence of
\texorpdfstring{\ensuremath{\ouro}}{unreachable} but they are not sufficient to describe realistic
compile{-}time transformations. Production compilers combine common
transformations such as Common Subexpression Elimination, Loop Unrolling, and Strength Reduction
with the \texorpdfstring{\ensuremath{\ouro}}{unreachable}{-}related transformations; the transformations work in tandem, as one
transformation may open up additional optimization opportunities for another.
For this reason, we add extra rules to the \texorpdfstring{\ensuremath{\ouro}}{unreachable} calculus that
enhance its power to capture realistic program transformations. In fact,
any extra rule is compatible with the calculus  under one
requirement: the new rule must be sound with respect to contextual
equivalence as induced by the standard reduction of the calculus.

\begin{Figure}\begin{Centerfigure}\begin{FigureInside}\ContextPreservingRules\end{FigureInside}\end{Centerfigure}

\Centertext{\Legend{\FigureTarget{\label{t:x28counter_x28x22figurex22_x22figx3acontextpreservingrulesx22x29x29}\textsf{Fig.}~\textsf{9}. }{t:x28counter_x28x22figurex22_x22figx3acontextpreservingrulesx22x29x29}\textsf{Compile{-}Time Transformations Based on Standard Reduction}}}\end{Figure}

As an example of extra compile{-}time rules, rules M.1 to M.5 in
Figure~\hyperref[t:x28counter_x28x22figurex22_x22figx3acontextpreservingrulesx22x29x29]{\FigureRef{9}{t:x28counter_x28x22figurex22_x22figx3acontextpreservingrulesx22x29x29}} make available at compile{-}time the
standard notions of reduction of the calculus from
Figure~\hyperref[t:x28counter_x28x22figurex22_x22figx3astandardx2dreductionx22x29x29]{\FigureRef{3}{t:x28counter_x28x22figurex22_x22figx3astandardx2dreductionx22x29x29}}. In essence,
they allow the compiler to partially {``}compute{''} forward and backward in
any sub{-}expression of a program:

\begin{itemize}\atItemizeStart

\item  M.1 corresponds to the reverse notion of reduction S.2. It enables
a compiler to  {``}reverse the evaluation{''} of conditional expressions whose test is known to be not
false.  The else{-}branch of the produced conditional expression can be any
arbitrary expression as long as it doesn{'}t introduce free variables not
accounted for by $\Delta$ {---} even expressions that contain
\texorpdfstring{\ensuremath{\ouro}}{unreachable}.

\item M.2 is the reverse of notion of reduction S.1, and similarly to M.1 produces conditional
expression that in this case have a test that is equal to false.

\item M.3 is a generalized reversed beta reduction. Specifically, it allows any
\emph{safe} sub{-}expression to be lifted
out of an expression, which, in the compiler realm, is useful for
modeling transformations such as Common Subexpression Elimination.

\item M.4 corresponds to the notion of reduction S.4. It permits the compiler
to drop all expressions in a sequence of expressions except the last one
as long as these dropped expressions evaluate to a value. This rule
together with rules U.1 and U.2 plays an important rule for the
simplification of the examples in \ChapRefUC{\SectionNumberLink{t:x28part_x22secx3aexamplesx22x29}{2}}{The Essence of Unreachable, by Example}.

\item M.5 is the reverse of notion of reduction S.5, which allows for the
backward {``}evaluation{''} of arithmetic expressions.

\item \textsc{Ctx.M} and
\textsc{CtxSym.M} define the symmetric and compatible closure (over
contexts) of the above rules. Hence, they allow compilers to use all
these rules forward and backward, and in any part of a program.\end{itemize}

\begin{Figure}\begin{Centerfigure}\begin{FigureInside}\AdditionalRMRules
 $\vspace*{.05in}$\end{FigureInside}\end{Centerfigure}

\Centertext{\Legend{\FigureTarget{\label{t:x28counter_x28x22figurex22_x22figx3aadditionalrmrulesx22x29x29}\textsf{Fig.}~\textsf{10}. }{t:x28counter_x28x22figurex22_x22figx3aadditionalrmrulesx22x29x29}\textsf{Additional $R_m$ rules}}}\end{Figure}

In addition to the forward and backward compile{-}time partial evaluation of
expressions, compilers come with a large number of transformation rules
that move the sub{-}expressions of a program.
 Figure~\hyperref[t:x28counter_x28x22figurex22_x22figx3aadditionalrmrulesx22x29x29]{\FigureRef{10}{t:x28counter_x28x22figurex22_x22figx3aadditionalrmrulesx22x29x29}}
contains a collection of such rules.

\begin{itemize}\atItemizeStart

\item M.6 and M.7 enable the compiler to move computation from the context of  a sequence and a conditional expression respectively
to tail position. The $\evalctx^+$ represents
a generalized version of evaluation context which admits variables.

\item M.8 allows the compiler to optimize conditional expressions whose test
is  a variable (say $x$). If the else{-}branch of the expression is ever evaluated,
then it must be the case that $x$ is $\False$. Note the same is not
correct for the then{-}branch and $\True$ due to the truthiness of the language.

\item M.9 permits the compiler to collapse certain conditional expressions
whose test is a trivial conditional statement as well. This is helpful for modeling optimizations for
languages with conditional select statements such as the LLVM IR.

\item M.10 empowers the compiler to change the order of equality checks in
nested conditional expressions when the branches involved are
syntactically identical.  Such rearrangements of checks can model the
behavior of transformations of switch statements in the LLVM IR.\end{itemize}

To establish the correctness of $\arrM$, we demonstrate that each of its
rules preserves contextual equivalence (as induced by the standard
reduction of the calculus). We do so following the standard recipe that we
discuss at the beginning of \ChapRefUC{\SectionNumberLink{t:x28part_x22secx3acorrectnessx2dcompilerx22x29}{4}}{The Correctness of the Compile{-}Time Reductions}: we define
a standard binary step{-}indexed logical relation for the calculus, we prove
it sound with respect to contextual approximation, and then we use the
logical relation to prove correct each rule of $\arrM$. Specifically,
for all $i\ge 0$, we define value and expression logical approximation
  at step $i$ to be $\relV_i^{\arrS}$ and $\relE_i^{\arrS}$.

\begin{definition}[Standard Logical Relation]\[
  \begin{array}{@{}r@{}c@{}l@{}}
  \relV_i^{\arrS} & = &
  \{ (\Gl x.e_1,\Gl x.e_2) \;|\;
   \forall j<i. \;
   \forall v_1\, v_2.\; (v_1,v_2)\in \relV_j^{\arrS} \Implies
   (\subst{e_1}{x}{v_1},\subst{e_2}{x}{v_2})\in \relE_j^{\arrS}
   \}
  \\
  && \; \cup \lset{ (c,c) }
\\
  \relE_i^{\arrS} & \;=\; &
  \left\{\vphantom{\relV^{\arrS}_{i-j}}\right. \!
   (e_1,e_2) \,\left|\vphantom{\relV^{\arrS}_{i-j}}\right.\,
   \forall j < i. \; \forall a_1. \; e_1 \arrS^j a_1 \Implies \\
   & & \quad \quad \quad \quad \quad
   \exists a_2.\; e_2 \arrS^* a_2 \; \land \\
   & & \quad \quad \quad\quad\quad \quad
   \left((a_1,a_2)\in \relV_{i-j}^{\arrS} \lor
   (a_1 \equiv a_2 \equiv \ouro) \lor
   (\exists k.\, a_1\equiv a_2\equiv \error{k})\right)
   \! \left.\vphantom{\relV^{\arrS}_{i-j}}\right\}
    \end{array}
\]\end{definition}

The relations $\relV^{\arrS}$ and $\relE^{\arrS}$ are straight{-}forward
adaptations of standard binary step{-}indexed logical relations for proving
contextual equivalences in functional languages. According to
$\relV^{\arrS}$, base values, errors, and \texorpdfstring{\ensuremath{\ouro}}{unreachable} are only related to
themselves. Two functions are related only when, given related arguments,
they produce related results.  $\relE^{\arrS}$ relates closed
expressions $e$, $e'$ as long as $e$ approximates $e'$ after at
most $i$ steps under the standard reduction of the calculus.  As usual,
the step indices guarantee that the relations are well{-}founded.

$\relE^{\arrS}$ relates only closed expressions for a given number of steps.
To prove contextual approximation for a pair of expressions, a stronger
statement that generalizes over open expressions and arbitrary number of
steps is necessary. Hence, we define logical approximation $\GD\vdash
e_1 \preceq e_2$:

\begin{definition}[Logical Approximation] For all $\GD$, $e_1$, $e_2$ such that
 $\wfe{\GD}{e_1}$ and $\wfe{\GD}{e_2}$, we say that $e_1$ logically
 approximates $e_2$,
  $\GD\vdash e_1 \preceq e_2$, iff
 \[
\displaystyle
\forall i \geq 0.\;
\forall \Gg\in\relG_i^{\arrS}[\GD].\;
(\Gg_1(e_1), \Gg_2(e_2))\in\relE^{\arrS}_i.
\]
where $\relG_i^{\arrS}[\GD]$ is as in
\ChapRefUC{\SectionNumberLink{t:x28part_x22secx3acorrectnessx2dcompilerx22x29}{4}}{The Correctness of the Compile{-}Time Reductions} except that it draws pairs of values
from $\relV^{\arrS}$.  \end{definition}

Due to the compatibility lemmas of the standard logical
relation, logical approximation is sound with respect to
contextual approximation. The complete proof is on page 28
of Appendix G.

\NamedLemma{Soundness}{\label{lem:std-sound}
 If $\GD\vdash e \preceq e'$ and $\cdot \mmodels C:\GD$ then

\noindent \begin{itemize}\atItemizeStart

\item $\undefp{C[e]} \implies \undefp{C[e']}$,

\item $\forall c.\ C[e]\arrS^* c \implies C[e'] \arrS^* c$,

\item $(\exists e_1.\ C[e]\arrS^* \Func{x}{e_1}) \implies
(\exists e_1'.\ C[e'] \arrS^* \Func{x}{e_1'})$, and

\item $\forall k.\ C[e]\arrS^* \error{k} \implies C[e'] \arrS^* \error{k}$.\end{itemize}}

Finally, given the soundness of the standard logical relation,
we prove that  the $\arrM$ reductions
are correct by demonstrating that all expressions before
and after each reduction logically approximate each other.
  For instance, to prove M.1 is correct, we
show:

\noindent \begin{itemize}\atItemizeStart

\item If $v\not\equiv\False$, $\GD \mmodels e_t$ and $\GD \mmodels e_f$ then $\GD \vdash e_t \preceq \If{v}{e_t}{e_f}$

\item If $v\not\equiv\False$, $\GD \mmodels e_t$ and $\GD \mmodels e_f$ then $\GD \vdash \If{v}{e_t}{e_f} \preceq e_t$.\end{itemize}

\noindent \noindent The complete proof can be found on page 14 of Appendix H, where the proof for other $R_m$ rules are also located.

\sectionNewpage

\Ssection{Beyond Unreachable Code}{Beyond Unreachable Code}\label{t:x28part_x22secx3aotherx2dubx22x29}

Beyond simple unreachability, \texorpdfstring{\ensuremath{\ouro}}{unreachable} can model a broad
range of undefined behaviors, including division by zero,
arithmetic overflow and under flow, and null pointer
dereferencing. In particular, we can use \texorpdfstring{\ensuremath{\ouro}}{unreachable} to
encode undefined behaviors due to uses of primitive
operators with an illegal argument. For instance, integer
division is well{-}defined over all integers unless the
divisor is zero. Therefore, a compiler can assume that the
erroneous divisor is never supplied to the operation. In
other words, calls to division with a zero divisor are
unreachable.

Abstractly, if some operation $\mathcal{P}$ is undefined over some
portion of its domain $\mathcal{X}_\mathit{undef}$, then we can encode such an
operation as a conditional wrapper of the raw operation
$\mathcal{P}_{r}$:
\[
\mathcal{P} = \Func{x}{
 \If{(x \in \mathcal{X}_\mathit{undef})}{
  \ouro
  }{
  (\App{\mathcal{P}_r}{x})
 }
}
\]
As a result, with an $\arrU$ reduction,
a call $(\App{\mathcal{P}}{x})$ reduces to
$(\App{\mathcal{P}_{r}}{x})$ only if $x \not\in \mathcal{X}_\mathit{undef}$.
In essence, the insertion of \texorpdfstring{\ensuremath{\ouro}}{unreachable}
signals to the compiler that the undefined behavior never occurs, and therefore,
 the compiler is allowed to optimize away the checks, exposing the
raw operation. The beauty of this encoding, however, is that it other
transformations that match the kinds of optimizations that compilers
do, but without needing any new \texorpdfstring{\ensuremath{\ouro}}{unreachable}{-}specific rules.

To see how this captures more than simply eliminating the checks, we explore
the undefined behavior of signed integer addition. According to the C
standard\Autobibref{~(\hyperref[t:x28autobib_x22International_Organization_for_StandardizationISOx2fIEC_14882x3a2011_Cx2bx2b_Standard2011httpsx3ax2fx2fwwwx2eisox2eorgx2fstandardx2f50372x2ehtmlx22x29]{\AutobibLink{International Organization for Standardization}} \hyperref[t:x28autobib_x22International_Organization_for_StandardizationISOx2fIEC_14882x3a2011_Cx2bx2b_Standard2011httpsx3ax2fx2fwwwx2eisox2eorgx2fstandardx2f50372x2ehtmlx22x29]{\AutobibLink{2011}})}, {``}If during the evaluation of an
expression, the result is not mathematically defined or not in the range
of representable values for its type the behavior is undefined{''}.
Put differently, if the result of signed integer addition would be greater than
\Scribtexttt{INT{\char`\_}MAX} or smaller than \Scribtexttt{INT{\char`\_}MIN}, then the addition exhibits
undefined behavior. For instance, consider the C code snippet
\Scribtexttt{x {\Stttextless} x + 1}. If a C compiler
decides to implement \Scribtexttt{+} by deferring to the underlying
machine arithmetic, and that arithmetic is two{'}s complement,
then the code snippet is equivalent to \Scribtexttt{x == INT{\char`\_}MAX}.
If, however, the compiler decides that overflow is undefined behavior, and therefore
assuming it assumes that overflow never happen, then it is justified to
compile to a constant expression that is always true.
Indeed, gcc v7.5 (with the default options) produces code semantically
(but not syntactically)
equivalent to comparing \Scribtexttt{x} with \Scribtexttt{INT{\char`\_}MAX}, but gcc v8.1
(with the default options) compiles the snippet to \Scribtexttt{1}.

Our calculus can capture both of these possibilities. To do so, the
calculus is extended with primitive operators, $=_\mathbb{Z}$, $\ne_\mathbb{Z}$, $<_\mathbb{Z}$ and
$+_\mathbb{Z}$, where each operator
corresponds to its mathematical counterpart,
which is well{-}defined for all (mathematical) integers.
Additionally, $x +_\texttt{int} y$, which is akin to C{'}s signed integer addition,
is a shorthand for the
conditional expression
\[
\begin{array}{@{}l@{}}
\If{(\mathsf{MAX}_\texttt{int} <_\mathbb{Z} x +_\mathbb{Z} y)}{ \\
 \quad \ouro
 }{ \\
 \quad
 \If{(x +_\mathbb{Z} y <_\mathbb{Z} \mathsf{MIN}_\texttt{int})}{\\
  \quad \quad \ouro
  }{ \\
  \quad \quad (x +_\mathbb{Z} y)
 }
}
\end{array}
\]
Given these operators, our calculus
is able to equate the expression
$x <_\mathbb{Z} (x +_\texttt{int} 1)$ with both the
expressions $x \ne_\mathbb{Z} \mathsf{MAX}_\texttt{int}$ and $\True$.
First, the \texorpdfstring{\ensuremath{\ouro}}{unreachable} calculus can equate
$x <_\mathbb{Z} (x +_\texttt{int} 1)$ with $\True$ by eliminating $\ouro$s in $x +_\texttt{int} 1$ with
$\arrU$ reductions. Second, the calculus can  equate $x <_\mathbb{Z} (x +_\texttt{int} 1)$
with $x \ne_\mathbb{Z} \mathsf{MAX}_\texttt{int}$ with successive applications of $\arrM$
reductions. The basic idea is that an $\arrM$ reduction can
perform a reverse standard reduction to introduce a conditional that
checks whether $x$ is $\mathsf{MAX}_\texttt{int}$ or not.
This transformation enables further $\arrM$ reductions in each of the two
branches taking advantage of the assumption that the result of the
conditional{'}s check is different for the then
and else branch. The complete proof can be found on Appendix K.

\sectionNewpage

\Ssection{Unreachable in Racket on Chez}{Unreachable in Racket on Chez}\label{t:x28part_x22Unreachablex5finx5fRacketx5fonx5fChezx22x29}

The simplicity of the \texorpdfstring{\ensuremath{\ouro}}{unreachable} calculus raises the question of its
relation to production compilers. In this section, we look at the
relation between the reduction rules of our calculus and the way Racket{'}s
compiler transforms unreachable code.

Racket{'}s core compiler is Chez Scheme. Overall, Chez Scheme iterates
through a sequence of source{-}to{-}source passes a configurable number of
times (the default is two).  Chez Scheme{'}s support for unreachable is
primarily in the \Scribtexttt{cptypes} pass, which performs source{-}to{-}source
optimizations based on type inference.  In addition, transformations that
involve simplifications for unused variables or expressions are also part
of the earlier \Scribtexttt{cp0} source{-}to{-}source pass.

The notions of reduction in
Figure~\hyperref[t:x28counter_x28x22figurex22_x22figx3acompilerruleUx22x29x29]{\FigureRef{6}{t:x28counter_x28x22figurex22_x22figx3acompilerruleUx22x29x29}} and Figure~\hyperref[t:x28counter_x28x22figurex22_x22figx3acompilerrulePx22x29x29]{\FigureRef{5}{t:x28counter_x28x22figurex22_x22figx3acompilerrulePx22x29x29}} map to
specific lines in the implementation of Chez Scheme (the corresponding
full files are part of the supplemental Appendix M):

Rules U.1 and U.2 correspond directly to the case that matches
 \RktSym{if} forms in the implementations of the  \Scribtexttt{cptypes} pass.
 Specifically, the \Scribtexttt{cptypes} pass optimizes recursively  both  branches of
 an \RktSym{if} form. Together with the optimized branches, the pass
 returns additional information, including a type.  If the pass determines
 than one of the branches never returns then the corresponding type is
 \RktVal{{\textquotesingle}}\RktVal{bottom}. The
 pass uses  the
 \RktSym{unsafe{-}unreachable{\hbox{\texttt{?}}}} predicate to identify a non{-}returning
 branch due to unreachable, and then replaces the \RktSym{if}
 form with a call to \RktSym{make{-}seq}, which constructs a \RktSym{seq}
 expression, the same as our calculus{'}s \RktSym{begin} expression.
 Hence, the pass eliminates the branches of a
 conditional that do not terminate because they are unreachable in exactly
 the same way that the U.1 and U.2 rules simplify conditionals.

Rule P.1 corresponds to uses of  the \RktSym{make{-}seq} function. In
  detail,  the \Scribtexttt{cp0} pass uses \RktSym{make{-}seq}
  to optimize \RktSym{seq} forms; it acts as a
  smart constructor that, same as rule P.1, drops from a sequence those
  expressions that are {``}simple{''} and whose result is unused.  There{'}s
  also a limited variant of the same simplification in the \Scribtexttt{cptypes} pass;
  the duplication aims to reduce the number of optimizer{-}pass iterations
  needed in practice.

Rule P.2 corresponds to a case that matches \RktSym{seq} forms in
  the \Scribtexttt{cptypes} pass. This case handles all
   expressions at the beginning of a  \RktSym{seq} form that do not return.
   Specifically, if the first expression of  a  \RktSym{seq}
   does not return, then all subsequent expressions are ignored:

\begin{smaller}\begin{SCodeFlow}\begin{RktBlk}\begin{SingleColumn}\RktPn{(}\RktSym{define{-}pass}\mbox{\hphantom{\Scribtexttt{x}}}\RktSym{cptypes}\mbox{\hphantom{\Scribtexttt{x}}}\RktSym{{\hbox{\texttt{.}}}{\hbox{\texttt{.}}}{\hbox{\texttt{.}}}{\hbox{\texttt{.}}}}

\mbox{\hphantom{\Scribtexttt{xx}}}\RktPn{[}\RktPn{(}\RktSym{seq}\mbox{\hphantom{\Scribtexttt{x}}}\RktRdr{,}\RktSym{e1}\mbox{\hphantom{\Scribtexttt{x}}}\RktRdr{,}\RktSym{e2}\RktPn{)}

\mbox{\hphantom{\Scribtexttt{xxx}}}\RktPn{(}\RktSym{let{-}values}\mbox{\hphantom{\Scribtexttt{x}}}\RktPn{(}\RktPn{[}\RktPn{(}\RktSym{e1}\mbox{\hphantom{\Scribtexttt{x}}}\RktSym{ty1}\mbox{\hphantom{\Scribtexttt{x}}}\RktSym{{\hbox{\texttt{.}}}{\hbox{\texttt{.}}}{\hbox{\texttt{.}}}{\hbox{\texttt{.}}}}\RktPn{)}\mbox{\hphantom{\Scribtexttt{x}}}\RktPn{(}\RktSym{recur}\mbox{\hphantom{\Scribtexttt{x}}}\RktSym{e1}\mbox{\hphantom{\Scribtexttt{x}}}\RktSym{{\hbox{\texttt{.}}}{\hbox{\texttt{.}}}{\hbox{\texttt{.}}}{\hbox{\texttt{.}}}}\RktPn{)}\RktPn{]}

\mbox{\hphantom{\Scribtexttt{xxxxxxxxxxxxxxxx}}}\RktPn{[}\RktPn{(}\RktSym{e2}\mbox{\hphantom{\Scribtexttt{x}}}\RktSym{ty2}\mbox{\hphantom{\Scribtexttt{x}}}\RktSym{{\hbox{\texttt{.}}}{\hbox{\texttt{.}}}{\hbox{\texttt{.}}}{\hbox{\texttt{.}}}}\RktPn{)}\mbox{\hphantom{\Scribtexttt{x}}}\RktPn{(}\RktSym{recur}\mbox{\hphantom{\Scribtexttt{x}}}\RktSym{e2}\mbox{\hphantom{\Scribtexttt{x}}}\RktSym{{\hbox{\texttt{.}}}{\hbox{\texttt{.}}}{\hbox{\texttt{.}}}{\hbox{\texttt{.}}}}\RktPn{)}\RktPn{]}\RktPn{)}

\mbox{\hphantom{\Scribtexttt{xxxxx}}}\RktPn{(}\RktSym{cond}

\mbox{\hphantom{\Scribtexttt{xxxxxxx}}}\RktPn{[}\RktPn{(}\RktSym{predicate{-}implies{\hbox{\texttt{?}}}}\mbox{\hphantom{\Scribtexttt{x}}}\RktSym{ty1}\mbox{\hphantom{\Scribtexttt{x}}}\RktVal{{\textquotesingle}}\RktVal{bottom}\RktPn{)}\mbox{\hphantom{\Scribtexttt{x}}}\RktPn{(}\RktSym{unwrapped{-}error}\mbox{\hphantom{\Scribtexttt{x}}}\RktSym{{\hbox{\texttt{.}}}{\hbox{\texttt{.}}}{\hbox{\texttt{.}}}{\hbox{\texttt{.}}}}\mbox{\hphantom{\Scribtexttt{x}}}\RktSym{e1}\RktPn{)}\RktPn{]}

\mbox{\hphantom{\Scribtexttt{xxxxxxx}}}\RktPn{[}\RktSym{else}\mbox{\hphantom{\Scribtexttt{x}}}\RktPn{(}\RktSym{values}\mbox{\hphantom{\Scribtexttt{x}}}\RktPn{(}\RktSym{make{-}seq}\mbox{\hphantom{\Scribtexttt{x}}}\RktSym{{\hbox{\texttt{.}}}{\hbox{\texttt{.}}}{\hbox{\texttt{.}}}{\hbox{\texttt{.}}}}\mbox{\hphantom{\Scribtexttt{x}}}\RktSym{e1}\mbox{\hphantom{\Scribtexttt{x}}}\RktSym{e2}\RktPn{)}\mbox{\hphantom{\Scribtexttt{x}}}\RktSym{ty}\mbox{\hphantom{\Scribtexttt{x}}}\RktSym{{\hbox{\texttt{.}}}{\hbox{\texttt{.}}}{\hbox{\texttt{.}}}{\hbox{\texttt{.}}}}\RktPn{)}\RktPn{]}\RktPn{)}\RktPn{)}\RktPn{]}\mbox{\hphantom{\Scribtexttt{x}}}\RktSym{{\hbox{\texttt{.}}}{\hbox{\texttt{.}}}{\hbox{\texttt{.}}}}\RktPn{)}\end{SingleColumn}\end{RktBlk}\end{SCodeFlow}\end{smaller}

In the code snippet, the call to \RktSym{unwrapped{-}error} aims to
   adjust \RktSym{e1} in some contexts to preserve non{-}tail positioning.
   The adjustment does not apply to unreachable. Hence, same as rule P.2, the pass
   replaces the \RktSym{seq} form with unreachable.

Rule P.3 corresponds to how the \Scribtexttt{cp0} pass simplifies code with unused
 bindings using \Scribtexttt{begin} forms. Specifically, a loop of the
 \RktSym{letify} function of the \Scribtexttt{cp0} pass gathers unused bindings and uses
 \RktSym{residualize{-}seq} to lift their right{-}hand{-}side expressions to an
 enclosing sequence.

Rules P.4 and P.5 correspond to the \RktSym{fold{-}call/other} function
 of the \Scribtexttt{cptypes} pass. The function recursively optimizes
 a list containing the function and argument expressions
 of an application.  If, after the recursive optimizations, any of these
 sub{-}expressions has type \RktVal{{\textquotesingle}}\RktVal{bottom}, the application is replaced by
 that non{-}returning sub{-}expression.
Unlike the calculus that has a fixed order of evaluation, order of
evaluation in Chez Scheme is unspecified. As a result, the pass can treat all
other sub{-}expressions as being downstream the non{-}returning
sub{-}expression. Therefore, it discards them similar to the way rules P.4
and P.5 simplify applications in the calculus.

In sum, the rules of the calculus accurately describe the five places in
the source of the Racket compiler that take advantage of unreachable for
optimizations.

\sectionNewpage

\Ssection{Unreachable in LLVM}{Unreachable in LLVM}\label{t:x28part_x22secx3allvmx2drelationx22x29}

The LLVM Intermediate Representation (LLVM IR) represents programs as
control{-}flow graphs (CFGs).
Unlike expression{-}based languages, CFGs break down programs
into basic blocks that each consists of of a linear sequence of
instructions. Transfer of control between basic blocks is dictated
by the  edges of the CFG.

As such, compile{-}time transformations are no longer
transformations of the structure of expressions but transformations of
the structure of the CFG.
Despite this difference, CFG transformations due to \texorpdfstring{\ensuremath{\ouro}}{unreachable} follow
similar intuitions as those in our calculus.
Since the \texorpdfstring{\ensuremath{\ouro}}{unreachable} is never executed,
the LLVM compiler can prune any branches of control transfer operations that
lead to \texorpdfstring{\ensuremath{\ouro}}{unreachable} and erase code preceding \texorpdfstring{\ensuremath{\ouro}}{unreachable}.

In this section, we establish a connection between these transformations
and our calculus. Specifically, we prove a
function that performs those two transformations correct, using the $\arrC$ reductions of
our calculus.

\Ssubsection{\texorpdfstring{\ensuremath{\ouro}}{unreachable} Transformations in LLVM, Informally}{\texorpdfstring{\ensuremath{\ouro}}{unreachable} Transformations in LLVM, Informally}\label{t:x28part_x22secx3allvmx2dinformallyx22x29}

\begin{Figure}\begin{Centerfigure}\begin{FigureInside}\raisebox{-2.7379687499999985bp}{\makebox[375.06218750000005bp][l]{\includegraphics[trim=2.4000000000000004 2.4000000000000004 2.4000000000000004 2.4000000000000004]{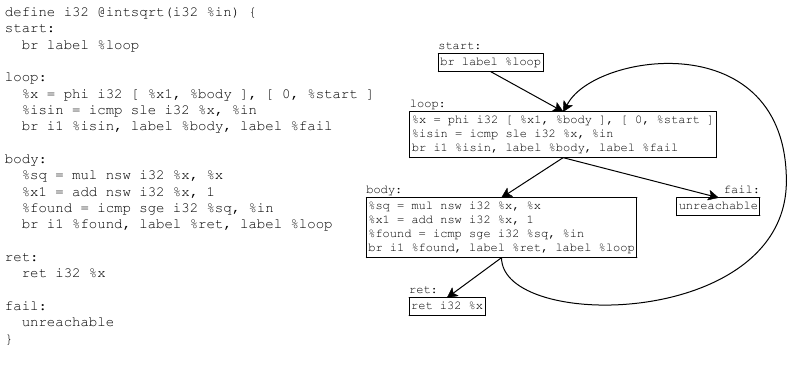}}}\end{FigureInside}\end{Centerfigure}

\Centertext{\Legend{\FigureTarget{\label{t:x28counter_x28x22figurex22_x22figx3allvmx2dbitcodex22x29x29}\textsf{Fig.}~\textsf{11}. }{t:x28counter_x28x22figurex22_x22figx3allvmx2dbitcodex22x29x29}\textsf{LLVM Code Illustrating the Control Flow Graph Representation}}}\end{Figure}

LLVM uses static single{-}assignment (SSA) form, thus
a basic block in LLVM starts with a (possibly empty) sequence of
$\phi$ nodes that assign different values to variables
depending on the predecessor executed at run time. The $\phi$ nodes are
followed by a series of instructions that typically define variables using the result of their computation.
The last instruction of a basic block is called a terminator designating where the control
should transfer to afterwards.
Examples of terminators in LLVM include a return statement, an unconditional branch, a
conditional branch, and the $\ouro$ instruction.

As a concrete example, figure~\hyperref[t:x28counter_x28x22figurex22_x22figx3allvmx2dbitcodex22x29x29]{\FigureRef{11}{t:x28counter_x28x22figurex22_x22figx3allvmx2dbitcodex22x29x29}} contains an LLVM IR
function, \Scribtexttt{@intsqrt}, and its CFG.
The \Scribtexttt{@intsqrt} function loops through the natural numbers, returning the first integer
that is no smaller than the square root of \Scribtexttt{\%in}.
The block \Scribtexttt{loop{\hbox{\texttt{:}}}} starts with a $\phi$ node that assigns a value to variable \Scribtexttt{\%x},
the loop induction variable, counting up through the naturals.
If the run{-}time predecessor of block \Scribtexttt{loop{\hbox{\texttt{:}}}} is block
\Scribtexttt{start{\hbox{\texttt{:}}}}, then \Scribtexttt{\%x} is \Scribtexttt{0}; if the predecessor is
block \Scribtexttt{body{\hbox{\texttt{:}}}} (following the long curved back edge),
\Scribtexttt{\%x} becomes equal to \Scribtexttt{\%x1}, which will hold $\Scribtexttt{\%x}+1$.
Subsequently, the loop checks the exit condition to determine whether
\Scribtexttt{\%x {\Stttextless}= \%in}, and the \Scribtexttt{br} terminator
either transfers control to the body of loop or to \Scribtexttt{fail{\hbox{\texttt{:}}}}
Note that the loop
increments \Scribtexttt{\%x}, and the annotation \Scribtexttt{nsw} signals that
no overflow happens,\NoteBox{\NoteContent{More precisely, signed overflow would
produce a \Scribtexttt{poison} value.}} hence, control transfers to \Scribtexttt{fail{\hbox{\texttt{:}}}} only if the input
is negative.

As the block \Scribtexttt{fail{\hbox{\texttt{:}}}} contains the \Scribtexttt{unreachable} instruction,
LLVM assumes that the block is never executed.
Thus, LLVM removes it and its incoming edge, leaving only
an \emph{unconditional} \Scribtexttt{br} instruction in \Scribtexttt{loop{\hbox{\texttt{:}}}}.
This simplification allows LLVM to erase the now{-}dead \Scribtexttt{\%isin}
definition, and merge \Scribtexttt{loop{\hbox{\texttt{:}}}} and \Scribtexttt{body{\hbox{\texttt{:}}}} since they are each other{'}s
unique predecessor and successor.

\Ssubsection{\texorpdfstring{\ensuremath{\ouro}}{unreachable} Transformations in Extended Vminus}{\texorpdfstring{\ensuremath{\ouro}}{unreachable} Transformations in Extended Vminus}\label{t:x28part_x22secx3allvmx2dtransformsx22x29}

\newcommand{\Changed}{\mathit{changed}}
\begin{figure}[t]
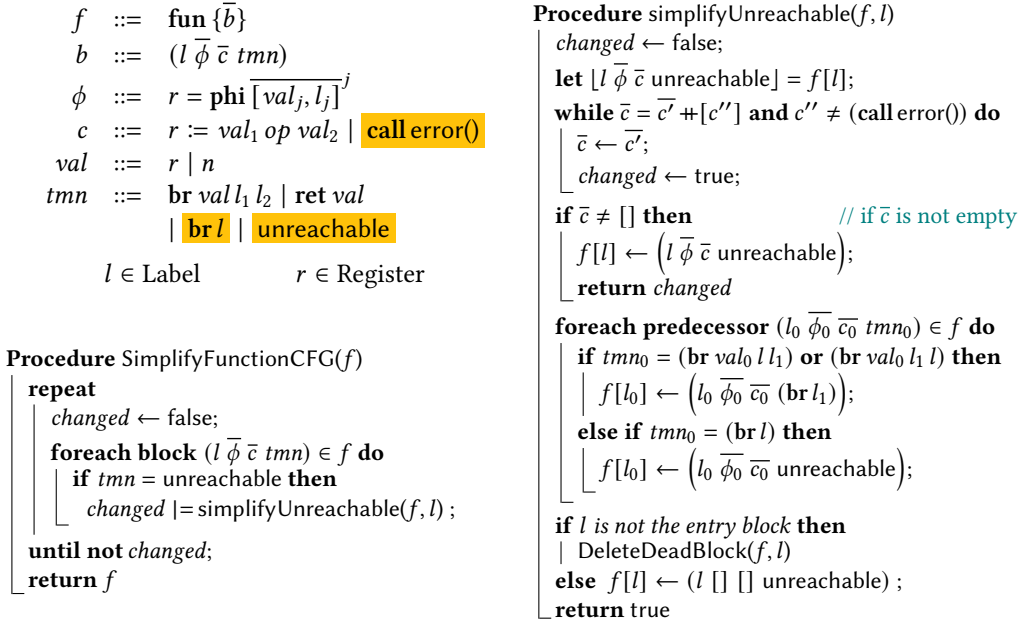

  \setlength{\algomargin}{0.1em}
  \SetInd{0.2em}{0.6em}
  \begin{minipage}{\textwidth}
    \begin{minipage}[t]{0.5\textwidth}
    \NTVminus
    \small
\begin{algorithm}[H]
  % \KwIn{A function $f$}
  % \KwResult{An updated version of $f$}
  \Proc{\SimplifyFunctionCFG{$f$}}{
    \Repeat{$\KwSty{not}\,\Changed$}{
      $\Changed\gets\mathsf{false}$\;
      \ForEach{\KwSty{block} $(l\; \overline{\phi}\;\overline{c}\;\Tmn) \in f$}{
        \lIf{$\Tmn=\ouro$}{
          $\Changed \OrAssign \SimplifyUnreachable{$f, l$}$
        }
      }
    }
    \KwRet{$f$}
  }
\end{algorithm}
    \end{minipage}
    \begin{minipage}[t]{0.5\textwidth}
      \small\vspace*{-1em}
      \begin{algorithm}[H]
        \Proc{\SimplifyUnreachable{$f,l$}}{
          $\Changed\gets\False$\;
          \KwSty{let} $\lfloor l\; \overline{\phi}\;\overline{c}\;\ouro\rfloor=f[l]$\;
          \While{$\overline{c}=\overline{c'} \lappend [c'']$ \KwSty{and} $c''\neq (\CallError)$}{
            $\overline{c} \gets \overline{c'}$\;
            $\Changed\gets\True$\;
          }
          \AIf(\hfill \textcolor{\CommentColor}{// if $\overline{c}$ is not empty}){$\overline{c}\ne{} []$}{
            $f[l]\gets \left( l\; \overline{\phi}\;\overline{c}\;\ouro \right)$\;
            \KwRet{$\Changed$}
          }
          % $f[l] \gets \left( l\; []\;[]\;\ouro \right)$\;
          \ForEach{\KwSty{predecessor} $(l_0\;\overline{\phi_0}\;\overline{c_0}\;\Tmn_0) \in f$}{
            \uIf{$\Tmn_0=(\BrCond{\Val_0}{l}{l_1})$ \KwSty{or}
            $(\BrCond{\Val_0}{l_1}{l})$}{
              $f[l_0] \gets
                \left( l_0\;\overline{\phi_0}\;\overline{c_0}\;(\BrUncond{l_1}) \right)$\;
            }
            \ElseIf{$\Tmn_0=(\BrUncond{l})$}{
              $f[l_0] \gets
                \left( l_0\;\overline{\phi_0}\;\overline{c_0}\;\ouro \right)$\;
            }
          }
          \uIf{$l$ is not the entry block}{
            $\DeleteDeadBlock{$f,l$}$
          }
          \lElse{
            $f[l] \gets \left(l\;[]\;[]\;\ouro\right)$
          }
          \KwRet{$\True$}
        }
      \end{algorithm}
    \end{minipage}
  \end{minipage}
  \caption{The syntax of Extended Vminus and the $\ouro$ transformations}
  \label{fig:llvm-simplifycfg}
\end{figure}

We designed a transformation inspired by the optimization passes added along with the
\Scribtexttt{unreachable} instruction to the LLVM\Autobibref{~(\hyperref[t:x28autobib_x22_LLVM_ProjectLLVM_13x2e0x2e0_Release_Notes2021httpsx3ax2fx2freleasesx2ellvmx2eorgx2f13x2e0x2e0x2fdocsx2fReleaseNotesx2ehtmlRetrievedx3a_Sepx2c_2021x22x29]{\AutobibLink{LLVM Project}} \hyperref[t:x28autobib_x22_LLVM_ProjectLLVM_13x2e0x2e0_Release_Notes2021httpsx3ax2fx2freleasesx2ellvmx2eorgx2f13x2e0x2e0x2fdocsx2fReleaseNotesx2ehtmlRetrievedx3a_Sepx2c_2021x22x29]{\AutobibLink{2021}})}
codebase in 2004.\NoteBox{\NoteContent{\href{https://github.com/llvm/llvm-project/commit/5edb2f32d00d39f7d9fd98b90ff440b5dbbdcb45}{\Snolinkurl{https://github.com/llvm/llvm-project/commit/5edb2f32d00d39f7d9fd98b90ff440b5dbbdcb45}}}}

That code, upon spotting an \Scribtexttt{unreachable} terminator in the
current block, erases all preceding instructions and
prunes all the incoming edges, and so our transformation follows suit.

At the same time, other
commits\NoteBox{\NoteContent{\href{https://github.com/llvm/llvm-project/commit/8ba9ec9bbbf6625b149ae1ceeb46876553cd2f11}{\Snolinkurl{https://github.com/llvm/llvm-project/commit/8ba9ec9bbbf6625b149ae1ceeb46876553cd2f11}}
and \href{https://github.com/llvm/llvm-project/commit/a67dd32004bcc1a0a6fa2f0342e584187f5a403d}{\Snolinkurl{https://github.com/llvm/llvm-project/commit/a67dd32004bcc1a0a6fa2f0342e584187f5a403d}}}}
to LLVM added the ability to replace instructions
that obviously cannot be reached with \Scribtexttt{unreachable}. Our
transformation does not capture this second capability of LLVM.

\Cref{fig:llvm-simplifycfg} shows our transformation.
We extend the syntax of Vminus\Autobibref{~(\hyperref[t:x28autobib_x22_Jianzhou_Zhaox2c__Santosh_Nagarakattex2c__Milo_Mx2eKx2e_Martinx2c_and__Steve_ZdancewicFormal_Verification_of_SSAx2dBased_Optimizations_for_LLVMIn_Procx2e_ACM_Conference_on_Programming_Language_Design_and_Implementationx2c_PLDI_x2713x2c_ppx2e_175x2dx2d1862013httpsx3ax2fx2fdoix2eorgx2f10x2e1145x2f2491956x2e2462164x22x29]{\AutobibLink{Zhao et al\Sendabbrev{.}}} \hyperref[t:x28autobib_x22_Jianzhou_Zhaox2c__Santosh_Nagarakattex2c__Milo_Mx2eKx2e_Martinx2c_and__Steve_ZdancewicFormal_Verification_of_SSAx2dBased_Optimizations_for_LLVMIn_Procx2e_ACM_Conference_on_Programming_Language_Design_and_Implementationx2c_PLDI_x2713x2c_ppx2e_175x2dx2d1862013httpsx3ax2fx2fdoix2eorgx2f10x2e1145x2f2491956x2e2462164x22x29]{\AutobibLink{2013}})}
with additional terminators
and commands highlighted in yellow.\NoteBox{\NoteContent{We also omit the type annotations and use integers
as the only constant values.}}
Vminus is a minimal model of LLVM IR for studying
SSA{-}based optimizations. While Vminus omits features like memory access
and function calls, it is sufficient for our purposes
since LLVM primarily transforms \texorpdfstring{\ensuremath{\ouro}}{unreachable} instructions
by rearranging basic blocks and erasing commands.

A program in Vminus is a function ($f$) that contains
a list of basic blocks ($\overline{b}$). Each basic block is a 4{-}tuple
consisting of a label ($l$), a list of $\phi$ nodes ($\overline{\phi}$),
a list of commands ($\overline{c}$) and a terminator ($\Tmn$).
A command $c$ either assigns the result of a binary operation
to a fresh variable or calls the \textsf{error} function
to end the execution.
We use $\CallError$ to model commands that \emph{do not}
transfer control to the next command
in the basic block.\NoteBox{\NoteContent{$\CallError$ terminates the program
much like the \Scribtexttt{exit} function. It is unrelated to exceptions.}}
Following \Autobibref{\hyperref[t:x28autobib_x22_Jianzhou_Zhaox2c__Santosh_Nagarakattex2c__Milo_Mx2eKx2e_Martinx2c_and__Steve_ZdancewicFormal_Verification_of_SSAx2dBased_Optimizations_for_LLVMIn_Procx2e_ACM_Conference_on_Programming_Language_Design_and_Implementationx2c_PLDI_x2713x2c_ppx2e_175x2dx2d1862013httpsx3ax2fx2fdoix2eorgx2f10x2e1145x2f2491956x2e2462164x22x29]{\AutobibLink{Zhao et al\Sendabbrev{.}}}~(\hyperref[t:x28autobib_x22_Jianzhou_Zhaox2c__Santosh_Nagarakattex2c__Milo_Mx2eKx2e_Martinx2c_and__Steve_ZdancewicFormal_Verification_of_SSAx2dBased_Optimizations_for_LLVMIn_Procx2e_ACM_Conference_on_Programming_Language_Design_and_Implementationx2c_PLDI_x2713x2c_ppx2e_175x2dx2d1862013httpsx3ax2fx2fdoix2eorgx2f10x2e1145x2f2491956x2e2462164x22x29]{\AutobibLink{2013}})}{'}s notation, an overlined non{-}terminal
represents a list of the underlying non{-}terminal. The notation $l.i$
denotes the $i${-}th command in block $l$, and
$f[l]=\lfloor b\rfloor$ represents a look up of the block labeled $l$
in $f$. The notation $\lfloor b\rfloor$ asserts that the lookup succeeds
with the result $b$.

\SimplifyFunctionCFG{$f$} is the entry point of our transformation.
It iteratively updates $f$ until reaching a fixed point.
In each iteration, \FuncSty{SimplifyFunctionCFG} scans through the basic blocks to
identify basic blocks terminating with $\ouro$.
For any such basic block $l$, \FuncSty{SimplifyFunctionCFG}
invokes \FuncSty{simplifyUnreachable}
to remove commands and paths leading up to $l${'}s $\ouro$.
This includes simplifying the commands in block $l$ and possibly
the branch instructions in the predecessors of $l$.

The \SimplifyUnreachable{$f,l$} transformation exploits
the intuition that an \texorpdfstring{\ensuremath{\ouro}}{unreachable} instruction should never be reached;
otherwise the behavior of the program is undefined.
Specifically, \FuncSty{simplifyUnreachable} does two simplifications.
First, it scans the commands in $l$ from the end of the list. If the scanned
command ($c''$) \emph{always} transfers control to the next
instruction in $l$, it is erased from $l$.
The scan ends when \FuncSty{simplifyUnreachable} hits a command that
\emph{may not} transfer control the next command.
Second, \FuncSty{simplifyUnreachable} prunes the incoming edges.
For each predecessor $l_0$ of $l$, \FuncSty{simplifyUnreachable}
removes $l$ from the branch targets of $l_0$.
If $l_0${'}s terminator is an unconditional branch,
\FuncSty{simplifyUnreachable} replaces it with $\ouro$.

\Ssubsection{Translating Extended Vminus to The Unreachable Calculus}{Translating Extended Vminus to The Unreachable Calculus}\label{t:x28part_x22secx3allvmx2dkelseyx22x29}

\begin{figure}[t]
 \[\begin{array}{@{}l@{}}
  \KHjumpfn \vspace{0.5em} \\
  \KHcmdfn \vspace{0.5em} \\
  \KHtermfn
  \end{array}\]
  \caption{Translating Extended Vminus Functions to the Unreachable Calculus}
  \label{fig:kelsey-h}
\end{figure}

The transformation from \SecRef{\SectionNumberLink{t:x28part_x22secx3allvmx2dtransformsx22x29}{8.2}}{\texorpdfstring{\ensuremath{\ouro}}{unreachable} Transformations in Extended Vminus}
is based on the same insights about \texorpdfstring{\ensuremath{\ouro}}{unreachable} as the $\arrU$ and
$\arrP$ reductions of our calculus.
As such, it should be correct despite the fact that it belongs to
a different linguistic setting than that of the calculus.

To use the calculus to prove that it is correct, we leverage \Autobibref{\hyperref[t:x28autobib_x22_Richard_Ax2e_KelseyA_Correspondence_between_Continuation_Passing_Style_and_Static_Single_Assignment_FormIn_Procx2e_Papers_from_the_1995_ACM_SIGPLAN_Workshop_on_Intermediate_Representationsx2c_IR_x2795x2c_ppx2e_13x2dx2d221995httpsx3ax2fx2fdoix2eorgx2f10x2e1145x2f202530x2e202532x22x29]{\AutobibLink{Kelsey}}~(\hyperref[t:x28autobib_x22_Richard_Ax2e_KelseyA_Correspondence_between_Continuation_Passing_Style_and_Static_Single_Assignment_FormIn_Procx2e_Papers_from_the_1995_ACM_SIGPLAN_Workshop_on_Intermediate_Representationsx2c_IR_x2795x2c_ppx2e_13x2dx2d221995httpsx3ax2fx2fdoix2eorgx2f10x2e1145x2f202530x2e202532x22x29]{\AutobibLink{1995}})}{'}s algorithm that
translates programs from SSA to $\lambda$ expressions and back.
In detail, we adapt their algorithm
to work on Extended Vminus programs and use it to show that the transformation
in \cref{fig:llvm-simplifycfg} given a program A produces program B
such that the  translation of A
is equal to the translation of B under $\arrC$ reductions.

\cref{fig:kelsey-h} gives the complete definition of the translation.
$\KHproc$ translates a function $f$ in Extended Vminus
to a (closed) function $\Func{x}{e}$ in the \texorpdfstring{\ensuremath{\ouro}}{unreachable} calculus.
It comprises three auxiliary functions:
$\KHjump$ translates a basic block to a $\lambda$ function,
$\KHcmd$ takes a list of commands and morally translates them into
a nested \textsf{let} expression, and finally
$\KHterm$ translates the terminator of a block into an appropriate expression in our calculus.
In the definition of $\KHcmd$, $(\textsf{let}\, ([x\, e_1])\, e_2)$ is a syntactic sugar for
$\App{(\Func{x}{e_2})}{e_1}$ and \textsf{letrec} is implemented using the Y combinator.
Appendix L details the implementation.

The key idea behind $\KHproc$ is the encoding of the CFG of a given
program. $\KHjump$ encodes each basic block as a single function, while
$\KHterm$ adds edges that represent branches using function applications.
For each basic block handled by $\KHjump$,
the variables defined by the $\phi$ nodes become the formal parameters
of the resulting function, and the incoming values of the
$\phi$ nodes are turned into the arguments in the function applications
that $\KHterm$ uses to encode branch instructions.

Having translated individual basic block and the edges of the CFG,
the translation assembles the results of $\KHjump$ into a single expression
while respecting the variable scoping rules of Vminus programs.
To ensure that variables are defined before they are referenced, Vminus requires
a variable definition to appear in a block that dominates
all the blocks that use the variable.
Intuitively, block $l$ \label{t:x28tech_x22dominatex22x29}\textit{dominates} block $l'$,
written as $\Dom{l}{l'}$, if block $l$ appears
in every path from the entry to block $l'$.
The dominance relation between basic blocks form a so{-}called  \label{t:x28tech_x22dominator_treex22x29}\textit{dominator tree}\Autobibref{~(\hyperref[t:x28autobib_x22_Thomas_Lengauer_and__Robert_Endre_TarjanA_Fast_Algorithm_for_Finding_Dominators_in_a_FlowgraphACM_Transactions_on_Programming_Languages_and_Systems_1x281x29x2c_ppx2e_121x2dx2d1411979httpsx3ax2fx2fdoix2eorgx2f10x2e1145x2f357062x2e357071x22x29]{\AutobibLink{Lengauer and Tarjan}} \hyperref[t:x28autobib_x22_Thomas_Lengauer_and__Robert_Endre_TarjanA_Fast_Algorithm_for_Finding_Dominators_in_a_FlowgraphACM_Transactions_on_Programming_Languages_and_Systems_1x281x29x2c_ppx2e_121x2dx2d1411979httpsx3ax2fx2fdoix2eorgx2f10x2e1145x2f357062x2e357071x22x29]{\AutobibLink{1979}}; \hyperref[t:x28autobib_x22Edward_Sx2e_Lowry_and_Cx2e_Wx2e_MedlockObject_Code_OptimizationCommunications_of_the_ACM_12x281x29x2c_ppx2e_13x2dx2d221969httpsx3ax2fx2fdoix2eorgx2f10x2e1145x2f362835x2e362838x22x29]{\AutobibLink{Lowry and Medlock}} \hyperref[t:x28autobib_x22Edward_Sx2e_Lowry_and_Cx2e_Wx2e_MedlockObject_Code_OptimizationCommunications_of_the_ACM_12x281x29x2c_ppx2e_13x2dx2d221969httpsx3ax2fx2fdoix2eorgx2f10x2e1145x2f362835x2e362838x22x29]{\AutobibLink{1969}})}: if block $l$ dominates $l'$ then
$l$ is an ancestor of $l'$ in the dominator tree. Conversely,
if $l$ has children $l_1,\dots,l_m$ in the dominator tree,
then $l$ dominates all $l_i$ and each function $\KHjump(f,l_i)$
should be able to reference the variables that block $l$ defines.
Therefore, $\KHcmd$ arranges the results of $\KHjump$ into nested \textsf{letrec}s
in accordance with the dominator tree of the control flow graph
to preserve correct scoping of the variables.

With the definition of the translation in hand, we can prove the
correctness of its correctness by establishing the
correctness of  \FuncSty{simplifyUnreachable}:

\begin{theorem}\label{thm:simplify-unreachable}Let $f[l]=\lfloor l\; \overline{\phi}\; \overline{c}\; \ouro\rfloor$
and $f'$ be the new function after running the transformation $\SimplifyUnreachable{$f,l$}$.
If $l$ is reachable from the entry point of $f$
then $\KHproc(f)\arrC^*\KHproc(f')$.\end{theorem}

\begin{proof}
Because block $l$ is reachable from the entry point of $f$ and
has no successor, it must be a leaf node in the dominator tree.
Thus, $\KHcmd(f,l,\overline{c})$ is a subexpression of $\KHproc(f)$.
To analyze how \FuncSty{simplifyUnreachable} changes $f$, we take cases on whether
$\overline{c}$ contains $\CallError$ or not.

If $\overline{c}$ does not contain $\CallError$,
\FuncSty{simplifyUnreachable} prunes all incoming edges of block $l$
and deletes the entire block from $f$.
Let $f[l_0]=\lfloor l_0\; \overline{\phi_0}\; \overline{c_0}\; \Tmn_0\rfloor$
be any predecessor of $l$. If $\Tmn_0=(\BrCond{\Val_0}{l}{l_1})$, its translation
is $\KHterm(f,l_0)=\If{\Val_0}{(\App{l}{\overline{\Val}})}{(\App{l_1}{\overline{\Val_1}})}$
where the arguments are extracted from the $\phi$ nodes in block $l$ and $l_1$.

By inlining the translation of block $l$ with a series of $\arrC$
reductions, we obtain the expression
$\If{\Val_0}{(\App{\KHjump(f,l)}{\overline{\Val}})}{(\App{l_1}{\overline{\Val_1}})}$.
However, the body of $\KHjump(f,l)$, $\KHcmd(f,l,\overline{c})$, reduces to $\ouro$
under $\arrC$ since it is a nested \textsf{let} expression whose body, $\KHterm(f,l)$, is $\ouro$.
Thus the entire $\Ifk$ expression simplifies to $(\App{l_1}{\overline{\Val_1}})$,
which is precisely the translation of $(\BrUncond{l_1})$.

The case where $\Tmn_0=(\BrUncond{l})$ is also similar.

Finally, after all predecessors of $l$ are updated, $\KHjump(f,l)$ has no reference
and thus its binding can be dropped, resulting $\KHproc(f')$.

In this proof, we have used several \textsf{letrec} identities such as inlining a definition
and dropping an unreferenced binding. Appendix L proves these identities
using the $\arrM$ rules from Figure~\hyperref[t:x28counter_x28x22figurex22_x22figx3acontextpreservingrulesx22x29x29]{\FigureRef{9}{t:x28counter_x28x22figurex22_x22figx3acontextpreservingrulesx22x29x29}}.

The case where $\overline{c}$ includes $\CallError$ is similar to the
simplification of $\KHcmd(f,l,\overline{c})$ in the previous case
except that $\ouro$ stops erasing the \textsf{let} bindings after reaching
the expression $\error{}$.
Say $\overline{c}$ equals $\overline{c'}\lappend [\CallError]\lappend \overline{c''}$
such that $\overline{c''}$ contains no $\CallError$ commands, we know that
$\SimplifyUnreachable{$f,l$}$ changes block $l$ to
$(l\; \overline{\phi}\; (\overline{c'}\lappend [\CallError])\; \ouro)$.
Therefore we need to prove
\[
\KHcmd(f,l,\overline{c'}\lappend [\CallError]\lappend \overline{c''}) \arrC^*
\KHcmd(f',l,\overline{c'}\lappend [\CallError]).
\]
Note that $\KHterm(f,l)=\KHterm(f',l)=\ouro$,  so this is
straightforward as $\KHcmd$ translates $\overline{c''}$ to
a nested \textsf{let} expression whose body is just $\ouro$.
Thus \FuncSty{simplifyUnreachable} preserves the behavior
of the program.
\end{proof}

\medskip
 As a final remark in this section, the approach to proving CFG
 transformations correct via translation to the $\lambda$ calculus does not scale
 to proving realistic compilers correct. Of course, this is not the goal
 of this section but we discuss it here to eliminate any confusion. Beyond
 the obvious shortcoming that the $\lambda$ calculus and imperative features are
 not well{-}aligned, there is an additional and subtle technical challenge
 at play.  While the reductions of the  \texorpdfstring{\ensuremath{\ouro}}{unreachable} calculus leave the
 overall structure of expressions unchanged, \texorpdfstring{\ensuremath{\ouro}}{unreachable} tranformations in
 Extended Vminus, and LLVM, can modify the dominator tree of a program in
 complex ways. Hence, relating the result of a series of reductions with
 the result of an \texorpdfstring{\ensuremath{\ouro}}{unreachable} tranformations in Extended Vminus requires
 equating expressions with arbitrarily different structure. Put
 differently, the compile{-}time semanctics of the \texorpdfstring{\ensuremath{\ouro}}{unreachable} calculus do
 not map directly to the \texorpdfstring{\ensuremath{\ouro}}{unreachable} tranformations in Extended Vminus, at
 least via translations, like Kelsey{'}s, that depend on the dominator tree
 of the input program.  We conjecture that this discrepancy is due to the
 translation, but we cannot exclude a misalignment between the \texorpdfstring{\ensuremath{\ouro}}{unreachable}
 calculus and the CFG{-}based world of Extended Vminus and LLVM.

\sectionNewpage

\Ssection{Related Work}{Related Work}\label{t:x28part_x22Relatedx5fWorkx22x29}

Our work is the first that develops an equational theory for \texorpdfstring{\ensuremath{\ouro}}{unreachable}.

Other techniques that examine the correctness of compilers have to
also deal, one way or another, with the semantics of \texorpdfstring{\ensuremath{\ouro}}{unreachable} and
other undefined behaviors.
CompCert\Autobibref{~(\hyperref[t:x28autobib_x22Xavier_LeroyFormal_verification_of_a_realistic_compilerCommunications_of_the_ACM_52x2c_ppx2e_7x3a107x2dx2d7x3a1152009httpsx3ax2fx2fdoix2eorgx2f10x2e1145x2f1538788x2e1538814x22x29]{\AutobibLink{Leroy}} \hyperref[t:x28autobib_x22Xavier_LeroyFormal_verification_of_a_realistic_compilerCommunications_of_the_ACM_52x2c_ppx2e_7x3a107x2dx2d7x3a1152009httpsx3ax2fx2fdoix2eorgx2f10x2e1145x2f1538788x2e1538814x22x29]{\AutobibLink{2009\AutobibLink{a},\AutobibLink{b}}})} gives semantics to undefined behavior
implicitly, by specifying defined behavior with a co{-}inductive structure.
Based on this structure, the CompCert project proves correct whole{-}program
transformations in a realistic C compiler. While the
linguistic setting and the scale of CompCert are not comparable with
this work, the reductions of our calculus equate (open)
expressions in all contexts rather than whole programs.
Furthermore, the correctness of our reductions assumes the intuitive
predicate $\neg\Lundefp$, instead of CompCert{'}s co{-}inductive definition of
defined behavior.

Similar to CompCert, \Autobibref{\hyperref[t:x28autobib_x22Ralf_Jungx2c_Hoanghai_Dangx2c_Jeehoon_Kangx2c_and_Derek_DreyerStacked_borrowsx3a_an_aliasing_model_for_RustProceedings_of_the_ACM_on_Programming_Languages_x28POPLx29_4x2c_ppx2e_41x3a1x2dx2d41x3a322020httpsx3ax2fx2fdoix2eorgx2f10x2e1145x2f3371109x22x29]{\AutobibLink{Jung et al\Sendabbrev{.}}}~(\hyperref[t:x28autobib_x22Ralf_Jungx2c_Hoanghai_Dangx2c_Jeehoon_Kangx2c_and_Derek_DreyerStacked_borrowsx3a_an_aliasing_model_for_RustProceedings_of_the_ACM_on_Programming_Languages_x28POPLx29_4x2c_ppx2e_41x3a1x2dx2d41x3a322020httpsx3ax2fx2fdoix2eorgx2f10x2e1145x2f3371109x22x29]{\AutobibLink{2020}})} give a definition for defined
behavior in Rust, namely the Stacked Borrows pattern. Any
program that violates this pattern is considered to exhibit undefined behavior.
The authors show that  Stacked Borrows is sufficient to validate
optimizations in the Rust compiler involving both safe and unsafe code, and that it admits realistic Rust
programs. Instead of an approximate predicate,
our calculus relies on a
precise definition of undefined behavior.

Vellvm\Autobibref{~(\hyperref[t:x28autobib_x22_Yannick_Zakowskix2c__Calvin_Beckx2c__Irene_Yoonx2c__Ilia_Zaichukx2c__Vadim_Zalivax2c_and__Steve_ZdancewicModularx2c_Compositionalx2c_and_Executable_Formal_Semantics_for_LLVM_IRProceedings_of_the_ACM_on_Programming_Languages_x28ICFPx29_5x2c_ppx2e_67x3a1x2dx2d67x3a302021httpsx3ax2fx2fdoix2eorgx2f10x2e1145x2f3473572x22x29]{\AutobibLink{Zakowski et al\Sendabbrev{.}}} \hyperref[t:x28autobib_x22_Yannick_Zakowskix2c__Calvin_Beckx2c__Irene_Yoonx2c__Ilia_Zaichukx2c__Vadim_Zalivax2c_and__Steve_ZdancewicModularx2c_Compositionalx2c_and_Executable_Formal_Semantics_for_LLVM_IRProceedings_of_the_ACM_on_Programming_Languages_x28ICFPx29_5x2c_ppx2e_67x3a1x2dx2d67x3a302021httpsx3ax2fx2fdoix2eorgx2f10x2e1145x2f3473572x22x29]{\AutobibLink{2021}}; \hyperref[t:x28autobib_x22_Jianzhou_Zhaox2c__Santosh_Nagarakattex2c__Milo_Mx2eKx2e_Martinx2c_and__Steve_ZdancewicFormal_Verification_of_SSAx2dBased_Optimizations_for_LLVMIn_Procx2e_ACM_Conference_on_Programming_Language_Design_and_Implementationx2c_PLDI_x2713x2c_ppx2e_175x2dx2d1862013httpsx3ax2fx2fdoix2eorgx2f10x2e1145x2f2491956x2e2462164x22x29]{\AutobibLink{Zhao et al\Sendabbrev{.}}} \hyperref[t:x28autobib_x22_Jianzhou_Zhaox2c__Santosh_Nagarakattex2c__Milo_Mx2eKx2e_Martinx2c_and__Steve_ZdancewicFormalizing_the_LLVM_Intermediate_Representation_for_Verified_Program_TransformationsIn_Procx2e_ACM_Symposium_on_Principles_of_Programming_Languagesx2c_POPL_x2712x2c_ppx2e_427x2dx2d4402012httpsx3ax2fx2fdoix2eorgx2f10x2e1145x2f2103656x2e2103709x22x29]{\AutobibLink{2012}}, \hyperref[t:x28autobib_x22_Jianzhou_Zhaox2c__Santosh_Nagarakattex2c__Milo_Mx2eKx2e_Martinx2c_and__Steve_ZdancewicFormal_Verification_of_SSAx2dBased_Optimizations_for_LLVMIn_Procx2e_ACM_Conference_on_Programming_Language_Design_and_Implementationx2c_PLDI_x2713x2c_ppx2e_175x2dx2d1862013httpsx3ax2fx2fdoix2eorgx2f10x2e1145x2f2491956x2e2462164x22x29]{\AutobibLink{2013}})} is a long{-}running project for the formal verification of
transformations in LLVM. It covers various versions of undefined behavior,
including \texorpdfstring{\ensuremath{\ouro}}{unreachable}, by reducing them to a basic notion of undefined
behavior similar to the discussion in \ChapRef{\SectionNumberLink{t:x28part_x22secx3aotherx2dubx22x29}{6}}{Beyond Unreachable Code}. The latest
version of Vellvm relies on  interaction
trees\Autobibref{~(\hyperref[t:x28autobib_x22Lix2dyao_Xiax2c_Yannick_Zakowskix2c_Paul_Hex2c_Chungx2dKilx5cn____________________Hurx2c_Gregory_Malechax2c_Benjamin_Cx2e_Piercex2c_and_Stevex5cn______________________ZdancewicInteraction_Treesx3a_rRepresenting_Recursive_and_Impurex5cn____Programs_in_CoqProceedings_of_the_ACM_on_Programming_Languages_x28POPLx29_4x2c_ppx2e_51x3a1x2dx2d51x3a322020httpsx3ax2fx2fdoix2eorgx2f10x2e1145x2f3371119x22x29]{\AutobibLink{Xia et al\Sendabbrev{.}}} \hyperref[t:x28autobib_x22Lix2dyao_Xiax2c_Yannick_Zakowskix2c_Paul_Hex2c_Chungx2dKilx5cn____________________Hurx2c_Gregory_Malechax2c_Benjamin_Cx2e_Piercex2c_and_Stevex5cn______________________ZdancewicInteraction_Treesx3a_rRepresenting_Recursive_and_Impurex5cn____Programs_in_CoqProceedings_of_the_ACM_on_Programming_Languages_x28POPLx29_4x2c_ppx2e_51x3a1x2dx2d51x3a322020httpsx3ax2fx2fdoix2eorgx2f10x2e1145x2f3371119x22x29]{\AutobibLink{2020}})}, and we conjecture it can prove
correct equivalences that correspond to the compile{-}time reductions of our
calculus. However, such proofs would need to establish equalities between
denotations of LLVM code fragments, rather than the syntax{-}based
equivalences of our calculus, which we claim match the way compiler
writers reason about code through local rewriting steps.

\Autobibref{\hyperref[t:x28autobib_x22_Manjeet_Dahiya_and__Sorav_BansalModeling_Undefined_Behaviour_Semantics_for_Checking_Equivalence_Across_Compiler_OptimizationsIn_Procx2e_Haifa_Verification_Conference2017httpsx3ax2fx2fdoix2eorgx2f10x2e1007x2f978x2d3x2d319x2d70389x2d3x5f2x22x29]{\AutobibLink{Dahiya and Bansal}}~(\hyperref[t:x28autobib_x22_Manjeet_Dahiya_and__Sorav_BansalModeling_Undefined_Behaviour_Semantics_for_Checking_Equivalence_Across_Compiler_OptimizationsIn_Procx2e_Haifa_Verification_Conference2017httpsx3ax2fx2fdoix2eorgx2f10x2e1007x2f978x2d3x2d319x2d70389x2d3x5f2x22x29]{\AutobibLink{2017}})} presents a simulation relation for C programs
that takes undefined behavior into account. Their work considers a number
of different forms of undefined behavior, but not \texorpdfstring{\ensuremath{\ouro}}{unreachable}.

\Autobibref{\hyperref[t:x28autobib_x22_Melissa_MearsFunction_to_mark_unreachable_code2021httpx3ax2fx2fwwwx2eopenx2dstdx2eorgx2fjtc1x2fsc22x2fwg21x2fdocsx2fpapersx2f2021x2fp0627r6x2epdfx22x29]{\AutobibLink{Mears}}~(\hyperref[t:x28autobib_x22_Melissa_MearsFunction_to_mark_unreachable_code2021httpx3ax2fx2fwwwx2eopenx2dstdx2eorgx2fjtc1x2fsc22x2fwg21x2fdocsx2fpapersx2f2021x2fp0627r6x2epdfx22x29]{\AutobibLink{2021}})}{'}s proposals
 favor the addition of a construct like \texorpdfstring{\ensuremath{\ouro}}{unreachable} to C and C++. The authors
note the advantages of introducing such a construct for optimizations.
Similar to Racket and Rust, the authors point out that for
debugging purposes, the proposed construct could be treated as an
exception at run time.

As a final note, a considerable body of work focuses on program checkers
that detect undefined behavior, including \texorpdfstring{\ensuremath{\ouro}}{unreachable}.  Some checkers
rely on static analysis\Autobibref{~(\hyperref[t:x28autobib_x22_Will_Dietzx2c__Peng_Lix2c__John_Regehrx2c_and__Vikram_AdveUnderstanding_Integer_Overflow_in_Cx2fCx2bx2bIn_Procx2e_International_Conference_on_on_Software_Engineering2012httpsx3ax2fx2fdoix2eorgx2f10x2e1109x2fICSEx2e2012x2e6227142x22x29]{\AutobibLink{Dietz et al\Sendabbrev{.}}} \hyperref[t:x28autobib_x22_Will_Dietzx2c__Peng_Lix2c__John_Regehrx2c_and__Vikram_AdveUnderstanding_Integer_Overflow_in_Cx2fCx2bx2bIn_Procx2e_International_Conference_on_on_Software_Engineering2012httpsx3ax2fx2fdoix2eorgx2f10x2e1109x2fICSEx2e2012x2e6227142x22x29]{\AutobibLink{2012}}; \hyperref[t:x28autobib_x22Jacques_Henri_Jourdanx2c_Vincent_Laportex2c_Sandrine_Blazyx2c_Xavier_Leroyx2c_and_David_PichardieA_Formallyx2dVerified_C_Static_AnalyzerIn_Procx2e_ACM_Symposium_on_Principles_of_Programming_Languagesx2c_ppx2e_247x2dx2d2592015httpsx3ax2fx2fdoix2eorgx2f10x2e1145x2f2676726x2e2676966x22x29]{\AutobibLink{Jourdan et al\Sendabbrev{.}}} \hyperref[t:x28autobib_x22Jacques_Henri_Jourdanx2c_Vincent_Laportex2c_Sandrine_Blazyx2c_Xavier_Leroyx2c_and_David_PichardieA_Formallyx2dVerified_C_Static_AnalyzerIn_Procx2e_ACM_Symposium_on_Principles_of_Programming_Languagesx2c_ppx2e_247x2dx2d2592015httpsx3ax2fx2fdoix2eorgx2f10x2e1145x2f2676726x2e2676966x22x29]{\AutobibLink{2015}}; \hyperref[t:x28autobib_x22_Xi_Wangx2c__Nickolai_Zeldovichx2c__Mx2e_Frans_Kaashoekx2c_and__Armando_Solarx2dLezamaA_Differential_Approach_to_Undefined_Behavior_DetectionCommunications_of_the_ACM_59x283x29x2c_ppx2e_99x2dx2d1062016httpsx3ax2fx2fdoix2eorgx2f10x2e1145x2f2885256x22x29]{\AutobibLink{Wang et al\Sendabbrev{.}}} \hyperref[t:x28autobib_x22_Xi_Wangx2c__Nickolai_Zeldovichx2c__Mx2e_Frans_Kaashoekx2c_and__Armando_Solarx2dLezamaA_Differential_Approach_to_Undefined_Behavior_DetectionCommunications_of_the_ACM_59x283x29x2c_ppx2e_99x2dx2d1062016httpsx3ax2fx2fdoix2eorgx2f10x2e1145x2f2885256x22x29]{\AutobibLink{2016}})},
while others on testing\Autobibref{~(\hyperref[t:x28autobib_x22_John_RegehrBetter_Testing_With_Undefined_Behavior_Coverage2011httpsx3ax2fx2fblogx2eregehrx2eorgx2farchivesx2f388x22x29]{\AutobibLink{Regehr}} \hyperref[t:x28autobib_x22_John_RegehrBetter_Testing_With_Undefined_Behavior_Coverage2011httpsx3ax2fx2fblogx2eregehrx2eorgx2farchivesx2f388x22x29]{\AutobibLink{2011}})}.  Here we
focus on two of these works.  \Autobibref{\hyperref[t:x28autobib_x22Chris_Hathhornx2c_Chucky_Ellisonx2c_and_Grigore_Rox15fuDefining_the_Undefinedness_of_CACM_Conference_on_Programming_Language_Design_and_Implementation2015httpsx3ax2fx2fdoix2eorgx2f10x2e1145x2f2737924x2e2737979x22x29]{\AutobibLink{Hathhorn et al\Sendabbrev{.}}}~(\hyperref[t:x28autobib_x22Chris_Hathhornx2c_Chucky_Ellisonx2c_and_Grigore_Rox15fuDefining_the_Undefinedness_of_CACM_Conference_on_Programming_Language_Design_and_Implementation2015httpsx3ax2fx2fdoix2eorgx2f10x2e1145x2f2737924x2e2737979x22x29]{\AutobibLink{2015}})} describes a
model checker that detects undefined behavior based on a formal semantics
for undefined behavior in C.  RustBelt\Autobibref{~(\hyperref[t:x28autobib_x22Ralf_Jungx2c_Jacquesx2dHenri_Jourdanx2c_Robbert_Krebbersx2c_and_Derek_DreyerRustBeltx3a_securing_the_foundations_of_the_Rust_programming_languageProceedings_of_the_ACM_on_Programming_Languages_x28POPLx29_2x2c_ppx2e_66x3a1x2dx2d66x3a342018httpsx3ax2fx2fdoix2eorgx2f10x2e1145x2f3158154x22x29]{\AutobibLink{Jung et al\Sendabbrev{.}}} \hyperref[t:x28autobib_x22Ralf_Jungx2c_Jacquesx2dHenri_Jourdanx2c_Robbert_Krebbersx2c_and_Derek_DreyerRustBeltx3a_securing_the_foundations_of_the_Rust_programming_languageProceedings_of_the_ACM_on_Programming_Languages_x28POPLx29_2x2c_ppx2e_66x3a1x2dx2d66x3a342018httpsx3ax2fx2fdoix2eorgx2f10x2e1145x2f3158154x22x29]{\AutobibLink{2018}})} proves the absence
  of one kind of undefined behavior from Rust programs, data races. It
  relies on semantic type soundness, which admits programs that the
  conventional syntactic type soundness rejects. While all these
  techniques define what certain kinds of undefined behavior mean in
  different settings, they seek to eliminate unexpected undefined behavior
  as opposed to explaining the optimization opportunities that undefined
  behavior provides.

\sectionNewpage

\Ssection{Conclusion}{Conclusion}\label{t:x28part_x22Conclusionx22x29}

This paper gives a formal account of the essence of \texorpdfstring{\ensuremath{\ouro}}{unreachable}.
Specifically, it confronts head on that  \texorpdfstring{\ensuremath{\ouro}}{unreachable} is a form of
undefined behavior, and hence, a compiler may take advantage of
it to legally transform a program in a way that causes it to
evaluate differently.  For that reason, the paper presents a pair of
specifications for a programming language: one that covers how programs
evaluate normally and a separately{-}defined and strikingly simple one that covers what
transformations are legal for a compiler. We prove that, despite its
simplicity, the specification of the compiler is
correct: its rules preserve the meaning of programs (according to normal
evaluation), under the assumption that the original programs do not
exhibit undefined behavior, i.e., they do not evaluate \texorpdfstring{\ensuremath{\ouro}}{unreachable}.
Importantly, the correctness of the compile{-}time transformations
depends on this precise definition of undefined behavior, rather than an
approximation that aims to facilitate the proofs. In other words, our
formal specification of \texorpdfstring{\ensuremath{\ouro}}{unreachable} provides simple rewriting rules that
capture how compiler writers
reason about  undefined behavior and how they use this reasoning to
justify transformation implementations.

Taking a step back, we hope that our approach  can provide a
template for others who wish to formally state and prove meta{-}theoretic
properties that correctly capture aspects of undefined behavior.
While undefined behavior has a reputation as an unruly phenomenon,
our work shows that there are well{-}behaved undefined behaviors. Similar to
\texorpdfstring{\ensuremath{\ouro}}{unreachable}, we conjecture that other aspects of undefined behavior can
be described precisely and intuitively. Hence, we see this work as the first
step towards demystifying undefined behavior.

\sectionNewpage

\Ssectionstarx{References}{References}\label{t:x28part_x22docx2dbibliographyx22x29}

\begin{AutoBibliography}\begin{SingleColumn}\Autobibtarget{\label{t:x28autobib_x22CVEx2d2014x2d01602014httpsx3ax2fx2fcvex2emitrex2eorgx2fcgix2dbinx2fcvenamex2ecgix3fnamex3dcvex2d2014x2d0160The_Heartbleed_Bugx2e_Retrievedx3a_Julyx2c_2022x2e__Discovered_by_Neel_Mehta_from_Googlex2ex22x29}\Autobibentry{CVE{-}2014{-}0160. 2014. \href{https://cve.mitre.org/cgi-bin/cvename.cgi?name=cve-2014-0160}{\Snolinkurl{https://cve.mitre.org/cgi-bin/cvename.cgi?name=cve-2014-0160}} The Heartbleed Bug. Retrieved: July, 2022.  Discovered by Neel Mehta from Google.}}

\Autobibtarget{\label{t:x28autobib_x22_Manjeet_Dahiya_and__Sorav_BansalModeling_Undefined_Behaviour_Semantics_for_Checking_Equivalence_Across_Compiler_OptimizationsIn_Procx2e_Haifa_Verification_Conference2017httpsx3ax2fx2fdoix2eorgx2f10x2e1007x2f978x2d3x2d319x2d70389x2d3x5f2x22x29}\Autobibentry{ Manjeet Dahiya and  Sorav Bansal. Modeling Undefined Behaviour Semantics for Checking Equivalence Across Compiler Optimizations. In \textit{Proc. Haifa Verification Conference}, 2017. \href{https://doi.org/10.1007/978-3-319-70389-3_2}{\Snolinkurl{https://doi.org/10.1007/978-3-319-70389-3_2}}}}

\Autobibtarget{\label{t:x28autobib_x22_Will_Dietzx2c__Peng_Lix2c__John_Regehrx2c_and__Vikram_AdveUnderstanding_Integer_Overflow_in_Cx2fCx2bx2bIn_Procx2e_International_Conference_on_on_Software_Engineering2012httpsx3ax2fx2fdoix2eorgx2f10x2e1109x2fICSEx2e2012x2e6227142x22x29}\Autobibentry{ Will Dietz,  Peng Li,  John Regehr, and  Vikram Adve. Understanding Integer Overflow in C/C++. In \textit{Proc. International Conference on on Software Engineering}, 2012. \href{https://doi.org/10.1109/ICSE.2012.6227142}{\Snolinkurl{https://doi.org/10.1109/ICSE.2012.6227142}}}}

\Autobibtarget{\label{t:x28autobib_x22Chris_Hathhornx2c_Chucky_Ellisonx2c_and_Grigore_Rox15fuDefining_the_Undefinedness_of_CACM_Conference_on_Programming_Language_Design_and_Implementation2015httpsx3ax2fx2fdoix2eorgx2f10x2e1145x2f2737924x2e2737979x22x29}\Autobibentry{Chris Hathhorn, Chucky Ellison, and Grigore Ro\c{s}u. Defining the Undefinedness of C. ACM Conference on Programming Language Design and Implementation, 2015. \href{https://doi.org/10.1145/2737924.2737979}{\Snolinkurl{https://doi.org/10.1145/2737924.2737979}}}}

\Autobibtarget{\label{t:x28autobib_x22International_Organization_for_StandardizationISOx2fIEC_14882x3a2011_Cx2bx2b_Standard2011httpsx3ax2fx2fwwwx2eisox2eorgx2fstandardx2f50372x2ehtmlx22x29}\Autobibentry{International Organization for Standardization. ISO/IEC 14882:2011 C++ Standard. 2011. \href{https://www.iso.org/standard/50372.html}{\Snolinkurl{https://www.iso.org/standard/50372.html}}}}

\Autobibtarget{\label{t:x28autobib_x22Jacques_Henri_Jourdanx2c_Vincent_Laportex2c_Sandrine_Blazyx2c_Xavier_Leroyx2c_and_David_PichardieA_Formallyx2dVerified_C_Static_AnalyzerIn_Procx2e_ACM_Symposium_on_Principles_of_Programming_Languagesx2c_ppx2e_247x2dx2d2592015httpsx3ax2fx2fdoix2eorgx2f10x2e1145x2f2676726x2e2676966x22x29}\Autobibentry{Jacques Henri Jourdan, Vincent Laporte, Sandrine Blazy, Xavier Leroy, and David Pichardie. A Formally{-}Verified C Static Analyzer. In \textit{Proc. ACM Symposium on Principles of Programming Languages}, pp. 247{--}259, 2015. \href{https://doi.org/10.1145/2676726.2676966}{\Snolinkurl{https://doi.org/10.1145/2676726.2676966}}}}

\Autobibtarget{\label{t:x28autobib_x22Ralf_JungUndefined_Behavior_deserves_a_better_reputation2021httpsx3ax2fx2fblogx2esigplanx2eorgx2f2021x2f11x2f18x2fundefinedx2dbehaviorx2ddeservesx2dax2dbetterx2dreputationx2fRetrievedx3a_Julyx2c_2022x2ex22x29}\Autobibentry{Ralf Jung. Undefined Behavior deserves a better reputation. 2021. \href{https://blog.sigplan.org/2021/11/18/undefined-behavior-deserves-a-better-reputation/}{\Snolinkurl{https://blog.sigplan.org/2021/11/18/undefined-behavior-deserves-a-better-reputation/}} Retrieved: July, 2022.}}

\Autobibtarget{\label{t:x28autobib_x22Ralf_Jungx2c_Hoanghai_Dangx2c_Jeehoon_Kangx2c_and_Derek_DreyerStacked_borrowsx3a_an_aliasing_model_for_RustProceedings_of_the_ACM_on_Programming_Languages_x28POPLx29_4x2c_ppx2e_41x3a1x2dx2d41x3a322020httpsx3ax2fx2fdoix2eorgx2f10x2e1145x2f3371109x22x29}\Autobibentry{Ralf Jung, Hoanghai Dang, Jeehoon Kang, and Derek Dreyer. Stacked borrows: an aliasing model for Rust. \textit{Proceedings of the ACM on Programming Languages (POPL)} 4, pp. 41:1{--}41:32, 2020. \href{https://doi.org/10.1145/3371109}{\Snolinkurl{https://doi.org/10.1145/3371109}}}}

\Autobibtarget{\label{t:x28autobib_x22Ralf_Jungx2c_Jacquesx2dHenri_Jourdanx2c_Robbert_Krebbersx2c_and_Derek_DreyerRustBeltx3a_securing_the_foundations_of_the_Rust_programming_languageProceedings_of_the_ACM_on_Programming_Languages_x28POPLx29_2x2c_ppx2e_66x3a1x2dx2d66x3a342018httpsx3ax2fx2fdoix2eorgx2f10x2e1145x2f3158154x22x29}\Autobibentry{Ralf Jung, Jacques{-}Henri Jourdan, Robbert Krebbers, and Derek Dreyer. RustBelt: securing the foundations of the Rust programming language. \textit{Proceedings of the ACM on Programming Languages (POPL)} 2, pp. 66:1{--}66:34, 2018. \href{https://doi.org/10.1145/3158154}{\Snolinkurl{https://doi.org/10.1145/3158154}}}}

\Autobibtarget{\label{t:x28autobib_x22_Richard_Ax2e_KelseyA_Correspondence_between_Continuation_Passing_Style_and_Static_Single_Assignment_FormIn_Procx2e_Papers_from_the_1995_ACM_SIGPLAN_Workshop_on_Intermediate_Representationsx2c_IR_x2795x2c_ppx2e_13x2dx2d221995httpsx3ax2fx2fdoix2eorgx2f10x2e1145x2f202530x2e202532x22x29}\Autobibentry{ Richard A. Kelsey. A Correspondence between Continuation Passing Style and Static Single Assignment Form. In \textit{Proc. Papers from the 1995 ACM SIGPLAN Workshop on Intermediate Representations}, IR {'}95, pp. 13{--}22, 1995. \href{https://doi.org/10.1145/202530.202532}{\Snolinkurl{https://doi.org/10.1145/202530.202532}}}}

\Autobibtarget{\label{t:x28autobib_x22_Thomas_Lengauer_and__Robert_Endre_TarjanA_Fast_Algorithm_for_Finding_Dominators_in_a_FlowgraphACM_Transactions_on_Programming_Languages_and_Systems_1x281x29x2c_ppx2e_121x2dx2d1411979httpsx3ax2fx2fdoix2eorgx2f10x2e1145x2f357062x2e357071x22x29}\Autobibentry{ Thomas Lengauer and  Robert Endre Tarjan. A Fast Algorithm for Finding Dominators in a Flowgraph. \textit{ACM Transactions on Programming Languages and Systems} 1(1), pp. 121{--}141, 1979. \href{https://doi.org/10.1145/357062.357071}{\Snolinkurl{https://doi.org/10.1145/357062.357071}}}}

\Autobibtarget{\label{t:x28autobib_x22Xavier_LeroyFormal_verification_of_a_realistic_compilerCommunications_of_the_ACM_52x2c_ppx2e_7x3a107x2dx2d7x3a1152009httpsx3ax2fx2fdoix2eorgx2f10x2e1145x2f1538788x2e1538814x22x29}\Autobibentry{Xavier Leroy. Formal verification of a realistic compiler. \textit{Communications of the ACM} 52, pp. 7:107{--}7:115, 2009a. \href{https://doi.org/10.1145/1538788.1538814}{\Snolinkurl{https://doi.org/10.1145/1538788.1538814}}}}

\Autobibtarget{\label{t:x28autobib_x22Xavier_LeroyMechanized_Semantics_for_the_Clight_Subset_of_the_C_LanguageJournal_of_Automated_Reasoning_43x2c_ppx2e_263x2dx2d3682009httpsx3ax2fx2fdoix2eorgx2f10x2e1007x2fs10817x2d009x2d9148x2d3x22x29}\Autobibentry{Xavier Leroy. Mechanized Semantics for the Clight Subset of the C Language. \textit{Journal of Automated Reasoning} 43, pp. 263{--}368, 2009b. \href{https://doi.org/10.1007/s10817-009-9148-3}{\Snolinkurl{https://doi.org/10.1007/s10817-009-9148-3}}}}

\Autobibtarget{\label{t:x28autobib_x22_LLVM_ProjectLLVM_13x2e0x2e0_Release_Notes2021httpsx3ax2fx2freleasesx2ellvmx2eorgx2f13x2e0x2e0x2fdocsx2fReleaseNotesx2ehtmlRetrievedx3a_Sepx2c_2021x22x29}\Autobibentry{ LLVM Project. LLVM 13.0.0 Release Notes. 2021. \href{https://releases.llvm.org/13.0.0/docs/ReleaseNotes.html}{\Snolinkurl{https://releases.llvm.org/13.0.0/docs/ReleaseNotes.html}} Retrieved: Sep, 2021}}

\Autobibtarget{\label{t:x28autobib_x22Edward_Sx2e_Lowry_and_Cx2e_Wx2e_MedlockObject_Code_OptimizationCommunications_of_the_ACM_12x281x29x2c_ppx2e_13x2dx2d221969httpsx3ax2fx2fdoix2eorgx2f10x2e1145x2f362835x2e362838x22x29}\Autobibentry{Edward S. Lowry and C. W. Medlock. Object Code Optimization. \textit{Communications of the ACM} 12(1), pp. 13{--}22, 1969. \href{https://doi.org/10.1145/362835.362838}{\Snolinkurl{https://doi.org/10.1145/362835.362838}}}}

\Autobibtarget{\label{t:x28autobib_x22_Melissa_MearsFunction_to_mark_unreachable_code2021httpx3ax2fx2fwwwx2eopenx2dstdx2eorgx2fjtc1x2fsc22x2fwg21x2fdocsx2fpapersx2f2021x2fp0627r6x2epdfx22x29}\Autobibentry{ Melissa Mears. Function to mark unreachable code. 2021. \href{http://www.open-std.org/jtc1/sc22/wg21/docs/papers/2021/p0627r6.pdf}{\Snolinkurl{http://www.open-std.org/jtc1/sc22/wg21/docs/papers/2021/p0627r6.pdf}}}}

\Autobibtarget{\label{t:x28autobib_x22Godon_Dx2e_Plotkinx7bCallx2dbyx2dnamex2c_callx2dbyx2dvalue_and_the_x3bbx2dcalculusTheoretical_Computer_Sciencex281x282x29x29x2c_ppx2e_125x2dx2d1591975x22x29}\Autobibentry{Godon D. Plotkin. \{Call{-}by{-}name, call{-}by{-}value and the $\lambda${-}calculus. \textit{Theoretical Computer Science}(1(2)), pp. 125{--}159, 1975.}}

\Autobibtarget{\label{t:x28autobib_x22_John_RegehrBetter_Testing_With_Undefined_Behavior_Coverage2011httpsx3ax2fx2fblogx2eregehrx2eorgx2farchivesx2f388x22x29}\Autobibentry{ John Regehr. Better Testing With Undefined Behavior Coverage. 2011. \href{https://blog.regehr.org/archives/388}{\Snolinkurl{https://blog.regehr.org/archives/388}}}}

\Autobibtarget{\label{t:x28autobib_x22_Xi_Wangx2c__Nickolai_Zeldovichx2c__Mx2e_Frans_Kaashoekx2c_and__Armando_Solarx2dLezamaA_Differential_Approach_to_Undefined_Behavior_DetectionCommunications_of_the_ACM_59x283x29x2c_ppx2e_99x2dx2d1062016httpsx3ax2fx2fdoix2eorgx2f10x2e1145x2f2885256x22x29}\Autobibentry{ Xi Wang,  Nickolai Zeldovich,  M. Frans Kaashoek, and  Armando Solar{-}Lezama. A Differential Approach to Undefined Behavior Detection. \textit{Communications of the ACM} 59(3), pp. 99{--}106, 2016. \href{https://doi.org/10.1145/2885256}{\Snolinkurl{https://doi.org/10.1145/2885256}}}}

\Autobibtarget{\label{t:x28autobib_x22Lix2dyao_Xiax2c_Yannick_Zakowskix2c_Paul_Hex2c_Chungx2dKilx5cn____________________Hurx2c_Gregory_Malechax2c_Benjamin_Cx2e_Piercex2c_and_Stevex5cn______________________ZdancewicInteraction_Treesx3a_rRepresenting_Recursive_and_Impurex5cn____Programs_in_CoqProceedings_of_the_ACM_on_Programming_Languages_x28POPLx29_4x2c_ppx2e_51x3a1x2dx2d51x3a322020httpsx3ax2fx2fdoix2eorgx2f10x2e1145x2f3371119x22x29}\Autobibentry{Li{-}yao Xia, Yannick Zakowski, Paul He, Chung{-}Kil
                    Hur, Gregory Malecha, Benjamin C. Pierce, and Steve
                      Zdancewic. Interaction Trees: rRepresenting Recursive and Impure
    Programs in Coq. \textit{Proceedings of the ACM on Programming Languages (POPL)} 4, pp. 51:1{--}51:32, 2020. \href{https://doi.org/10.1145/3371119}{\Snolinkurl{https://doi.org/10.1145/3371119}}}}

\Autobibtarget{\label{t:x28autobib_x22_Yannick_Zakowskix2c__Calvin_Beckx2c__Irene_Yoonx2c__Ilia_Zaichukx2c__Vadim_Zalivax2c_and__Steve_ZdancewicModularx2c_Compositionalx2c_and_Executable_Formal_Semantics_for_LLVM_IRProceedings_of_the_ACM_on_Programming_Languages_x28ICFPx29_5x2c_ppx2e_67x3a1x2dx2d67x3a302021httpsx3ax2fx2fdoix2eorgx2f10x2e1145x2f3473572x22x29}\Autobibentry{ Yannick Zakowski,  Calvin Beck,  Irene Yoon,  Ilia Zaichuk,  Vadim Zaliva, and  Steve Zdancewic. Modular, Compositional, and Executable Formal Semantics for LLVM IR. \textit{Proceedings of the ACM on Programming Languages (ICFP)} 5, pp. 67:1{--}67:30, 2021. \href{https://doi.org/10.1145/3473572}{\Snolinkurl{https://doi.org/10.1145/3473572}}}}

\Autobibtarget{\label{t:x28autobib_x22_Jianzhou_Zhaox2c__Santosh_Nagarakattex2c__Milo_Mx2eKx2e_Martinx2c_and__Steve_ZdancewicFormalizing_the_LLVM_Intermediate_Representation_for_Verified_Program_TransformationsIn_Procx2e_ACM_Symposium_on_Principles_of_Programming_Languagesx2c_POPL_x2712x2c_ppx2e_427x2dx2d4402012httpsx3ax2fx2fdoix2eorgx2f10x2e1145x2f2103656x2e2103709x22x29}\Autobibentry{ Jianzhou Zhao,  Santosh Nagarakatte,  Milo M.K. Martin, and  Steve Zdancewic. Formalizing the LLVM Intermediate Representation for Verified Program Transformations. In \textit{Proc. ACM Symposium on Principles of Programming Languages}, POPL {'}12, pp. 427{--}440, 2012. \href{https://doi.org/10.1145/2103656.2103709}{\Snolinkurl{https://doi.org/10.1145/2103656.2103709}}}}

\Autobibtarget{\label{t:x28autobib_x22_Jianzhou_Zhaox2c__Santosh_Nagarakattex2c__Milo_Mx2eKx2e_Martinx2c_and__Steve_ZdancewicFormal_Verification_of_SSAx2dBased_Optimizations_for_LLVMIn_Procx2e_ACM_Conference_on_Programming_Language_Design_and_Implementationx2c_PLDI_x2713x2c_ppx2e_175x2dx2d1862013httpsx3ax2fx2fdoix2eorgx2f10x2e1145x2f2491956x2e2462164x22x29}\Autobibentry{ Jianzhou Zhao,  Santosh Nagarakatte,  Milo M.K. Martin, and  Steve Zdancewic. Formal Verification of SSA{-}Based Optimizations for LLVM. In \textit{Proc. ACM Conference on Programming Language Design and Implementation}, PLDI {'}13, pp. 175{--}186, 2013. \href{https://doi.org/10.1145/2491956.2462164}{\Snolinkurl{https://doi.org/10.1145/2491956.2462164}}}}\end{SingleColumn}\end{AutoBibliography}

\postDoc
\end{document}